\documentclass[sn-nature,pdflatex]{sn-jnl}
\usepackage{amsmath,amssymb,amsfonts,mathrsfs,amsthm}

\usepackage[english]{babel}
\usepackage{xspace}
\usepackage{cases}
\usepackage{graphicx}
\usepackage[]{caption}
\usepackage[]{subcaption}
\usepackage{upgreek}
\usepackage{tikz}
\usetikzlibrary{shapes.geometric, arrows}
\usepackage{makecell}
\usepackage{booktabs}
\usepackage{array}
\usepackage{blindtext}
\usepackage{multirow}
\usepackage{siunitx}
\usepackage{lipsum}
\usepackage{glossaries}
\usepackage{setspace}
\usepackage{arydshln}

\usepackage[
]{hyperref}
\newtheorem{myTheorem}{Theorem}
\newtheorem{myLemma}{Lemma}
\newtheorem{myAss}{Assumption}

\newtheorem{myRemark}{Remark}
\newtheorem{myCorollary}{Corollary}
\newtheorem{myProposition}{Proposition}

\usepackage{threeparttable} % add the footnote in the table

\usepackage{algorithm,algorithmic}
\usepackage{textcomp}

\setkeys{glslink}{hyper=false}

\newacronym{opf}{OPF}{optimal power flow}
\newacronym{qp}{\textsc{qp}}{quadratic program}
\newacronym{nlp}{\textsc{nlp}}{nonlinear programming}
\newacronym{milp}{\textsc{milp}}{Mixed-Integer Linear Programming}
\newacronym{rapidpf}{rapid\textsc{pf}}{rapid prototyping for distributed Power Flow}
\newacronym{admm}{\textsc{admm}}{Alternating Direction Method of Multipliers}
\newacronym{aladin}{\textsc{aladin}}{Augmented Lagrangian based Alternating Direction Inexact Newton method}
\newacronym{ocd}{\textsc{ocd}}{Optimality Condition Decomposition}
\newacronym{app}{\textsc{app}}{Auxiliary Problem Principle}
\newacronym{sqp}{\textsc{sqp}}{Sequential Quadratic Programming}
\definecolor{todo}{rgb}{.2,.8,.1}
\definecolor{change}{rgb}{0.,.8,.4}
\definecolor{revised}{rgb}{.2,.8,.1}

\definecolor{deepgreen}{rgb}{0.0, 0.5, 0.0}
\usepackage{rotating}
\renewcommand*{\proofname}{Proof}

\newacronym{dn}{DN}{distribution network}
\newacronym{af}{AF}{aggregated flexibility}
\newacronym{gpu}{GPU}{graphics processing unit}
\newacronym{itd}{ITD}{integrated transmission-distribution}
\newacronym{qcqp}{QCQP}{quadratically constrained quadratic program}

\newacronym{mpopf}{MPOPF}{multiperiod AC optimal power flow}
\newacronym{pf}{PF}{power flow}
\newacronym{pcc}{PCC}{point of common coupling}
\newacronym{ders}{DERs}{distributed energy resources}
\newacronym{der}{DER}{distributed energy resource}
\newacronym{ess}{ESS}{energy storage systems}
\newacronym{pv}{PV}{photovoltaics}
\newacronym{ev}{EV}{electric vehicles}
\newacronym{tsos}{TSOs}{transmission system operators}
\newacronym{tso}{TSO}{transmission system operator}
\newacronym{dsos}{DSOs}{distribution system operators}
\newacronym{dso}{DSO}{distribution system operator}
\newacronym{soc}{SoC}{State of Charge}
\newacronym{kitcn}{KIT-CN}{Karlsruhe Institute of Technology - Campus Nord}
\newcommand{\norm}[1]{\left\lVert#1\right\rVert}

\newcommand{\abs}[1]{\left|#1\right|}

\newcommand{\parens}[1]{\left(#1\right)}
\newcommand{\squarebrackets}[1]{\left[#1\right]}
\newcommand{\braces}[1]{\left\{#1\right\}}

\newcommand{\matpower}{\textsc{matpower}\xspace}

\newcommand{\ipopt}{\textsc{ipopt}\xspace}

\newcommand{\casadi}{\textsc{c}as\textsc{ad}i\xspace}

\title{Distributed coordination for transmission-distribution systems with nonlinear flexibility aggregation}
\author[1,2]{Xinliang Dai}%\email{xinliang.dai@princeton.edu}
\author[2]{Yanlin Jiang} 
\author[2]{Frederik Zahn} %\ead{frederik.zahn@kit.edu},  % (ead) as shown
\author[3]{Yi Guo} %\ead{yi.guo@ieee.org},  % (ead) as shown
\author[2]{Veit Hagenmeyer}
\affil[1]{\orgdiv{Andlinger Center for Energy and the Environment}, \orgname{Princeton University}, 
\country{USA}}
\affil[2]{\orgdiv{Institute for Automation and Applied Informatics}, \orgname{Karlsruhe Institute of Technology}, %\orgaddress{\street{Street}, 
\country{Germany}}
\affil[3]{\orgname{School of Automation, Beijing Institute of Technology}, %\orgaddress{\street{Street}, 
\country{China}}
\abstract{\normalsize 
High shares of distributed energy resources (DERs) transform distribution systems into active participants in integrated transmission and distribution (ITD) operations. 
{
Linear models enable scalable distribution-level flexibility aggregation but can misclassify AC feasible operating points, whereas direct nonlinear aggregation becomes costly, especially in multiperiod ITD coordination.}
This paper reformulates transmission–distribution coordination within a hierarchical optimization framework and introduces a non-iterative predictor–corrector aggregation method.
By leveraging path-following techniques from real-time optimal control, the approach achieves tractable computation with guaranteed error bounds. 
{Across 24 radial distribution-network cases and seven meshed variants, including the real KIT Campus North grid,  the proposed method yields substantially lower sampled false- and lost-flexibility rates than linear surrogates and a convex relaxation.
On two 24-period ITD testcases, the formulation reduces end-to-end wall-clock time by factors of $6$ relative to the corresponding centralized formulation, primarily through dimensionality reduction.}
An open-source implementation supports transparency and further research.

}
\begin{document}

%\setlength\abovedisplayskip{4pt}
%\setlength\belowdisplayskip{4pt}
%-----------------------
% Begin document
%-----------------------
\maketitle
\section{Introduction}

The increasing penetration of \acrfull{ders}, including rooftop \acrfull{pv}, \acrfull{ess}, and \acrfull{ev} charging, {is transforming distribution networks from passive loads into active participants in system operation, socalled active distribution network~\cite{itd2020review,karthikeyan2019activeDistr,dutta2020topology}. This transition introduces variable and distributed power injections driven by weather conditions and consumer behavior, and motivates closer \acrfull{tso}--\acrfull{dso} coordination, as also reflected in recent European regulations~\cite{kerscher2022key}. A central challenge is to coordinate power exchange at the TSO--DSO interface while accounting for controlling the \acrshort{ders} in the distribution system.}

A straightforward approach to compute set points for the \acrshort{ders} is a fully centralized managed model, in which the \acrshort{tso} collects detailed information from all \acrshort{dso}s, and directly optimizes the dispatch of all controllable \acrshort{der}s across the entire \acrfull{itd} system~\cite{itd2020review}.
In this way, it can effectively manage the entire \acrshort{itd} system as a single optimization problem. However, this approach is not feasible in practice. The joint optimization of multiple nonlinear distribution systems with potentially thousands of individual DERs quickly exceeds available computational capacities. Additionally, centralized management requires detailed consumer data, which creates privacy concerns and places substantial demands on communication infrastructure~\cite{molzahn2017survey,patari2021distributed,dai2025largescale}.

An alternative framework is the decomposition of the problem into two layers.
This approach relies on flexibility aggregation to represent the set of feasible power exchanges at the coupling points between \acrshort{tso} and \acrshort{dsos}, which enables distributed coordination between transmission and distribution systems. 
In this framework, each \acrshort{dso} provides its flexibility set to the \acrshort{tso}. The \acrshort{dso} ensures that any operating point within this set can be securely accommodated in its network, while \acrshort{tso} schedules power exchanges without explicitly modeling distribution-level details. 
{
The aggregation-based interface can be computed locally by each \acrshort{dso} to communicate only compact flexibility sets to the \acrshort{tso} and, thereby, replaces detailed distribution-level models in the upper-level problem. 
It therefore enables parallel preprocessing, reduces privacy and communication burdens, and lowers the dimension and complexity of the transmission-level \acrshort{nlp}.
}

%The approach offers three main advantages: 
%\begin{itemize}
%    \item This aggregating procedure can be performed locally and independently by each \acrshort{dso}, which allows for efficient parallel preprocessing.

%    \item Since the resulting flexibility set condenses detailed distribution-level models into a compact representation that is communicated only once, it greatly reduces privacy concerns and alleviates communication burdens.

%    \item Replacing detailed distribution models with aggregated sets significantly lowers the dimensionality and complexity of the transmission-level nonlinear problems. 
%\end{itemize}

Despite these advantages, this non-iterative flexibility aggregation faces significant challenges. The nonlinear nature of distribution networks complicates the accurate representation of feasible regions. The presence of storage-based DERs further exacerbates this difficulty by introducing temporal coupling across multiple periods, making the problem both mathematically intricate and computationally demanding.

\subsection{Related Work}
\label{sec::relatedwork}
{
Existing flexibility aggregation methods can be classified along two main axes: the power flow model used to represent distribution-network physics and the method used to project or approximate the feasible region in the TSO-DSO coupling space. Table~\ref{tab:flexibility_models} summarizes representative approaches according to these two axes.
}

{
On the modeling side, the most scalable approaches rely on linearized power flow models. These include DC power flow~\cite{wei2015realtime,tan2019enforcinga,tan2024noniterativea}, enhanced DC power flow~\cite{tan2020estimatinga}, LinDistFlow~\cite{wen2023improvedDER,dai2024realtime}, and linearized AC models~\cite{contreras2019timebased,contreras2021congestion}. DC power flow is lossless and was originally developed for transmission systems, where branch reactances dominate resistances, differences between adjacent bus voltage angles are small, and voltage magnitudes are close to 1 p.u.~\cite{rau2003issues,zhu2009optimization,frank2016introduction}. LinDistFlow is a lossless simplification of the DistFlow model and is mainly suited to radial distribution networks~\cite{baran1989optimal1,baran1989lindisflow,farivar2013branch}. Enhanced DC models improve the classical DC approximation by adding voltage magnitude or loss corrections~\cite{yang2017linearized}, while linearized AC models use the Jacobian at an operating point to approximate the relation between injections and state variables. These linearized models enable tractable aggregation and have been extended to multi-period and three-phase settings~\cite{kalantar-neyestanaki2020characterizing,zhang2023coordination,wang2021aggregate,wang2025non,chen2020aggregate,chen2021leveraginga,wen2022tdder}, but their accuracy depends on the validity of the underlying assumptions. }

{
On the aggregation side, several set representation methods have been proposed for the linearized models. Outer approximations provide supersets of the feasible region and may therefore require additional feasibility checks~\cite{wei2015realtime,zhang2023coordination}. Inner approximations provide guaranteed feasible but potentially conservative subsets~\cite{wang2021aggregate,wang2025non,wen2023improvedDER}. Boundary detection methods sample or approximate the flexibility boundary~\cite{tan2019enforcinga,tan2024noniterativea,kalantar-neyestanaki2020characterizing,tan2020estimatinga,chen2020aggregate,chen2021leveraginga,contreras2019timebased,contreras2021congestion}, but their cost can increase rapidly with the dimension of the coupling space. Projection-based methods, such as Fourier-Motzkin elimination, eliminate internal variables to obtain a representation in the coupling-variable space, but can become expensive for large multi-period models~\cite{wen2022tdder}. More recently, implicit-function-based methods have introduced temporal decomposition strategies for LinDistFlow-based aggregation~\cite{dai2024realtime}.
}

{
However, using linearized models can compromise accuracy. The resulting aggregated sets may overestimate the feasible region, leading to infeasible dispatch commands, or underestimate it, leaving valuable flexibility unutilized~\cite{lopez2021quickflex}. Moreover, such modeling mismatches can accumulate over time in multi-period settings, as demonstrated in~\cite{jiang2025enhanced}. To improve accuracy, several works move beyond purely linear models. Direct nonlinear boundary detection has been proposed in~\cite{churkin2023tracing}, but it becomes computationally expensive for large networks. DistFlow-based methods can capture nonlinear effects and provide approximations of flexibility regions~\cite{jiang2023feasible,jiang2025enhanced}, but their reliance on branch flow formulations restricts them mainly to radial networks~\cite{farivar2013branch}. Approximate-dynamic-programming methods approximate value functions of local distribution-level subproblems~\cite{bandeira2024adp,engelmann2025approximate}, but face scalability limitations. Convex inner approximations of the AC power flow model have been developed~\cite{lee2019convex}. Originally intended for uncertainty management, these methods often produce overly conservative flexibility sets, which constrain their practical value in aggregation.
}

\setlength{\tabcolsep}{4pt}% Default value: 6pt
\begin{table*}[ht!]
\centering\vspace{-5pt}
\caption{Summary of Aggregation Methods and Power Flow Models}
\label{tab:flexibility_models}
\fontsize{7}{7}\selectfont
  \setlength\extrarowheight{1pt}
%\resizebox{\textwidth}{!}{%
\begin{tabular}{@{}clccccc@{}}
\toprule

\multirow{2}{*}{
\textbf{Ref}
}
& \multirow{2}{*}{\textbf{System Model}} & \textbf{Network}&\textbf{Time} & \textbf{Temporal} & \multirow{2}{*}{\textbf{Aggregation Method}} & \textbf{Nonlinear AC } \\
& & \textbf{Topology}&\textbf{Period}& \textbf{Decomp.} & & \textbf{Physics}\\
\midrule
\cite{wei2015realtime} & DC power flow & mesh  & single&-  & Outer approximation & -  \\
\cite{tan2019enforcinga,tan2024noniterativea}  & DC power flow & mesh   & single&-  & Boundary detection & -  \\
\cite{kalantar-neyestanaki2020characterizing} & Enhanced DC power flow (3-phase) & mesh & single&-  & Boundary detection                  & -  \\
\cite{tan2020estimatinga} & Enhanced DC power flow & mesh & single&-  & Boundary detection  & Yes \\
\midrule
\cite{zhang2023coordination}      & Enhanced DC power flow   & mesh & multiple& No & Outer approximation & -  \\
\cite{wang2021aggregate,wang2025non}          &  Enhanced DC power flow (3-phase)  & mesh& multiple& No & Inner approximation                 & -  \\
\cite{chen2020aggregate,chen2021leveraginga}          & Enhanced DC power flow (3-phase)  & mesh       & multiple& No & Boundary detection                  & -  \\
%\cite{wen2022aggregate}           & No network / batteries at one bus        & multiple & Fourier–Motzkin elimination         & Unnecessary \\
\cite{wen2022tdder}               & Enhanced DC power flow (3-phase)  & mesh      & multiple &No& Fourier–Motzkin elimination         & -  \\
\cite{contreras2019timebased,contreras2021congestion}    & Linearized AC power flow   & mesh                  & multiple& No & Boundary detection                  & -  \\
\cite{wen2023improvedDER}         & LinDistFlow        & radial                       & multiple &No& Inner approximation & -  \\
\cite{dai2024realtime}            & LinDistFlow            & radial                   & multiple &Yes& Implicit function reformulation & -  \\
\midrule
%\cite{lopez2021quickflex} & Convexified DistFlow  & radial   & single&-  & Boundary detection  & - \\
\cite{churkin2023tracing} & DistFlow        & radial    & single&-  & Boundary detection                  & Yes  \\
\cite{jiang2023feasible}  & DistFlow      & radial   & single&-  & Analytical expressions         & Yes \\
\cite{jiang2025enhanced}&  LinDistFlow & radial & single & - & Quadratic loss compensation & Yes\\
\cite{bandeira2024adp}& LinDistFlow with linearized losses & radial & single & - & Approximate dynamic programming & Yes\\
\cite{engelmann2025approximate}& AC power flow & mesh & single & - & Approximate dynamic programming & Yes\\
\cite{lee2019convex}                &  AC power flow   & mesh    & single&-  & Convex restriction & Yes \\
This paper & AC power flow & mesh & multiple & Yes & Predictor-corrector & Yes \\
\bottomrule
\end{tabular}\\
*All the models using small-angle assumptions to linearly approximate the sin and cos functions are classified as DC power flow.\\
*Enhanced DC power flow includes voltage difference and/or line losses compared with the standard DC model.\vspace{-5pt}
%}

\end{table*} 
%yl{Convex relaxation provides an alternative. For the DistFlow model, the nonconvex branch-flow constraint in the DistFlow model can be replaced by a second-order cone constraint. The resulting formulation is convex when the feasible sets of variables are also convex~\cite{gan2015exact,huang2017sufficient}. The relaxed model remains nonlinear, which requires the boundary-detection aggregation method.}

\subsection{
Challenges and 
Contributions}
{ Existing methods face a trade-off between scalability, AC accuracy, topology generality, and temporal decomposition. Linearized methods are scalable but may lose accuracy; nonlinear methods are more accurate but are often restricted to radial networks, single-period settings, or computationally expensive formulations. Motivated by this gap, we formulate the nonconvex ITD coordination problem within a graph-based hierarchical optimization framework that reflects the layered TSO-DSO structure. }

{Within this framework, we adapt predictor–corrector techniques---originally developed for real-time path-following control problems~\cite{diehl2002real,zavala2010real}---to approximate distribution-level flexibility sets. 
The proposed approach is conceptually related to recent work in generic \acrfull{nlp}~\cite{pacaud2025sensitivity,naik2025variable}, where implicit function reformulations are used to iteratively condense derivatives, thereby improving the reliability of interior-point methods~\cite{parker2022implicit,bugosen2023process} and accelerating computation~\cite{pacaud2022feasible,pacaud2024accelerating}.
In contrast to these iterative solver-level reductions, the present paper uses a non-iterative predictor-corrector construction to precompute local surrogates of distribution-level implicit feasible sets. The resulting aggregation directly accounts for nonlinear AC power flow constraints, is applicable to both radial and meshed distribution networks,  %under the stated regularity assumptions 
and enables scalable pre-computation through spatial and temporal decomposition.
}

The main contributions of this paper are in the following:
\begin{enumerate}
    \item We reformulate the ITD dispatch problem into a hierarchical optimization framework and introduce a predictor–corrector aggregation scheme for distribution networks’ nonconvex implicit feasible sets. We establish error bounds and validate accuracy on tutorial examples and extensive \matpower benchmarks, covering both radial and meshed topologies.
    \item We extend the approach to multi-period operation by constructing a coupled power–energy envelope that preserves temporal consistency while enabling parallel preprocessing across spatial and temporal scales. This reduces dimensionality and achieves $5-7$ times speedups compared to centralized state-of-the-art NLP solvers such as \ipopt~\cite{wachter2006implementation}.
%    \item Utilizing the proposed aggregation method, we quantify how line impedance, DER penetration, and network topology influence the shape of aggregated flexibility regions. These results both validate the robustness of the method and provide practical guidance for planners and operators seeking to unlock distribution system flexibility under high DER integration.    
    %\item We provide an open-source toolbox with benchmark cases to promote accessibility, foster comparative evaluation, and support future research on distribution system flexibility aggregation.
\end{enumerate}

After notations\footnote{\textit{Notations:} 
Let $\mathbb{R}^{n}$ denote the $n$-dimensional real Euclidean space, and 
$\mathbb{C}^{n}$ the $n$-dimensional complex Euclidean space. For
two column vectors, $x\in\mathbb{R}^{m},\,y\in\mathbb{R}^{n}$, their
concatenation is denoted by $(x,y):=[x^{\top},y^{\top}]^{\top}$. For a set $\mathcal{S}$
with cardinality $\abs{\mathcal{S}}$, the concatenation of subvectors
$x_{i}$ across all $i\in\mathcal{S}$ is denoted by
$
    \braces{x_i}_{i\in\mathcal{S}}:= (x_{1},\cdots,x_{\abs{\mathcal{S}}}).
$
{
Distance from a point $x$ to a set
$\mathcal S\subseteq\mathbb R^{n}$ is denoted by
$
\mathrm{dist}(x,\mathcal S):=\inf_{s\in\mathcal S}\norm{x-s}.
$
}
Following~\cite{johnson1990matrix}, the Hadamard (element-wise) product of vectors $x,y\in\mathbb{R}^{m}$ is
denoted by $x\circ y$, with elements given by
$
    [x\circ y]_{i}:= [x]_{i}\cdot[y]_{i}.
$
Element-wise squaring of a vector $x\in\mathbb{R}^{m}$ is denoted by
$
    x^{{\circ}2}:=x\circ x.
$
}, {Section~\ref{sec::formulation} presents the problem formulation for ITD coordination, including the graph representation of the network, the AC power flow models, the TSO--DSO management models, and the corresponding hierarchical optimization framework. Section~\ref{sec::predictor-corrector} introduces the proposed predictor-corrector method for nonlinear flexibility aggregation and presents the corresponding single-period benchmarks. Section~\ref{sec::multiperiod} extends the method to multi-period scenarios. Finally, Section~\ref{sec::conclusion} concludes the paper.}

\section{Problem Formulation for ITD Coordination}\label{sec::formulation}
{This section introduces the modeling ingredients underlying the coordination of \acrfull{itd} systems. We first describe the graph representation of the network, then present the AC power flow models of the transmission and distribution subsystems, and discuss the corresponding TSO-DSO management models. Based on these ingredients, we formulate the ITD coordination problem as a graph-based hierarchical optimization problem and introduce the associated implicit feasible sets that motivate the aggregation method developed in the next section.}

%In this section, we reformulated the coordination problems of \acrfull{itd} systems as a graph-based hierarchical optimizaiton problems to facilitate the conceptualization and implementation of decomposition and approximation schemes.  To further reduce local variables and nonlinear constraints of the distribution model, we introduce a predictor-corrector aggregation methods to approximate the implicit feasible set in primal space, to reduce local variables and nonlinear constraints of the distribution model. We establish error bounds, 
%In this section, we will start from the problem formulation of standard AC \acrfull{opf} and multiperiod problems for \acrfull{itd} systems and introduce equivalent two-layer optimization considering the aggregated flexibility from the lower-layer distribution systems. Note that this section only covers the spatial decomposition according to the system operator, while the extension to multiperiod AC \acrshort{opf} problems and relative temporal decomposition will be introduced in Section~\ref{sec::multiperiod}.

%Given a set $\mathcal{S}$, its cardinality is denoted $\abs{\mathcal{S}}$.
%The Euclidean norm is denoted by $\norm{x}$, and the $\infty$ norm is denoted by $\norm{x}_\infty$.

\begin{figure*}[htbp!]
    \centering
    \includegraphics[width=0.90\linewidth]{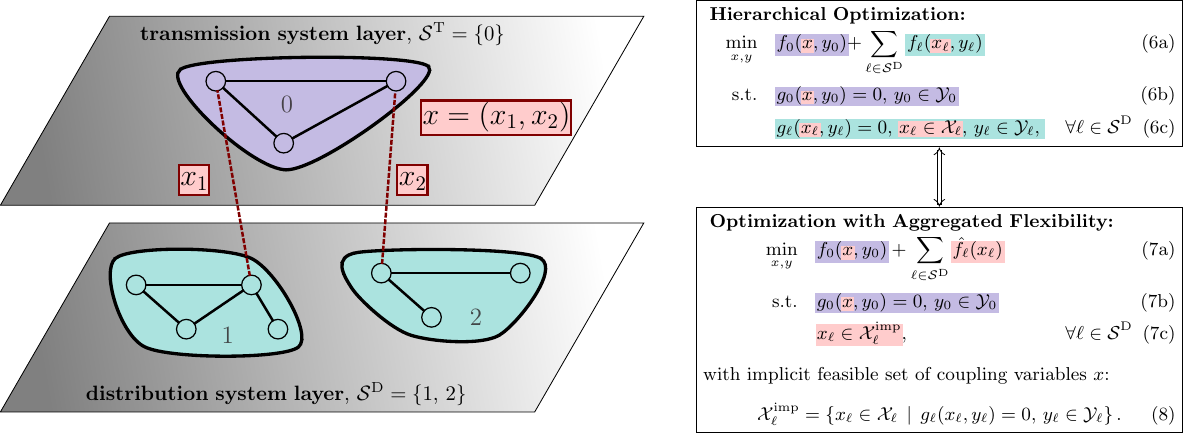}
    \caption{Hierarchical optimization framework for coordinating \acrfull{itd} systems. The transmission system (purple) and distribution systems (green) are modeled as layered subsystems, with coupling variables (red) linking the two layers. The right panel illustrates the corresponding mathematical reformulation: the original hierarchical optimization problem~\eqref{eq::opt::original} is reduced to an aggregated optimization problem~\eqref{eq::opt::flex}, where distribution-level feasibility is captured by the implicit feasible set~\eqref{eq::implicit::set}.}
    \label{fig::overview}
\end{figure*}
\subsection{Graph Representation of Power Networks}
\label{sec::formulation::graph}

A power network can be represented as an undirected graph \(\mathcal G=(\mathcal N,\mathcal L)\), where \(\mathcal N\) denotes the set of buses and \(\mathcal L\) denotes the set of branches.

{To represent the subsystem structure of the \acrshort{itd} network, let
\(\mathcal S^{\mathrm T}=\{0\}\) and \(\mathcal S^{\mathrm D}\) denote the
index sets of the transmission and distribution subsystems, respectively, and
define
\[
    \mathcal S
    :=
    \mathcal S^{\mathrm T}\cup\mathcal S^{\mathrm D}.
\]
The overall ITD network is represented by the graph
\[
    \mathcal G=(\mathcal N,\mathcal L),
\]
together with the subsystem node sets
\(\{\mathcal N_\ell\}_{\ell\in\mathcal S}\), which satisfy
\[
    \mathcal N=\bigcup_{\ell\in\mathcal S}\mathcal N_\ell,
    \qquad
    \mathcal N_\ell\cap\mathcal N_m=\varnothing
    \quad\text{for all }\ell\neq m.
\]
For each subsystem \(\ell\in\mathcal S\), define
\[
    \mathcal L_\ell
    :=
    \left\{
        (i,j)\in\mathcal L
        \,\middle|\,
        i,j\in\mathcal N_\ell
    \right\},
    \qquad
    \mathcal G_\ell=(\mathcal N_\ell,\mathcal L_\ell),
\]
where \(\mathcal N_\ell\) and \(\mathcal L_\ell\) contain the buses and
branches internal to subsystem \(\ell\). The coupling-edge set is
\[
    \mathcal L_s
    :=
    \left\{
        (i,j)\in\mathcal L
        \,\middle|\,
        \exists\,\ell,m\in\mathcal S,\ 
        \ell\neq m,\ 
        i\in\mathcal N_\ell,\ 
        j\in\mathcal N_m
    \right\}.
\]
Hence,
\[
    \mathcal L
    =
    \left(
        \bigcup_{\ell\in\mathcal S}\mathcal L_\ell
    \right)
    \cup
    \mathcal L_s.
\]
The edges in \(\mathcal L_s\) represent coupling between subsystems. This partition is used to distinguish subsystem-internal
variables and constraints from the PCC/interface quantities through which the
subsystems are coordinated.}

{
\begin{myRemark}[Extension to multiple TSOs]
The assumption \(|\mathcal S^{\mathrm T}|=1\) is adopted for notational simplicity and is not essential to the proposed aggregation method. In a more general setting, multiple interconnected transmission systems can be represented by \(|\mathcal S^{\mathrm T}|\ge 1\), with additional tie-lines between transmission subsystems. The predictor-corrector aggregation remains local to each distribution subsystem and can be applied independently. The main additional challenge is the upper-layer coordination among TSOs after aggregation, which results in a distributed nonconvex \acrfull{nlp}. Such transmission-level coordination can be handled by distributed optimization methods, e.g., a distributed optimization framework for large-scale AC OPF based on network decomposition and condensed coordination~\cite{dai2025largescale}.
\end{myRemark}
}
{Based on this partition, the AC power flow model of each subsystem
\(\ell\in\mathcal S\) contains the variables and constraints associated with
its internal buses and branches
\((\mathcal N_\ell,\mathcal L_\ell)\), while the quantities associated with
the interconnections in \(\mathcal L_s\) define the interface variables used
in the coordination model below.
}

%$\mathcal{N}_{\ell}$ 

%and
%$\mathcal{L}_{\ell}$ denote the bus and branch subset in the subgraph $\mathcal{G}_{\ell}$. Additionally,
%$\mathcal{L}_{s}$ denotes the connection between the transmission and
%the distribution subsystems, typically located at distribution substations corresponding
%to the local slack bus $i^{\text{slack}}_{\ell}\in\mathcal{N}_{\ell}$.
%The overall ITD system is thus composed of a transmission subgraph $\mathcal{G}_0$ and the collection of distribution subgraphs $\{\mathcal{G}_\ell\}_{\ell\in\mathcal{S}^\text{D}}$, interconnected through coupling lines $\mathcal{L}_s$.

%\begin{myRemark}
%    [\textbf{Slack Bus}] We assume that in each subsystem   $\ell\in\mathcal{S}$, the slack bus is the first bus in the bus subset $\mathcal{N}_{\ell}$, and its voltage angle is used as the reference for that subsystem. For a distribution system $\ell\in\mathcal{S}^{\text{D}}$, we also assume the slack bus voltage magnitude is fixed by the substation and does not depend on the transmission system voltage.
%\end{myRemark}

\subsection{AC Power Flow Models}

Let $Y_{\ell}\in\mathbb{C}^{\abs{\mathcal{N}_\ell}\times \abs{\mathcal{N}_\ell}}$ denote the complex bus
admittance matrix of subgraph $\mathcal{G}_{\ell}$, with real and imaginary parts $G_\ell, B_\ell\in\mathbb{R}^{\abs{\mathcal{N_\ell}}\times \abs{\mathcal{N_\ell}}}$.
For each bus $i\in\mathcal{N}_\ell$, the complex voltage 
$V_{i}\in\mathbb{C}$ is expressed in rectangular coordinates as
$V_{i}= u_{i}+\textbf{j}w_{i}$, with $u_i,w_i\in\mathbb{R}$.
The active and reactive power generated at bus $i$ are denoted by 
$p^{g}_{i},q^{g}_{i}\in\mathbb{R}$, and the corresponding demand 
by $p^{d}_{i}, q^{d}_{i}\in\mathbb{R}$. 
For each distribution subsystem $\ell\in\mathcal{S}^{\text{D}}$, 
let $p^{\text{pcc}}_{\ell},q^{\text{pcc}}_{\ell}\in\mathbb{R}$ denote the active and reactive power exchanged with the transmission system through the \acrshort{pcc}. 

For each subsystem $\ell\in\mathcal{S}$, the nodal power
injection from their neighboring buses can be expressed as
\begin{subequations}\label{eq::nodalPowerInj}
\begin{align}
     & p_{\ell}^{\text{inj}}(u_{\ell},w_{\ell})= (G_{\ell}\, w_{\ell}+ B_{\ell}\, u_{\ell})\circ w_{\ell}+ (G_{\ell}\, u_{\ell}- B_{\ell}\, w_{\ell})\circ u_{\ell} \\
     & q_{\ell}^{\text{inj}}(u_{\ell},w_{\ell})= (G_{\ell}\, u_{\ell}- B_{\ell}\, w_{\ell})\circ w_{\ell}- (G_{\ell}\, w_{\ell}+ B_{\ell}\,u_{\ell})\circ u_{\ell}
\end{align}
\end{subequations}
%\end{subequations}
where nonlinear mappings
$p_{\ell}^{\text{inj}},\,q_{\ell}^{\text{inj}}:\mathbb{R}^{\abs{\mathcal{N}_\ell}}
\times \mathbb{R}^{\abs{\mathcal{N}_\ell}}\rightarrow \mathbb{R}^{\abs{\mathcal{N}_\ell}}$
collect the nodal injections across all buses in $\mathcal{N}_{\ell}$, and $u_{\ell},w_{\ell}$ are the real and imaginary voltage components, i.e.,
$u_{\ell}=\{u_i\}_{i\in\mathcal{N}_\ell}$ and $w_{\ell}=\{w_i\}_{i\in\mathcal{N}_\ell}$.

For a distribution system $\ell\in\mathcal{S}^{\text{D}}$, the AC power flow model is given by
\begin{subequations}
    \label{eq::pf::distribution}
    \begin{align}
        %v_{\ell,\text{slack}}
        1=\;             & e_{1}^{\top}u_{\ell},\label{eq::pf::distribution::uslack}                                                                                                     \\
        0=\;                                 & e_{1}^{\top}w_{\ell},\label{eq::pf::distribution::vslack}                                                                                                     \\
        0=\;                                 & p_{\ell}^{\text{inj}}(u_{\ell},w_{\ell}) - C^{g}_{\ell}\, p_{\ell}^{g}+ p_{\ell}^{d}- e_{1}\, p^{\text{pcc}}_\ell,\label{eq::pf::distribution::active}             \\
        0=\;                                 & q_{\ell}^{\text{inj}}(u_{\ell},w_{\ell}) - C^{g}_{\ell}\, q_{\ell}^{g}+ q_{\ell}^{d}- e_{1}\, q^{\text{pcc}}_\ell,\label{eq::pf::distribution::reactive}           \\
        v_\ell =\; & u_{\ell}^{{\circ}2}+ w_{\ell}^{{\circ}2},\label{eq::pf::distribution::vlimit}                                                \\
        \underline{v}_{\ell}\leq\;   & v_{\ell}\leq\,\overline{v}_{\ell},\quad \underline{p}^{g}_{\ell}\leq\; p^{g}_{\ell}\leq\,\overline{p}^{g}_{\ell},\quad \underline{q}^{g}_{\ell}\leq\, q^{g}_{\ell}\leq\,\overline{q}^{g}_{\ell}\label{eq::pf::distribution::slimit}
    \end{align}
\end{subequations}
with $p^{d}_{\ell}=\braces{p^d_i}_{i\in\mathcal{N}_\ell}$ and $q^{d}_{\ell}=\braces
{q^d_i}_{i\in\mathcal{N}_\ell}.$ {Constraints~\eqref{eq::pf::distribution::uslack}-\eqref{eq::pf::distribution::vslack} fix the voltage at the PCC, where $e_{1}^{\top}=[1,0,0,\ldots]$~\cite{brandle2026flexibility,chen2020aggregate}. In the ITD model, the PCC represents the secondary side of the interface transformer. We assume that its voltage is regulated at a prescribed operating condition and, following the convention adopted in \cite{bandeira2024adp,fruh2023coordinated}, characterize the distribution-network flexibility in the P-Q domain under a fixed PCC voltage.}

Constraints~\eqref{eq::pf::distribution::active}-\eqref{eq::pf::distribution::reactive}
denote the nodal power balance with the active and reactive power generation
$p_{\ell}^{g},\,q_{\ell}^{g}\in\mathbb{R}^{n_\ell^g}$ for all $n_{\ell}^{g}$
controllable \acrshort{ders}, as well as the power injection from transmission
system $p^{\text{pcc}}_{\ell},q^{\text{pcc}}_{\ell}\in\mathbb{R}$. Here
$C^{g}_{\ell}\in\mathbb{R}^{\abs{\mathcal{N}_\ell}\times n_\ell^g}$ denote
the connectivity matrix to these controllable \acrshort{ders}. Constraints~\eqref{eq::pf::distribution::vlimit}-\eqref{eq::pf::distribution::slimit}
limit voltage magnitude, active power generation, and reactive generation, where
 overline and underline represent the corresponding upper and lower
bounds respectively.

Analogously, the transmission system $\ell=0$ is modeled as
\begin{subequations}
    \label{eq::pf::transmission}
    \begin{align}
        %v_{0,\text{slack}}
        1 =\;             & e_{1}^{\top}u_{0},\label{eq::pf::transmission::uslack}                                                                                         \\
        0=\;                              & e_{1}^{\top}w_{0},\label{eq::pf::transmission::vslack}                                                                                         \\
        0=\;                              & p_{0}^{\text{inj}}(u_{0},w_{0}) - C^{g}_{0}\, p_{0}^{g}+ p_{0}^{d}+ C^{\text{pcc}}_{0}p^{\text{pcc}},\label{eq::pf::transmission::active}      \\
        0=\;                              & q_{0}^{\text{inj}}(u_{0},w_{0}) - C^{g}_{0}\, q_{0}^{g}+ q_{0}^{d}+ C^{\text{pcc}}_{0}q^{\text{pcc}},\label{eq::pf::transmission::reactive}    \\
        v_0 =\; & u_{0}^{{\circ}2}+ w_{0}^{{\circ}2},\label{eq::pf::transmission::vlimit}                                          \\
        \underline{v}_{0}\leq\;       & v_{0}\leq\,\overline{v}_{0},\quad\underline{p}^{g}_{0}\leq\;  p^{g}_{0}\leq\,\overline{p}^{g}_{0},\quad \underline{q}^{g}_{0}\leq\, q^{g}_{0}\leq\,\overline{q}^{g}_{0}.\label{eq::pf::transmission::slimit}
    \end{align}
\end{subequations}
{The voltage at the slack bus is fixed, following the convention in the AC power flow problem.} The key distinction between~\eqref{eq::pf::distribution} and~\eqref{eq::pf::transmission} lies in the active
and reactive power exchanges via \acrshort{pcc} between them. While $p^{\text{pcc}}
_{\ell},q^{\text{pcc}}_{\ell}$ denote the power injection from the transmission
system at the substation (the first bus in the subgraph $\mathcal{G}_{\ell}$),
$p^{\text{pcc}},q^{\text{pcc}}\in\mathbb{R}^{\abs{\mathcal{S}^\text{D}}}$
stack the power transferred to distribution systems, i.e.,
$
    p^{\text{pcc}}=\braces{p^\text{pcc}_\ell}_{\ell\in\mathcal{S}^\text{D}}\text{
    and }q^{\text{pcc}}=\braces{q^\text{pcc}_\ell}_{\ell\in\mathcal{S}^\text{D}}
    .
$
Moreover, the connectivity matrix
$C_{0}^{\text{pcc}}\in\mathbb{R}^{\abs{\mathcal{N}_0}\times \abs{\mathcal{S}^\text{D}}}$
in the constraints~\eqref{eq::pf::transmission::active}-\eqref{eq::pf::transmission::reactive}
specifies the transmission buses to which distribution systems are attached.

These detailed subsystem models in~\eqref{eq::pf::distribution} and~\eqref{eq::pf::transmission} constitute the foundation of the hierarchical optimization framework~\eqref{eq::opt::original}. In particular, the transmission model~\eqref{eq::pf::transmission} corresponds to the upper-layer constraints~\eqref{eq::opt::original::transmission}, while the distribution models~\eqref{eq::pf::distribution} represent the lower-layer constraints~\eqref{eq::opt::original::distribution}.

\subsection{TSO-DSO Management Models}\label{sec::formulation::managmentModels}

{How these subsystem models are embedded into the ITD coordination problem depends on the adopted TSO-DSO management model, which determines the structure of the coupling variables between transmission and distribution.} 

The choice of whether the TSO can directly dispatch DERs shapes the structure of the coupling variables linking transmission and distribution.
In particular, the choice of operational paradigm, such as whether the \acrshort{tso} is allowed to issue dispatch commands to each distribution-level \acrshort{der}, determines the structure of the coupling variables $x_{\ell}$ that link the transmission system with each distribution subsystem $\ell$. 

{These management models determine the definition of the coupling variables \(x_\ell\) and local variables \(y_\ell\), and thereby specify how each subsystem enters the hierarchical optimization problem introduced next.}

\begin{figure}[htbp!]
    \centering
    \begin{subfigure}{0.3\textwidth}
        \centering
        \caption{TSO-managed model}
        \includegraphics[width=\textwidth]{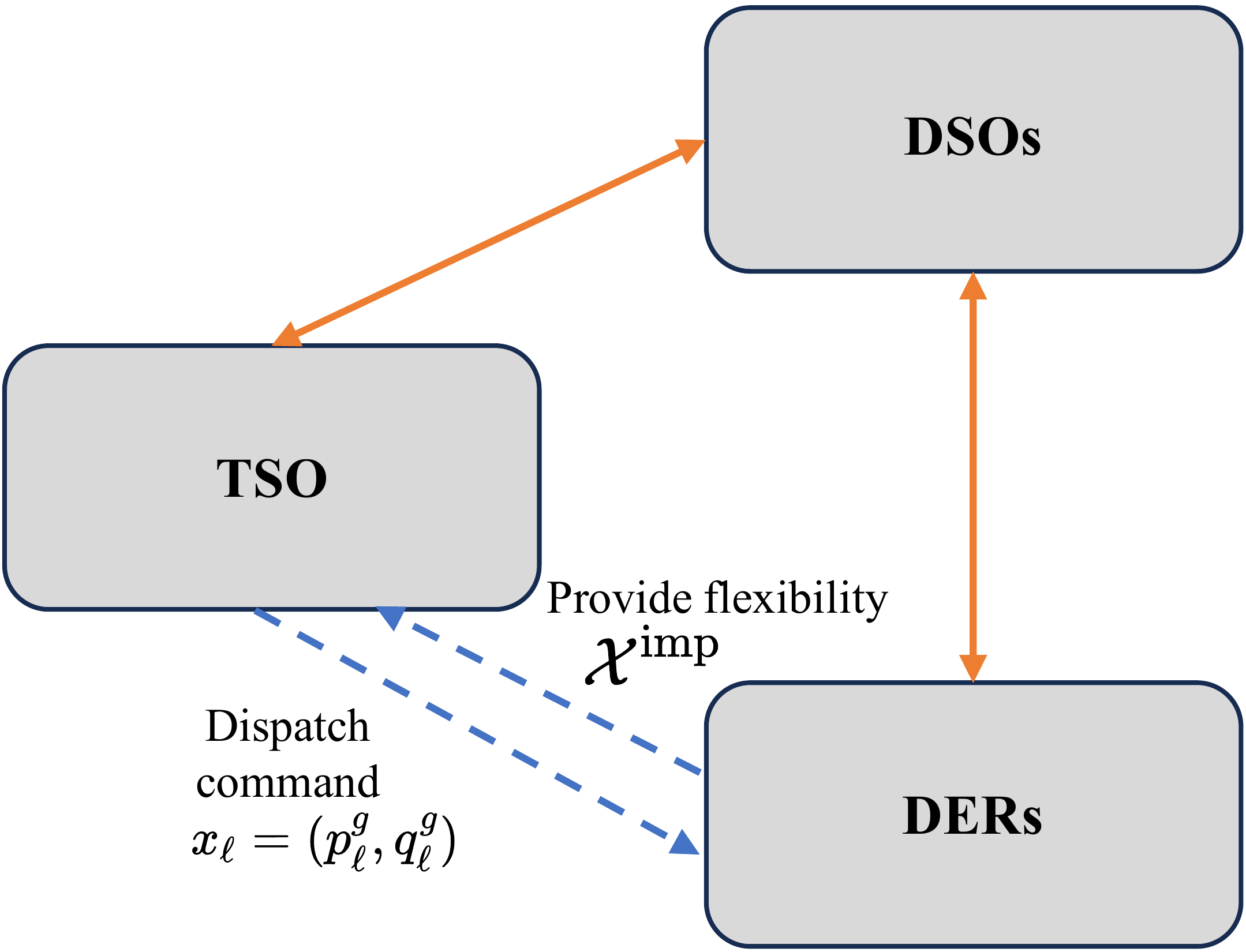}
        \label{fig::management::TSO}
    \end{subfigure}
    \hspace{60pt}
    \begin{subfigure}{0.3\textwidth}
        \centering
        \caption{DSO-managed model}
        \includegraphics[width=\textwidth]{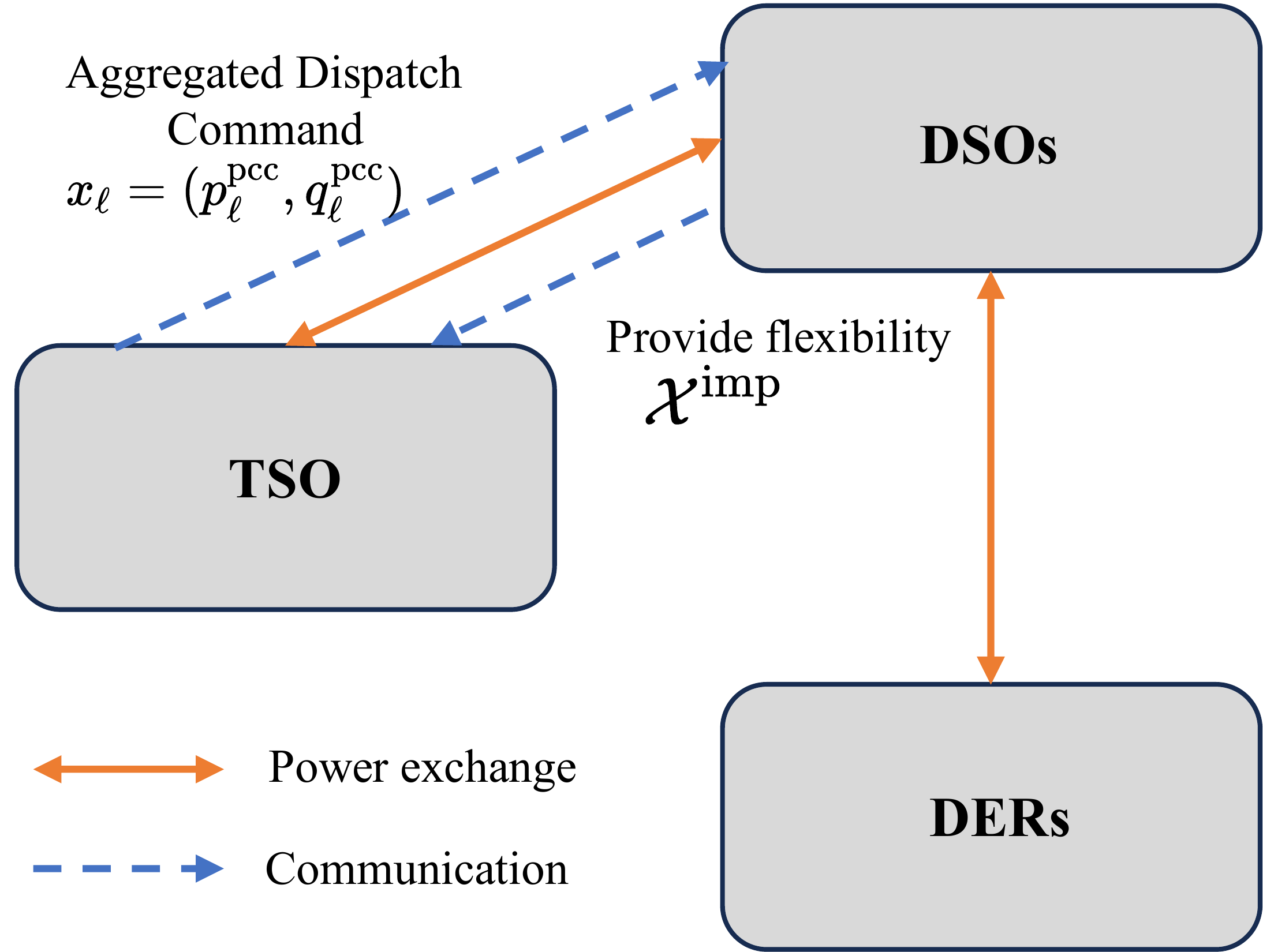}
        \label{fig::management::DSO}
    \end{subfigure}
    \caption{Both TSO–DSO management models keep detailed distribution-level data hidden from the TSO but differ mainly in communication and coordination.  
In the TSO-managed model, each distribution system $\ell$ provides $\mathcal{X}^\mathrm{imp}$ with respect to DER outputs $(p^g_\ell,q^g_\ell)$, enabling the TSO to dispatch DERs directly.  
In the DSO-managed model, each DSO instead supplies $\mathcal{X}^\mathrm{imp}$ with respect to PCC exchanges $(p^\mathrm{pcc}_\ell,q^\mathrm{pcc}_\ell)$ and disaggregates individual DERs based on the scheduled PCC exchanges.  
    }
    \label{fig::management::models}
\end{figure}

\subsubsection{TSO-managed model}

In a TSO-managed model (Fig.~\ref{fig::management::TSO}), the transmission operator directly controls the DERs within each distribution subsystem. The coupling variables are given by 
$$x_{\ell}=(p^g_{\ell},q^g_{\ell})\in\mathcal{X}_{\ell}\subseteq\mathbb{R}^{2n^g_\ell},$$
i.e., the active and reactive power outputs of all DERs. The remaining state variables 
$$y_{\ell}=(u_\ell,w_\ell,v_\ell,p^\text{pcc}_\ell,q^\text{pcc}_\ell)\in\mathcal{Y}_{\ell}\subseteq\mathbb{R}^{3\abs{\mathcal{N}_\ell}+2}$$ are determined internally through the distribution model~\eqref{eq::pf::distribution}, which can be expressed as
%\begin{subequations}
    \begin{equation}\label{eq::distribution::model::general}
        g_{\ell}(x_{\ell},y_{\ell})=0,\; x_{\ell}\in\mathcal{X}_{\ell},\; y_{\ell}\in\mathcal{Y}_{\ell},
    \end{equation}
where nonlinear function $g_{\ell}: \mathbb{R}^{2n^g_\ell}\times\mathbb{R}^{3\abs{\mathcal{N}_\ell}+2}
\rightarrow\mathbb{R}^{3\abs{\mathcal{N}_\ell}+2}$ summarizes the equality constraints~\eqref{eq::pf::distribution::uslack}-\eqref{eq::pf::distribution::vlimit}.
In this management model, the PCC exchanges $(p^\text{pcc}_\ell,q^\text{pcc}_\ell)$ are dependent variables: once $x_\ell$ is fixed, the power exchanged with the transmission system is uniquely determined. The proposed aggregation method thus directly yields an approximate implicit feasible set~\eqref{eq::approximate::implicitSet}.

\subsubsection{DSO-managed model}\label{sec::DSO::management}
By contrast, in a DSO-managed model (Fig.~\ref{fig::management::DSO}), the TSO interacts with each distribution subsystem only through the PCC, without direct control of internal DERs. The coupling variables reduce to the aggregated exchanges, 
$$x_{\ell}= (p^{\text{pcc}}_{\ell},q^{\text{pcc}}_{\ell})\in\mathcal{X}\subseteq\mathbb{R}^2.$$ 
A difficulty arises when multiple DERs are present, since the system then has more internal degrees of freedom. %, violating Assumption~\ref{ass::invertible} of equal dimensions $n=l$. 

%\begin{myRemark}
%    Directly extending the aggregation methods to the DSO-managed model is not possible if multiple controllable \acrshort{ders} are deployed within the same distribution system. Introducing additional DERs increases the system’s degrees of freedom, violating Assumption~\ref{ass::invertible} that guarantees that the Jacobian matrix $\frac{\partial g_\ell}{\partial y_\ell}$ is square and invertible.
%\end{myRemark}
%Note that in
To address this, we follow the strategy of adjusting system imbalance as presented in~\cite{lee2021robust}. In the DSO-managed model, we introduce the \textit{distributed slack} where each \acrshort{der} adjusts its power output to account for the system-wide power flexibility. For DER $i$ in subsystem $\ell$,
\begin{subequations}\label{eq::participation}
    \begin{align}
        p^g_{\ell,i} =\,& p^g_{\ell,i, \text{ref}} + \alpha^p_{i} \Delta p^g_{\ell},\\
        q^g_{\ell,i} =\,& q^g_{\ell,i, \text{ref}} + \alpha^q_i \Delta q^g_{\ell}
    \end{align}
\end{subequations}
with $\sum_i \alpha^p_i=1$ and $\sum_i \alpha^q_i=1$. Here, $\Delta p^g_\ell,\Delta q^g_\ell$ represent system-wide adjustments, and $p^g_{\ell,i, \text{ref}},\,q^g_{\ell,i, \text{ref}}$ are the nominal setpoint for the \acrshort{der} output. The participation factors, determined by the DSO, e.g., for minimizing system losses, remain constant during optimization, while the system-wide power adjustment can be implicitly governed by the coupling variables $x_\ell$. As a result, the distribution model~\eqref{eq::pf::distribution} can also be expressed as~\eqref{eq::distribution::model::general} with $g_{\ell}: \mathbb{R}^{2}\times\mathbb{R}^{3\abs{\mathcal{N}_\ell}+2}
\rightarrow\mathbb{R}^{3\abs{\mathcal{N}_\ell}+2}$ and
\begin{align*}
y_{\ell}= \parens
{u_\ell,w_\ell,v_\ell,\Delta p^g_{\ell},\Delta q^g_{\ell}}\in\mathcal{Y}\subseteq \mathbb{R}^{3\abs{\mathcal{N}_\ell}+2}.
\end{align*}
An alternative to participation factors is using the pseudo-inverse of $\frac{\partial g_\ell}{\partial y_\ell}$  to resolve invertibility challenges.

\begin{myRemark}[DER–PCC mapping]\label{rmk::bijection}
The fixed participation policy removes internal dispatch degrees of freedom by parameterizing individual DER outputs through aggregate adjustments. Under regularity assumptions, this parameterization yields a locally single-valued mapping from PCC exchanges to the corresponding policy-restricted internal state.
\end{myRemark}

{
%\begin{myRemark}[Extension of coupling variables]
The two-dimensional P–Q flexibility set considered in this paper is conditional on a prescribed PCC voltage magnitude. When the dependence of distribution-level flexibility on the interface voltage must be represented explicitly, the coupling variables may be extended to include the PCC active-power exchange, reactive-power exchange, and voltage magnitude:
\[
x_{\ell}
=
\left(
p_{\ell}^{\mathrm{pcc}},
q_{\ell}^{\mathrm{pcc}},
v_{\ell}^{\mathrm{pcc}}
\right),
\]
where yields a \(P\)-\(Q\)-\(V\) flexibility aggregation representation and allows the transmission-level problem to account explicitly for the dependence of distribution-network feasibility on the PCC voltage magnitude. The predictor-corrector aggregation remains applicable to the resulting smooth implicit formulation under the regularity conditions. The PCC voltage magnitude may be treated as a continuous coupling variable without explicitly modeling transformer tap positions. If tap positions are optimized directly, additional discrete variables are required. One possible alternative is to construct a separate continuous flexibility representation for each admissible tap position.
%\end{myRemark}
}

{With the subsystem models and coupling variable definitions now specified, we can formulate the overall ITD coordination problem as a graph-based hierarchical optimization problem.}

\subsection{Graph–Based Hierarchical Optimization}
\label{sec::formulation::aggregation}

The coordination of \acrshort{itd} systems can be expressed as a hierarchical optimization problem:
\begin{subequations}
    \label{eq::opt::original}
    \begin{align}
         \min_{x,y}  & \quad f_{0}(x,y_{0})+\sum_{\ell\in\mathcal{S}^\text{D}}f_{\ell}(x_{\ell},y_{\ell})                   &  &   \\
         \text{s.t.} & \quad g_{0}(x,y_{0})=0, \,\hfill y_{0}\in\mathcal{Y}_{0},                                             &  & \label{eq::opt::original::transmission}\\
                     & \quad g_{\ell}(x_{\ell},y_{\ell})=0,\,x_{\ell}\in\mathcal{X}_{\ell},\,y_{\ell}\in\mathcal{Y}_{\ell}, &  & \forall\ell\in\mathcal{S}^{\text{D}},\label{eq::opt::original::distribution} %       && &\quad g_\ell(x_\ell,y_\ell)=0,\,(x_\ell,y_\ell)\in\mathcal{X}_\ell\times\mathcal{Y}_\ell, &&\forall\ell\in\mathcal{S}^\text{D}
    \end{align}
\end{subequations}
with $x=\{x_\ell\}_{\ell\in\mathcal{S}^\text{D}}$, where $f_{\ell}$ and $g_{\ell}$ denote the objective functions and
constraints of the subsystem
$\ell\in \mathcal{S}= \mathcal{S}^{\text{T}}\bigcup\mathcal{S}^{\text{D}}$. 
The vector $y_{\ell}$ collects the variables that are totally local to the subsystem
$\ell\in\mathcal{S}$, while the vector $x_{\ell}$ represents the coupling
variables that link the transmission system $0$ with distribution subsystems $\ell\in\mathcal{S}^{\text{D}}$, with feasible sets $\mathcal{X}_{\ell}$ and $\mathcal{Y}_{\ell}$, respectively. Variable partitioning depends on the adopted control strategy discussed in Section~\ref{sec::formulation::managmentModels}. 
%Section~\ref{sec::formulation::managmentModels} will discuss how this partitioning of variables depends on the adopted control strategy. 

As illustrated in Fig.~\ref{fig::overview}, constraints~\eqref{eq::opt::original::transmission} and~\eqref{eq::opt::original::distribution} describe the transmission and distribution system models, respectively. The essence of flexibility aggregation is to replace detailed distribution models~\eqref{eq::opt::original::distribution} and also the local variables $y_\ell$ with implicit feasible sets in the coupling space. Then, the optimization problem \eqref{eq::opt::original} can be rewritten as
\begin{subequations}
    \label{eq::opt::flex}
    \begin{align}
          \min_{x,y_0} & \quad f_{0}(x,y_0)+\sum_{\ell\in\mathcal{S}^\text{D}}\hat{f}_{\ell}(x_{\ell}) &  &                                      \\
          \text{s.t.}  & \quad g_{0}(x,y_{0})=0,\,y_{0}\in\mathcal{Y}_{0},                   &  &                                      \\
                       & \quad x_{\ell}\in\mathcal{X}_{\ell}^{\text{imp}},                  &  & \forall\ell\in\mathcal{S}^{\text{D}}
    \end{align}
\end{subequations}
with implicit feasible sets for all $\ell\in\mathcal{S}^\text{D}$
\begin{equation}\label{eq::implicit::set}
    \mathcal{X}_{\ell}^{\text{imp}}= \left\{x_{\ell}\in\mathcal{X}_{\ell}\mid\exists\, y_{\ell}\in\mathcal{Y}_{\ell}\text{ such that }g_{\ell}(x_{\ell},y_{\ell})=0\right\}
\end{equation}
where $\hat{f}_{\ell}$ approximates the local objective $f_{\ell}$.
%\begin{rem}%[Local Approximation]
%Problem~\eqref{eq::opt::original} is generally nonconvex. The proposed predictor-corrector aggregation is to construct a tractable local approximation of the implicit distribution-level feasible set. The derived guarantees are local approximation-error bounds, while the numerical optimality comparisons are based on solutions returned by local nonlinear programming solvers. %Global certification would require additional techniques, such as global optimization or an exact convex relaxation with a verified zero relaxation gap.
%\end{rem}

%where $\hat{f}_{\ell}$ is the approximation of lower-layer cost function $f_{\ell}$, and $\mathcal{X}_{\ell}^{\text{imp}}$ is nonlinear mapping of the feasible set $\mathcal{Y}_{\ell}$ via nonlinear constraints $g_{\ell}(x_{\ell} ,y_{\ell})=0$ onto the coupling variables $x_{\ell}$, i.e.,
%\begin{equation}
%    \mathcal{X}_{\ell}^{\text{imp}}= \left\{x_{\ell}\in\mathcal{X}_{\ell}\,\mid
%    \,g_{\ell}(x_{\ell},y_{\ell})=0,\,y_{\ell}\in\mathcal{Y}_{\ell}\right\}.
%\end{equation}

%The main challenge arises from the nonlinear constraints $g_{\ell}(x_{\ell},y_{\ell})=0$, which render the implicit sets nonconvex and computationally intractable. 

{
The nonlinear implicit feasible sets \(\mathcal X_\ell^{\mathrm{imp}}\) in \eqref{eq::implicit::set} are generally nonconvex and computationally intractable. 
The predictor-corrector approximation introduced in the next section is not a convexification of these sets and does not provide a global optimality certificate for Problem~\eqref{eq::opt::original}. 
Instead, it constructs a tractable local surrogate of the implicit distribution-level feasible set around a regular operating point. 
The theoretical guarantees derived below are local approximation error bounds for the reconstructed state (Theorem~\ref{thm:accuracy}) and for the induced feasible-set approximation (Corollary~\ref{coro::set_error}).
}

\section{Nonlinear Flexibility Aggregation}
\label{sec::aggregation}

{This section develops the proposed nonlinear flexibility aggregation method for approximating the implicit feasible sets introduced in Section~\ref{sec::formulation::aggregation}. We first present the predictor-corrector approximation for flexibility aggregation, then derive local error bounds, and finally illustrate its performance through numerical examples and benchmark studies.}
\subsection{Predictor–Corrector Approximation}\label{sec::predictor-corrector}

We propose a predictor–corrector aggregation method to construct tractable surrogates of the implicit sets, thereby enabling coordination without requiring explicit distribution-grid models in the upper-level coordination problem. Starting from a known feasible point, the proposed method first computes a tangential predictor based on local Jacobian information, and then applies a corrector step to compensate for the residual of the nonlinear equality constraints. The resulting approximation $\hat{y}(x)$ provides a local surrogate of the exact implicit mapping while avoiding iterative nonlinear solves.

In the following, we illustrate how the proposed method approximates the implicit feasible set~\eqref{eq::implicit::set}. For notational simplicity, the subsystem subscript index $\ell$ is omitted in the following subsections.

Consider the nonlinear system associated with a single distribution grid:
\begin{equation}\label{eq::nonlinear::model}g(x,y) = 0,\;x\in\mathcal{X}\subseteq\mathbb{R}
    ^{m},\; y\in\mathcal{Y}\subseteq\mathbb{R}^{n},
\end{equation}
where $g:\mathbb{R}^{m}\times\mathbb{R}^{n}\rightarrow \mathbb{R}^{l}$
represents the system equality constraints in~\eqref{eq::opt::original::distribution}.

%Our approach is to apply the implicit function theorem to find a function $y^\star(x)$ to replace $y$ in \eqref{eq::nonlinear::model}. This allows us to formulate a feasible set for the coupling variables as 
{The exact projected feasible set of the coupling variable $x$ is given by
\begin{equation*}%\label{eq::exact::implicitSet}
    \mathcal{X}^{\text{imp}}= \left\{x\in\mathcal{X}\,\mid\,\exists\, y\in\mathcal{Y}\text{ such that }g(x,y)=0\right\}.
\end{equation*}
Hence, for every \(x\in\mathcal X^{\mathrm{imp}}\), there exists at least one \(y\in\mathcal Y\) satisfying \(g(x,y)=0\). Under the regularity conditions stated below, such a feasible \(y\) can be represented locally by an implicit mapping \(y^\star(\cdot)\) around a feasible point.
}

{Since computing the exact implicit set $\mathcal X^\mathrm{imp}$ is intractable, we propose a predictor-corrector method to approximate it.}
We first impose the regularity assumptions in the following:
\begin{myAss}
    \label{ass::c1} The function \(g\) is continuously differentiable on $\Pi := \mathcal{X}\times\mathcal{Y}.$
\end{myAss}

{\begin{myAss}
    \label{ass::invertible}  Let \((x_0,y_0)\in\mathcal X\times\mathcal Y\) be a feasible base point satisfying
\(g(x_0,y_0)=0\). The dimensions satisfy \(n=l\), and the base-point Jacobian
\(
    \partial_y g(x_0,y_0)\in\mathbb R^{n\times n}
\)
is nonsingular.
\end{myAss}}

Assumption~\ref{ass::invertible} is the standard regularity condition ensuring local existence and local uniqueness of the implicit mapping around a feasible operating point. Its numerical verification and practical implementation are provided in Remark~\ref{rem::numerical_regularities}. Based on the assumpption, we introduce the implicit function theorem in the following:
\begin{myTheorem}[Implicit Function Theorem]\label{thm:implicit_function}
Let \((x_0,y_0)\in\mathcal X\times\mathcal Y\) satisfy
\[
g(x_0,y_0)=0.
\]
Suppose Assumptions~\ref{ass::c1} and \ref{ass::invertible} hold in a neighborhood of \((x_0,y_0)\). Then there exist neighborhoods \(\mathcal U\subseteq\mathcal X\) of \(x_0\) and \(\mathcal V\subseteq\mathcal Y\) of \(y_0\), and a locally unique continuously differentiable mapping \(y^\star:\mathcal U\to\mathcal V\) such that
\begin{equation}\label{eq::localMapping}
g(x,y^\star(x))=0,\qquad \forall x\in\mathcal U.
\end{equation}
Moreover, 
\begin{equation}\label{eq::implicit::sens}
\hspace{-8pt}\partial_x y^\star(x)
=
-\left[\partial_y g(x,y^\star(x))\right]^{-1}\partial_x g(x,y^\star(x)),
\, \forall x\in\mathcal U
\end{equation}
\end{myTheorem}
Note that~\eqref{eq::implicit::sens} follows from differentiating \eqref{eq::localMapping} with respect to \(x\), which yields
\begin{equation*}
\partial_y g(x,y^\star(x))\,\partial_x y^\star(x)
+\partial_x g(x,y^\star(x))=0.
\end{equation*}
%For notational convenience, we define the paritial Jacobians at $(x_0,y_0)$ as
%\[
%J_{x} := \partial_x g(x_0,y_0),
%\qquad
%J_{y}:= \partial_y g(x_0,y_0),
%\]
For a given $x$ in the neighborhoods of $x_0$, the tangential predictor of $y$ at the base point \((x_0,y_0)\) is
\begin{align}    
   \label{eq::predictor}
    \bar{y}(x):= \;&y_0  - \squarebrackets{\partial_y g(x_0,y_0)}^{-1} \partial_x g(x_0,y_0)\, (x-x_0).%\notag\\
    %= \;&y_0  + S^\star(x_0) (x-x_0).
    % \frac{\partial y^{\star}(x_{0})}{\partial     x}\left(x- x_{0}\right)\notag\\ = \;&y_{0}- \left[\frac{\partial g}{\partial y}\right]^{-1}\frac{\partial g}{\partial x}\Delta x,
\end{align} 
Since $\bar y(x)$ is only a first-order approximation, it generally does not satisfy the nonlinear constraints $g(x,\bar y(x))=0$. We therefore introduce the corrector
\begin{equation*}
    %\Delta y^{\text{cor}}= - \left[\frac{\partial g}{\partial y}\right]^{-1}g(x,\bar{y}).
    \Delta y^{\mathrm{cor}}(x) := - \squarebrackets{\partial_y  g(x_0,y_0)}^{-1} \,g(x,\bar{y}(x))
\end{equation*}
which yields the predictor-corrector approximation:
\begin{align}\label{eq::predictor::corrector}
    %\hat{y}(x) = y_{0}\underbrace{- \left[\frac{\partial g}{\partial y}\right]^{-1}
    %\frac{\partial g}{\partial x} \parens{x-x_0}}_{\text{preditor step}}\underbrace{-
    %\left[\frac{\partial g}{\partial y}\right]^{-1} g(x,\bar{y})}_{\text{corrector step}}.
    \hat{y}(x) :=&\; \bar{y}(x) + \Delta y^{\mathrm{cor}}(x)\notag\\
    =&\; y_0 - \squarebrackets{\partial_y g(x_0,y_0)}^{-1} \partial_x g(x_0,y_0)\, (x-x_0) \notag\\
    &- \squarebrackets{\partial_y  g(x_0,y_0)}^{-1} \,g(x,\bar{y}(x))
\end{align}
{
The role of~\eqref{eq::predictor::corrector} is to approximate the equality-defined implicit relation $g(x,y)=0$. The original set constraint $y\in\mathcal{Y}$ is not enforced by the predictor-corrector step itself. Instead, it is incorporated through the reduced feasible set in the $x$-space induced by the implicit mapping. Accordingly, the approximated implicit feasible region is defined as
\begin{equation}\label{eq::approximate::implicitSet}
    \hat{\mathcal{X}}^{\text{imp}}= \left\{x\in\mathcal{X}\,\mid\,\hat{y}(x)\in
    \mathcal{Y}\right\}.
\end{equation}
Thus, the predictor-corrector construction approximates the equality-constrained manifold, while the inequality and bound constraints on \(y\) are enforced through the membership condition \(\hat y(x)\in\mathcal Y\). If the reconstructed state \(\hat y(x)\) satisfies the operational limits, i.e., \(\hat y(x)\in\mathcal Y\), the corresponding coupling point \(x\) is retained in \(\hat{\mathcal X}^{\mathrm{imp}}\). 
Otherwise, if \(\hat y(x)\notin\mathcal Y\), the coupling point is excluded from the predictor-corrector approximation of the implicit feasible region. 
}

{
\begin{myRemark} 
The proposed predictor-corrector aggregation provides a non-iterative local surrogate of the nonlinear implicit feasible set \(\mathcal X^{\mathrm{imp}}\). Compared with the existing aggregation methods summarized in Table~\ref{tab:flexibility_models}, it directly accounts for nonlinear AC power flow equations, is applicable to both radial and meshed distribution networks under the stated regularity assumptions, and enables scalable precomputation through spatio-temporal decomposition. 
\end{myRemark}
}

\subsection{Approximation Error Analysis}\label{sec::theorem}

We next analyze the local approximation error of the proposed predictor-corrector construction under the additional assumption of twice continuous differentiability.
\begin{myAss}
    \label{ass::C2}  The function \(g\) is twice continuously differentiable on \(\Pi := \mathcal{X}\times\mathcal{Y}\).
\end{myAss}
For convenience, we simplify notation by writing
    \[
    y^\star := y^\star(x),\qquad \bar y := \bar y(x),\qquad \hat y :=\hat{y} (x),
    \]
and define the sensitivity operator along the feasible manifold as
\begin{equation*}
\label{eq::Sstar}
S^\star(x):=\partial_x y^\star(x)
=
-\left[\partial_y g(x,y^\star(x))\right]^{-1}\partial_x g(x,y^\star(x)).
\end{equation*}
The following regularity result summarizes the local smoothness properties required in the subsequent analysis.

\begin{myLemma}\label{lemm:local_regularity}
Under the assumptions of Theorem~\ref{thm:implicit_function}, and additionally Assumption~\ref{ass::C2}, there exist constants \(L_1,L_2,L_3>0\) such that, for all \(x,x_1,x_2\in\mathcal U\), and all \(y_1,y_2\in\mathcal V\),
%\begin{subequations}
\begin{align}
\hspace{-5pt}\|S^\star(x_1)-S^\star(x_2)\| &\le L_1\|x_1-x_2\|,\label{eq::implicit::lipschitz::first}\\
\hspace{-5pt}\|\partial_y g(x_1,y^\star(x_1))-\partial_y g(x_2,y^\star(x_2))\| &\le L_2\|x_1-x_2\|,\label{eq::implicit::lipschitz::second}\\
\hspace{-5pt}\|\partial_y g(x,y_1)-\partial_y g(x,y_2)\| &\le L_3\|y_1-y_2\|.\label{eq::implicit::lipschitz::third}
\end{align}
%\end{subequations}
\end{myLemma}
We first establish the local error bound for the tangential predictor \eqref{eq::predictor}.
\begin{myLemma}\label{prop:predictor}
Under the assumptions of Lemma~\ref{lemm:local_regularity}, for all \(x\in\mathcal U\), the tangential predictor \(\bar y(x)\) defined in \eqref{eq::predictor} satisfies
 \[\bar y(x)-y^\star(x)=\mathcal O(\|x-x_0\|^2).\]
\end{myLemma}

\begin{proof}
    For a fixed \(x\in\mathcal U\), we simplify notation by writing
    \[
    y^\star := y^\star(x),\qquad \bar y := \bar y(x).
    \]
    From \eqref{eq::predictor} we have
    \[
        \bar{y}- y^{\star} =  \;y_0 - y^{\star} + S^\star(x_{0}) (x-x_{0}).
    \]      
    Since $y_0=y^\star(x_0)$, we apply the integral form of the mean value theorem to \(y^\star(\cdot)\) and obtain
    \[
    y_0-y^\star =-\int_{0}^1S^\star(x_0+\tau (x-x_0))(x-x_0)d\tau.
    \]
    Hence,    
    \begin{align*}
        \bar{y}- y^{\star}   =  -\int_{0}^{1}\squarebrackets{S^\star(x_0+\tau(x-x_0)) - S^\star(x_0)}(x-x_{0}) d\tau\notag
    \end{align*}
    Taking norms yields
    \begin{align}\label{eq::errorbound::predictor}
        \quad &\norm{\bar{y} - y^\star}\notag\\
        \leq&\norm{x-x_0}\int_{0}^{1}\norm{S^\star(x_0+\tau(x-x_0) - S^\star(x_0)}d\tau\notag \\
        \Downarrow&\;{\text{Lipschitz continuity~\eqref{eq::implicit::lipschitz::first}}}\notag                                   \\
        %\overset{\eqref{eq::implicit::lipschitz::first}}&{\leq} \norm{x-x_0}^2  \int_0^1 L_1 \tau d\tau\notag\\
        {\leq}& \norm{x-x_0}^2  \int_0^1 L_1 \tau d\tau\notag                                          \\
        \leq& \frac{1}{2}L_1\norm{x-x_0}^{2},
    \end{align}
    which concludes the proof.
\end{proof}

{To state the error bound with an explicit constant for the predictor-corrector approximation \eqref{eq::predictor::corrector}, let
$M_0:=\partial_y g(x_0,y_0)$ denote the base-point Jacobian, which is
nonsingular by Assumption~\ref{ass::invertible}, set
$\omega:=\norm{M_0^{-1}}$, and fix $\bar r>0$ such that
\[\mathcal B(x_0,\bar r):=\{x \in \mathcal U \subseteq \mathcal X\mid \norm{x-x_0}\le \bar r\}.\]}
\begin{myTheorem}\label{thm:accuracy}
Suppose the assumptions of Lemma~\ref{lemm:local_regularity} hold. Then, for all $x\in\mathcal B(x_0,\bar r)$, the predictor-corrector approximation $\hat y(x)$ defined in~\eqref{eq::predictor::corrector}
satisfies
{\[
\norm{\hat y(x)-y^\star(x)}\;\le\; C\,\norm{x-x_0}^3,
\]
with
{\[
C:=\tfrac{\omega}{2}\parens{L_1L_2+\tfrac14 L_1^2L_3\,\bar r}.
\]}
In particular, $\hat y(x)-y^\star(x)=\mathcal O(\norm{x-x_0}^{3}).$
}
\end{myTheorem}

\begin{proof}
    From \eqref{eq::predictor} and \eqref{eq::predictor::corrector}, we obtain  
    \begin{align*}
        \hat{y}- \bar{y}  & = -M^{-1}_0g(x,\bar{y})\notag\\
        \hat{y}-y^{\star} & = \bar{y}-y^{\star}- M^{-1}_0\parens{g(x,\bar{y}) - g(x,y^\star)}\notag         \end{align*}
    Applying the integral form of the mean value theorem to \(g(x,\cdot)\), we obtain
    \begin{align*}
    &g(x,\bar y)-g(x,y^\star)\notag\\
    =&
    \int_0^1
    \partial_y g\bigl(x,y^\star+\tau(\bar y-y^\star)\bigr)\,d\tau\,(\bar y-y^\star).
    \end{align*}
Hence,
\begin{align*}
\hat{y} - y^\star
 %         & = M^{-1}\Big\{M\parens{\bar{y} -y^\star} \notag\\
 %        & \qquad \qquad- \int_{0}^{1}\partial_y g(x,y^{\star}+ \tau\parens{\bar{y}-y^\star}) \cdot(\bar{y}-y^{\star})d\tau \Big\}\notag         \\
         = M^{-1}_0\Big\{M_0 - \int_{0}^{1}\partial_y g(x,y^{\star}+ \tau\parens{\bar{y}-y^\star}) d\tau \Big\}\parens{\bar{y} -y^\star}
    \end{align*}
    Taking the norm of both sides, we have
    \begin{align}\label{eq::errorbound::corrector}
        &\norm{\hat{y}-y^\star}\notag\\
        \leq &  \norm{M^{-1}_0}\Big\{\norm{M_0-\partial_y g(x,y^\star)}  + \int_0^1 \norm{\partial_y g(x,y^{\star}+\tau\parens{\bar{y}-y^\star}) - \partial_y g(x,y^{\star})}d\tau\Big\}\norm{\bar{y}-y^\star}\notag \\
        \Downarrow & \;{\text{Lipschitz continuous \eqref{eq::implicit::lipschitz::third}}}\notag\\
        %\overset{\eqref{eq::Lipschitz::g}}&{\leq} \norm{M^{-1}}\braces{\norm{M-\partial_y g(x,y^\star)} + L \norm{\bar{y}-y^\star}\int_0^1 \tau d\tau}\norm{\bar{y}-y^\star}\notag \\
        {\leq} & \norm{M^{-1}_0}\Big\{\norm{M_0-\partial_y g(x,y^\star)} + L_3 \norm{\bar{y}-y^\star}\int_0^1 \tau d\tau\Big\}\norm{\bar{y}-y^\star}\notag \\
        \Downarrow & \;{\text{Lipschitz continuous~\eqref{eq::implicit::lipschitz::second}}}\notag\\
        {\leq} & \norm{M^{-1}_0}\braces{L_2 \norm{x-x_0}+ \frac{L_3}{2}\norm{\bar{y}-y^\star}}\norm{\bar{y}-y^\star}\notag\\
        %\overset{\text{A.\ref{ass::lipschitz}}}&{\leq} \norm{M^{-1}}\braces{L \norm{\begin{bmatrix}x-x_0\\ y^\star-y_0\end{bmatrix}}+ \frac{L}{2}\norm{\bar{y}-y^\star}}\norm{\bar{y}-y^\star}\notag       \\
        %&{\leq} \norm{M^{-1}}\braces{L \norm{\begin{bmatrix}x-x_0\\ y^\star-y_0\end{bmatrix}}+ \frac{L}{2}\norm{\bar{y}-y^\star}}\norm{\bar{y}-y^\star}\notag       \\
        %&\Downarrow \;{\text{Lipschitz continuous (Eq.~\ref{eq::implicit::lipschitz::first})}}\notag\\
        %\overset{\eqref{eq::implicit::lipschitz::first}}&{\leq} \norm{M^{-1}}\braces{L (L_1+ 1)\norm{x-x_0}+ \frac{L}{2}\norm{\bar{y}-y^\star}}\norm{\bar{y}-y^\star}\notag       \\
        %&{\leq} \norm{M^{-1}}\braces{L (L_1+ 1)\norm{x-x_0}+ \frac{L}{2}\norm{\bar{y}-y^\star}}\norm{\bar{y}-y^\star}\notag       \\
        \Downarrow  &\;{\text{Remainder of $\bar{y}$~\eqref{eq::errorbound::predictor} % (Corollary~\ref{coro::predictor})
        }}\notag\\
        %&{\leq} \frac{1}{2} \omega\, L \braces{L_1(L_1+1) \norm{x-x_0}^3 + \frac{1}{4}\Tilde{L}^2_2 \norm{x-x_0}^4 }\\
        {\leq} & \frac{1}{2} \omega\, \braces{L_2 L_1 \norm{x-x_0}^3 + \frac{1}{4}L_3\,L^2_1 \norm{x-x_0}^4 },\notag\\
        %\Downarrow & \;\norm{x-x_0}\le\bar r\notag\\
        {\leq} & \frac{\omega}{2}\parens{L_1L_2+\tfrac14 L_1^2L_3\,\bar r}
                 \norm{x-x_0}^3 = C\norm{x-x_0}^3
    \end{align}
    which concludes the proof.
\end{proof}

{Theorem~\ref{thm:accuracy} also yields a guarantee for the induced
feasible set approximation~\eqref{eq::approximate::implicitSet}.
\begin{myCorollary}[Feasible-set approximation error]\label{coro::set_error}
Suppose the assumptions of Theorem~\ref{thm:accuracy} hold. Then, for all
$x\in\mathcal B(x_0,\bar r)$, the following conditons are satisfied:
\begin{flalign*}
\mathrm{dist}\parens{y^\star(x),\,\mathcal Y}\le C\norm{x-x_0}^3,\;&\forall x\in\hat{\mathcal X}^{\mathrm{imp}}\setminus
\mathcal X^{\mathrm{imp}},\\
\mathrm{dist}\parens{y^\star(x),\,\mathbb R^{n}\!\setminus\mathcal Y}
\le C\norm{x-x_0}^3,\;&\forall x\in\mathcal X^{\mathrm{imp}}\setminus
\hat{\mathcal X}^{\mathrm{imp}}.
\end{flalign*}
\end{myCorollary}

\begin{proof}
By Theorem~\ref{thm:accuracy},
$\norm{\hat y(x)-y^\star(x)}\le C\norm{x-x_0}^3$ on $\mathcal B(x_0,\bar r)$.
In the first case, $\hat y(x)\in\mathcal Y$ while $y^\star(x)\notin\mathcal Y$,
hence
$\mathrm{dist}(y^\star(x),\mathcal Y)\le\norm{y^\star(x)-\hat y(x)}
\le C\norm{x-x_0}^3$.
In the second case, $\hat y(x)\in\mathbb R^{n}\setminus\mathcal Y$ while
$y^\star(x)\in\mathcal Y$, hence
$\mathrm{dist}(y^\star(x),\mathbb R^{n}\!\setminus\mathcal Y)
\le\norm{y^\star(x)-\hat y(x)}\le C\norm{x-x_0}^3$.
\end{proof}

Hence misclassification can occur only where the exact state lies within
$C\norm{x-x_0}^3$ of the boundary of the operational limits. Coupling points
whose exact state is bounded away from this boundary are correctly classified
once $\norm{x-x_0}$ is sufficiently small.}

The next proposition provides an interpretation of the predictor-corrector approximation in the special case where \(g\) is quadratic.

\begin{myProposition}\label{prop::quadratic}
Suppose the assumptions of Theorem~\ref{thm:implicit_function} hold, and additionally assume that \(g:\mathbb R^m\times\mathbb R^n\to\mathbb R^n\) is quadratic. Then the predictor-corrector approximation \(\hat y(x)\) defined in \eqref{eq::predictor::corrector} coincides with the second-order Taylor approximation of the implicit mapping \(y^\star(x)\) around \(x_0\).
\end{myProposition}
\begin{proof}
For a quadratic function \(g\), the Taylor expansion of \(g(x,y)\) about the feasible point \((x_0,y_0)\) is exact up to second order. By construction, the tangential predictor \(\bar y\) cancels the constant and first-order terms in this expansion. Hence, the residual \(g(x,\bar y)\) contains only second-order terms in \(\Delta x:=x-x_0\).

The corrector step in \eqref{eq::predictor::corrector} is obtained by premultiplying this residual by \(-[\partial_y g(x_0,y_0)]^{-1}\), thereby producing exactly the second-order correction to the tangential predictor. Therefore, \(\hat y(x)\) coincides with the second-order Taylor approximation of the implicit mapping \(y^\star(x)\) around \(x_0\).
\end{proof}
{
\begin{myRemark}[Regularity of the base point]\label{rem::numerical_regularities}
Assumption~\ref{ass::invertible} is a local regularity condition imposed only at the selected feasible base point \((x_0,y_0)\). In implementation, the corrector is computed by solving linear systems with \(M_0=\partial_y g(x_0,y_0)\), whose numerical regularity can be assessed using its smallest singular value, condition number, or factorization pivots. Solvability~\cite[Thm.~2]{wang2016necessary} or voltage stability indicators~\cite[Thm.~2]{aolaritei2018hierarchical} may also be used as auxiliary screening tools. If \(M_0\) is singular or severely ill-conditioned, the base point can be rejected and the surrogate reconstructed around a better-conditioned feasible operating point or over a smaller neighborhood. A pseudo-inverse correction may also be used as a numerical fallback, although such a regularized step lies outside the theoretical predictor-corrector construction analyzed here.
\end{myRemark}
}

%This section discusses the practical implementation of the predictor-corrector aggregation and the numerical safeguards used when the Jacobian in Assumption~\ref{ass::invertible} becomes ill-conditioned. The discussion is presented for a single subsystem, with the subsystem index omitted for notational simplicity. 

\subsection{Illustrative Example: IEEE 33-Bus System}

We first illustrate the effectiveness of the proposed method using the IEEE 33-bus distribution network~\cite{baran2002network}. The system is evaluated under both TSO- and DSO-managed coordination models, where the predictor–corrector aggregation method is compared against two widely used linear surrogates.

\begin{figure}[htb!]
    % first row
    \begin{subfigure}{0.22\textwidth}
        \centering
        \caption{DSO-managed model}
        \label{fig::DSO::linear}
        \includegraphics[width=\columnwidth]{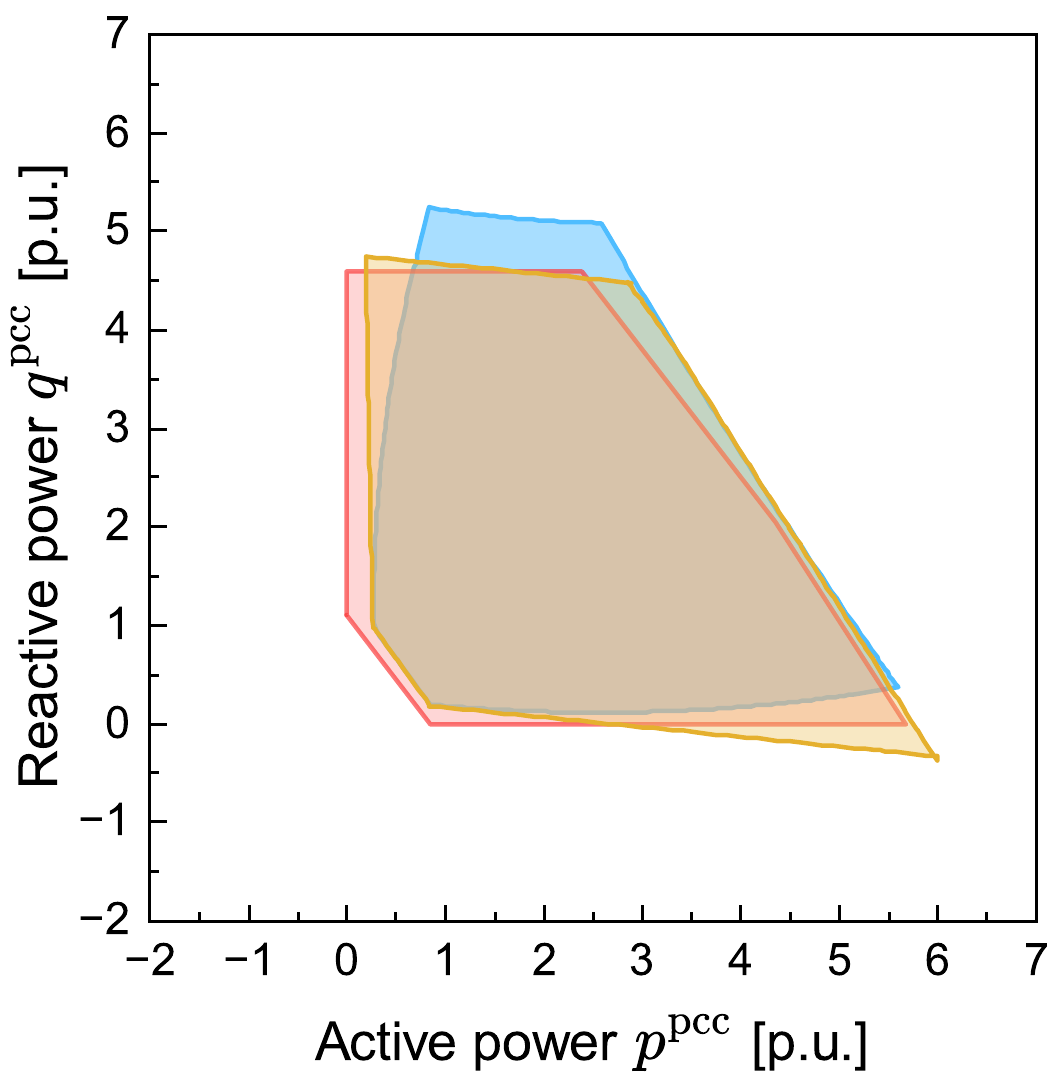}
    \end{subfigure}
    \hfill
    \begin{subfigure}{0.22\textwidth}
        \centering
        \caption{TSO-managed model}
        \label{fig::TSO::linear}
        \includegraphics[width=\columnwidth]{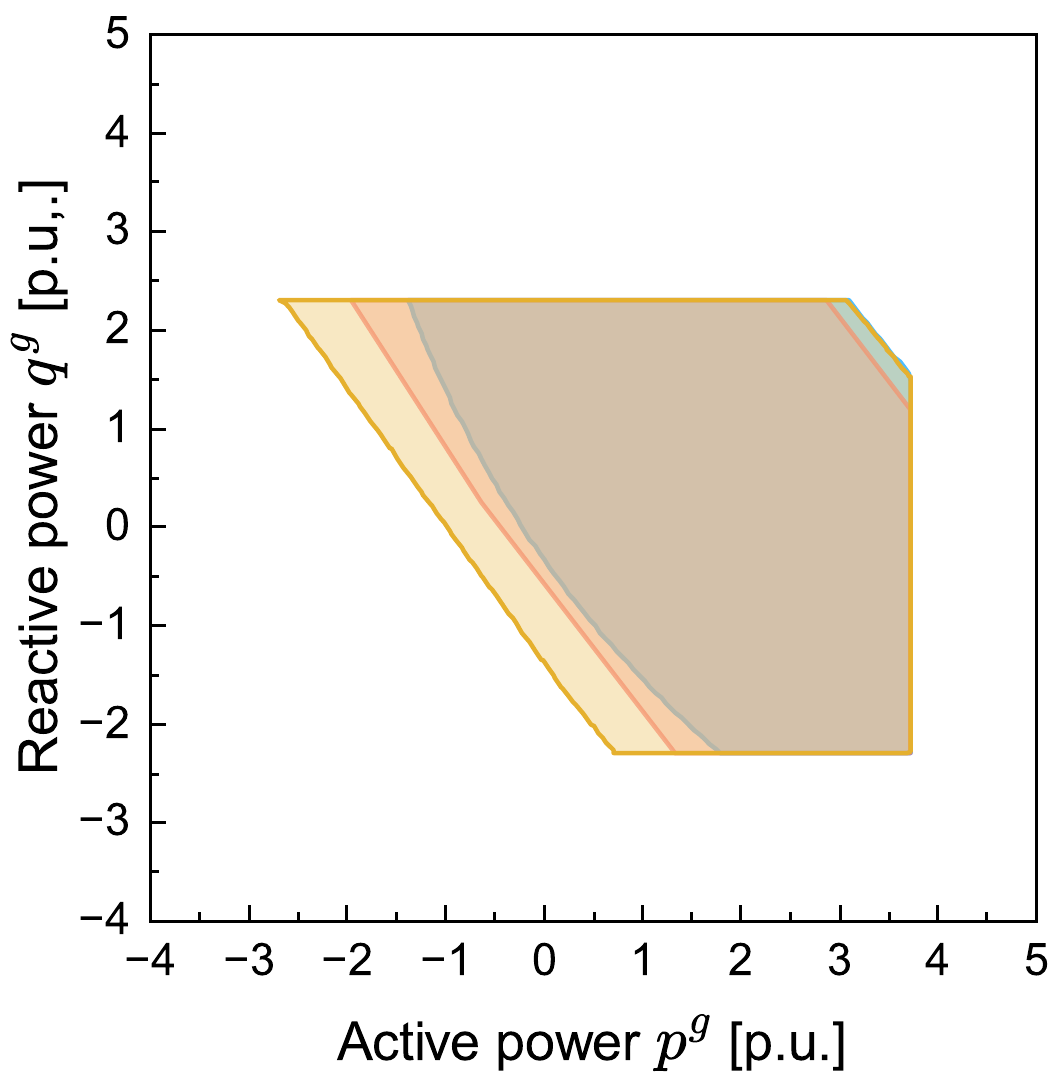}
    \end{subfigure}
    \hfill
    \begin{subfigure}{0.22\textwidth}
        \centering
        \caption{DSO-relaxation}
        \label{fig::DSO::relaxation}
        \includegraphics[width=\columnwidth]{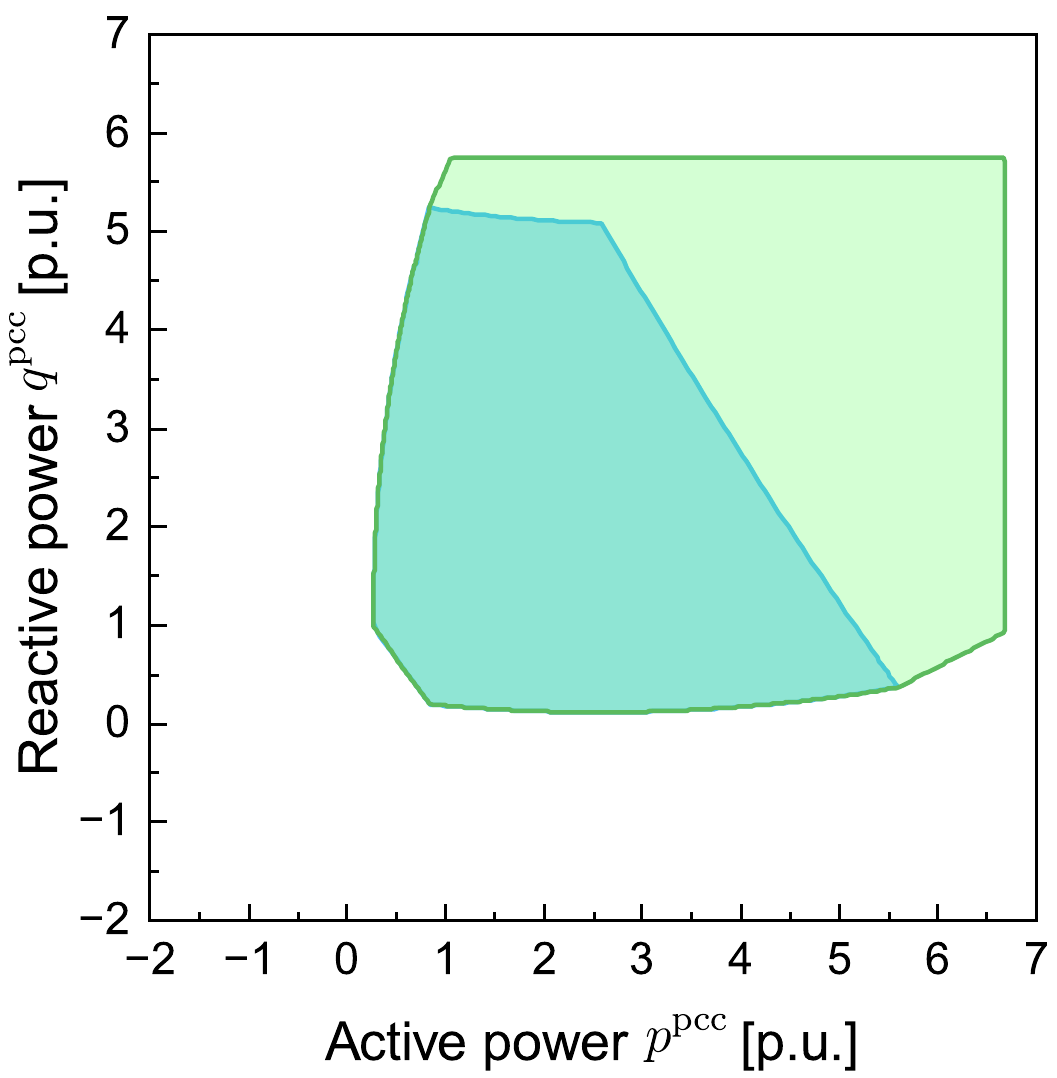}
    \end{subfigure}
    \hfill
    \begin{subfigure}{0.22\textwidth}
        \centering
        \caption{TSO-relaxation}
        \label{fig::TSO::relaxation}
        \includegraphics[width=\columnwidth]{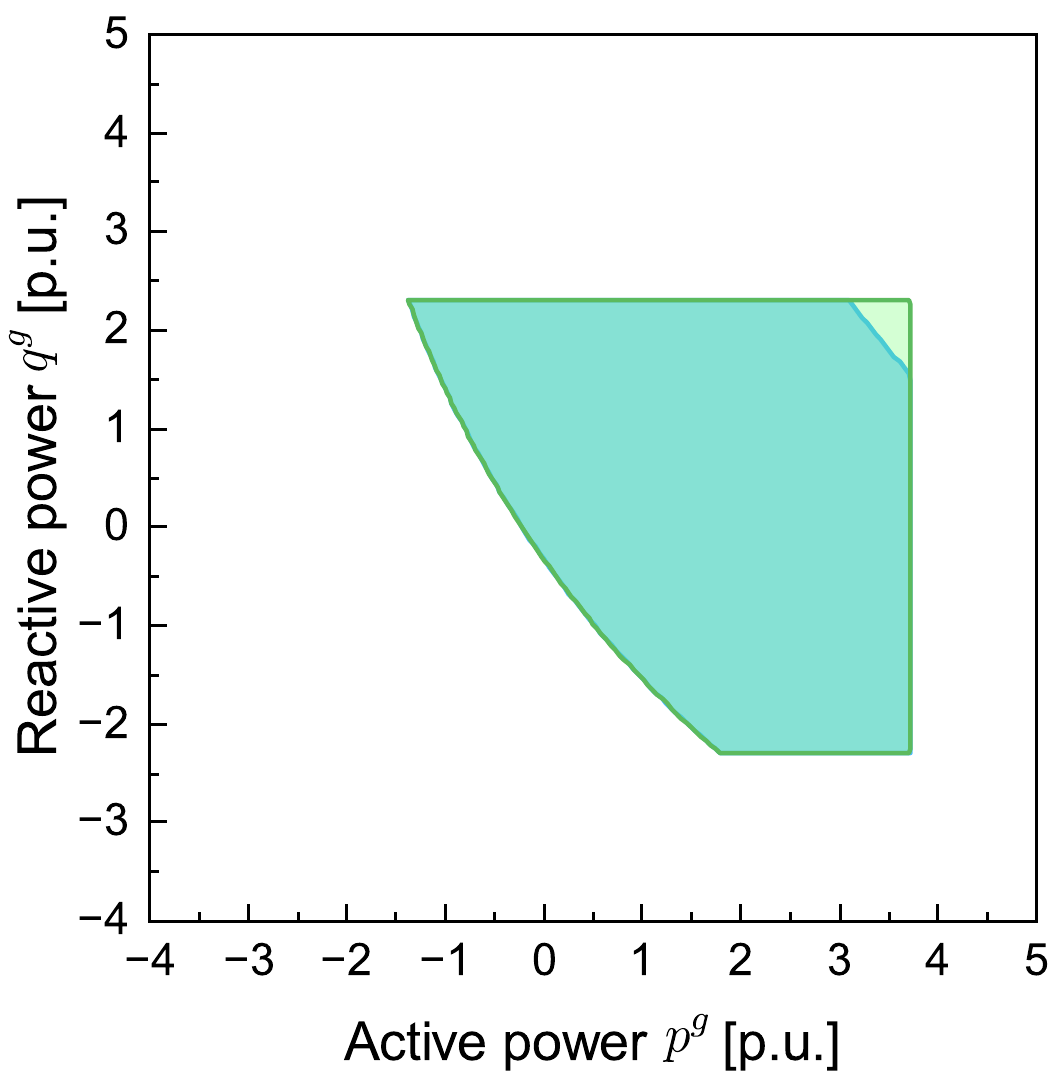}
    \end{subfigure}
    % \begin{subfigure}{0.22\textwidth}
    %     \centering
    %     \caption{DSO-restriction}
    %     \label{fig::DSO::restriction}
    %     \includegraphics[width=\columnwidth]{images/DSO_restriction.pdf}
    % \end{subfigure}
    % \hfill
    % \begin{subfigure}{0.22\textwidth}
    %     \centering
    %     \caption{TSO-restriction}
    %     \label{fig::TSO::restriction}
    %     \includegraphics[width=\columnwidth]{images/TSO_restriction.pdf} 
    % \end{subfigure}
    \caption{Comparisons of flexibility. The red and orange regions correspond to linear surrogate models, the green region to the convex relaxation, and the exact region is blue.}    \label{fig::comparison::method}
\end{figure}

Conventional studies typically rely on linear models, such as the LinDistFlow approximation or first-order tangential predictors of the AC equations. As shown in Figs.~\ref{fig::DSO::linear} and~\ref{fig::TSO::linear}, these models capture the general orientation of the reference feasible region \(\mathcal X^{\mathrm{imp}}\) but do not accurately reproduce its shape. Figs.~\ref{fig::DSO::relaxation} and~\ref{fig::TSO::relaxation} show that the convex relaxation of the DistFlow model~\cite{farivar2013branch,gan2015exact} provides an outer approximation that contains the exact AC feasible set but also introduces an additional false-feasible region.

\begin{figure}[htb!]
%    \begin{subfigure}{0.01\textwidth}
%%        \includegraphics[width=\columnwidth]{images/legend_linear.pdf}
 %   \end{subfigure}
    % second row
    \centering
    \begin{subfigure}{0.35\textwidth}
        \centering
        \caption{DSO-managed model}
        \label{fig::DSO::corrector}        
        \includegraphics[width=\columnwidth]{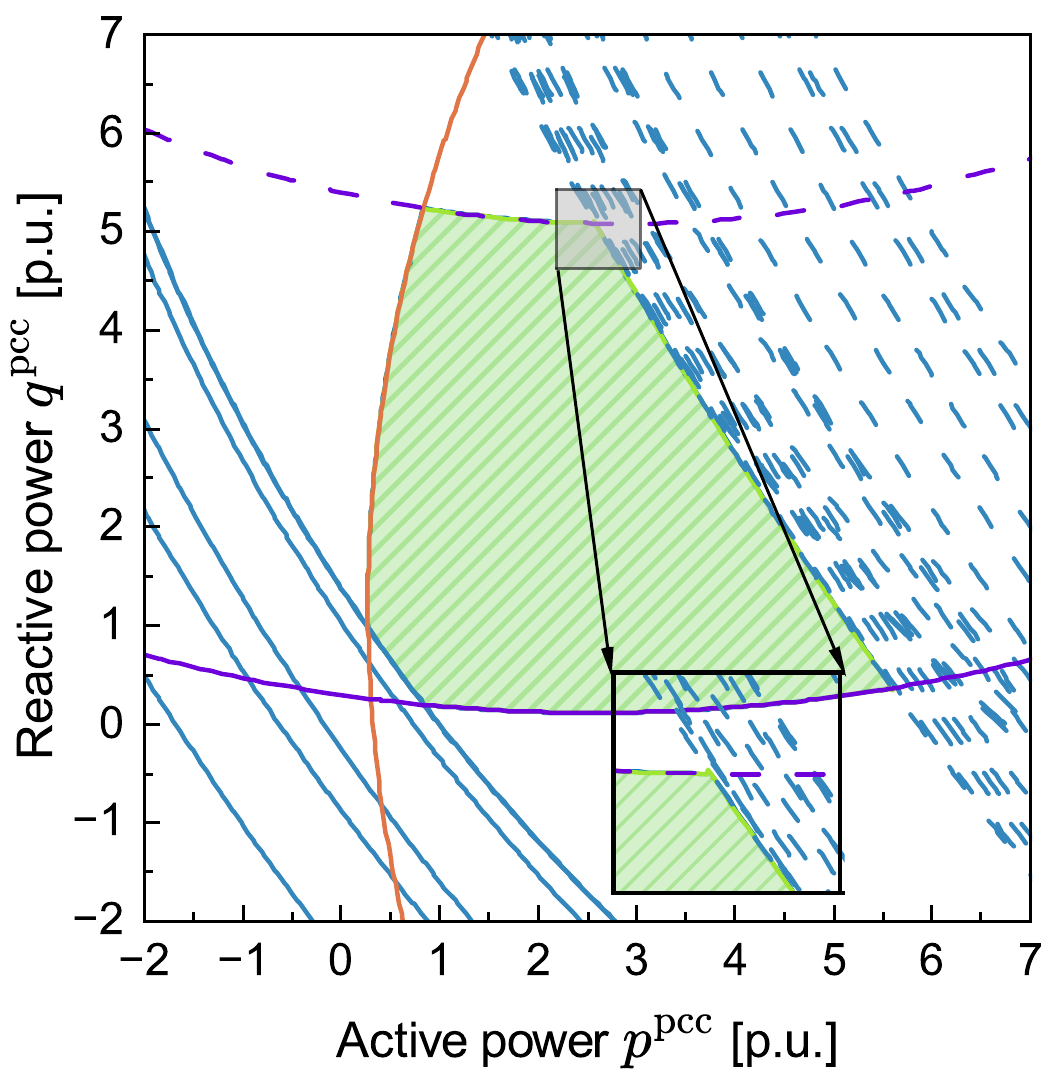}
    \end{subfigure}
    \hspace{40pt}
    \begin{subfigure}{0.35\textwidth}
        \centering
        \caption{TSO-managed model}
        \label{fig::TSO::corrector}
        \includegraphics[width=\columnwidth]{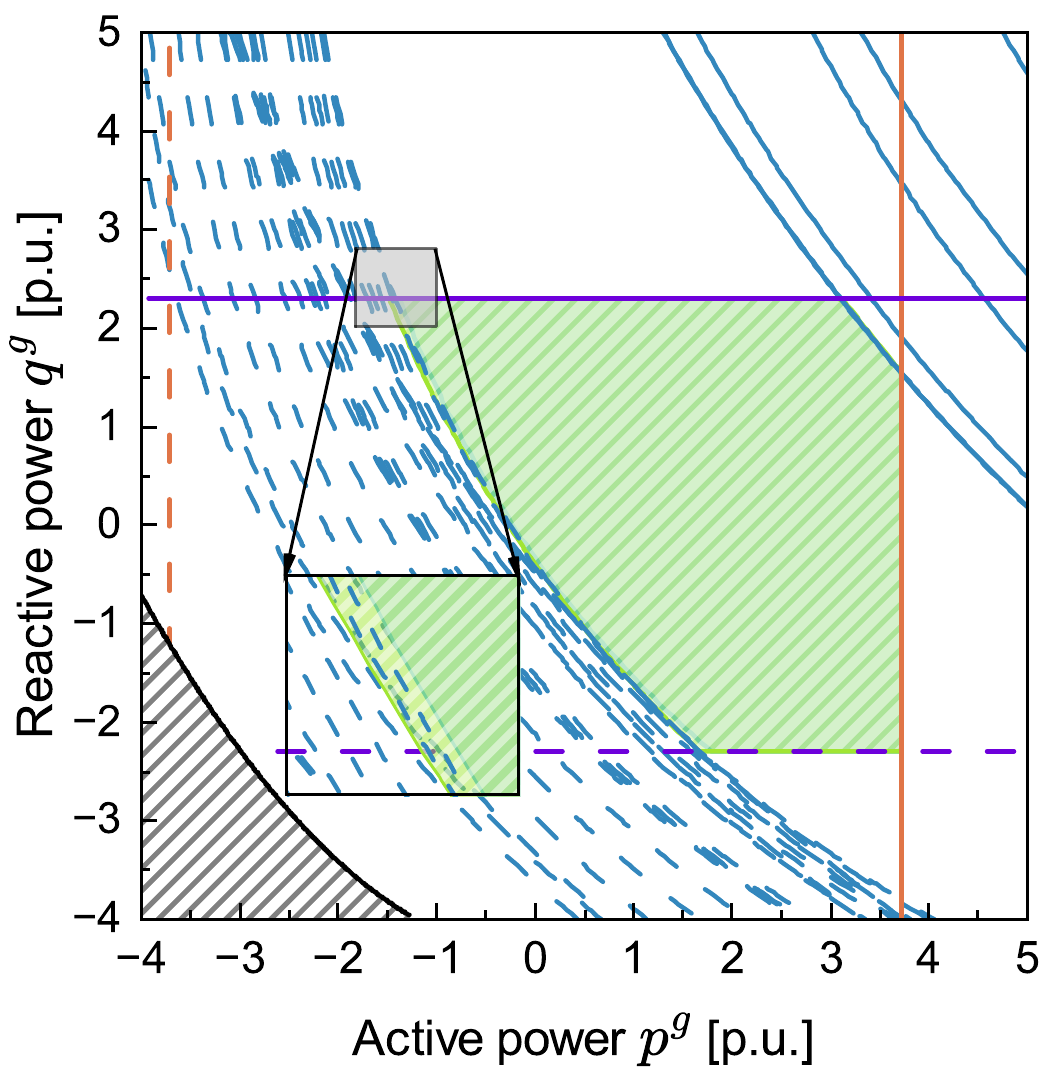}
    \end{subfigure}
    \caption{
    The proposed approximated set~\eqref{eq::approximate::implicitSet} vs. the exact feasible regions~\eqref{eq::implicit::set}. The solid and dashed curves give upper and lower limits on squared voltage magnitudes (blue), active power (red), and reactive power (purple) from constraint~\eqref{eq::pf::distribution::slimit}.}\label{fig::comparison::case33}
    \hfill
\end{figure}
\begin{figure}[htb!]    
    \centering
    %\begin{subfigure}{0.4\textwidth}
    %    \centering
    %    \vspace{1em}
    %    \includegraphics[width=\linewidth]{images/legend_bounds.pdf}
    %\end{subfigure}
    % third row
    \begin{subfigure}{0.35\textwidth}
        \centering
        \caption{DSO-managed model}        \label{fig::DSO::error}
        \includegraphics[width=\columnwidth]{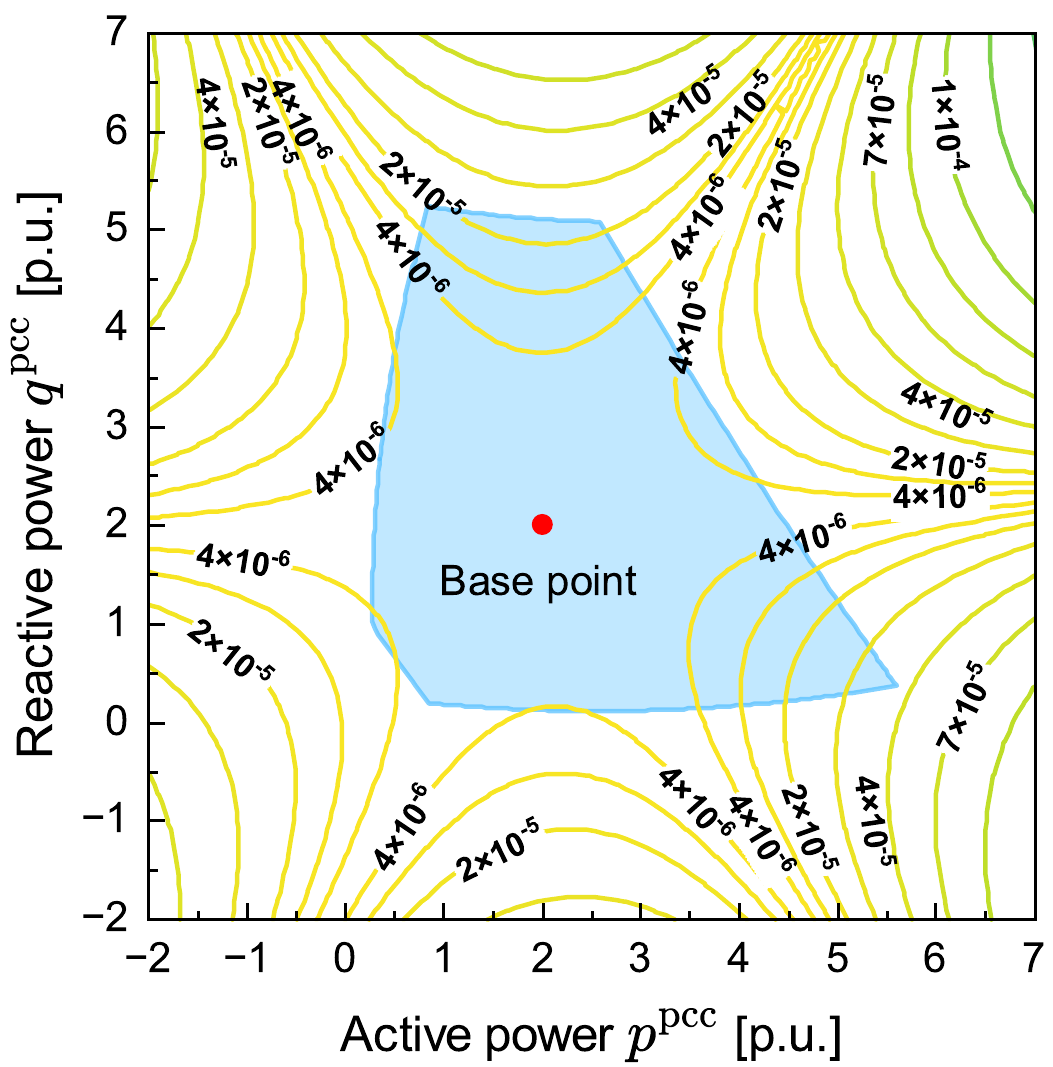}
    \end{subfigure}
    \hspace{40pt}
    \begin{subfigure}{0.4\textwidth}
         \centering
        \captionsetup{justification=raggedright,singlelinecheck=false,margin=5pt}
        \caption{TSO-managed model}        \label{fig::TSO::error}
        \includegraphics[width=\columnwidth]{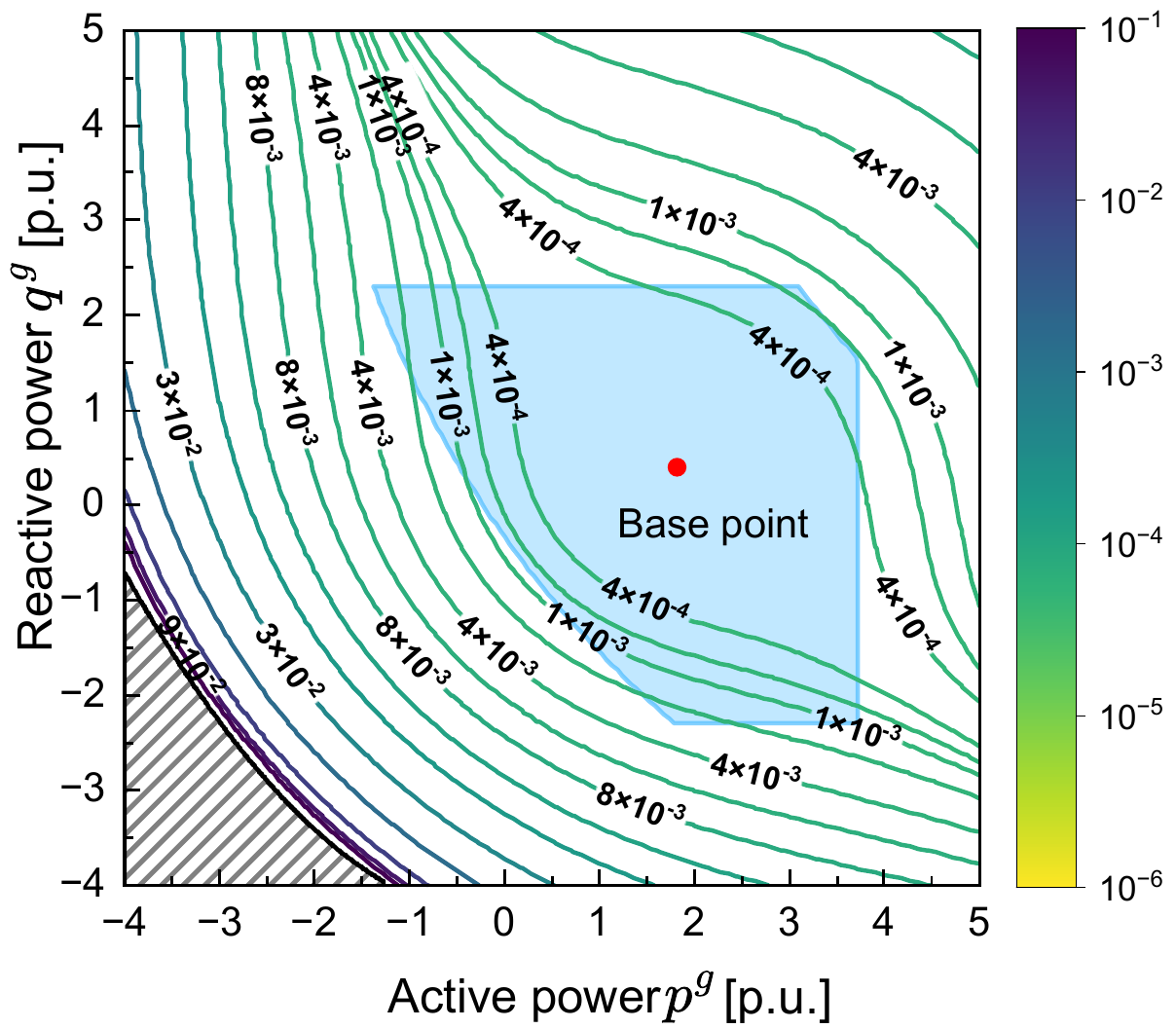}
    \end{subfigure}
    \caption{Error in voltage magnitude between the proposed approach~\eqref{eq::predictor::corrector} and exact AC power flow solutions.}
    \label{fig::illustrate}
\end{figure}
\begin{figure}[htb!]
    \centering
    \begin{subfigure}
        {0.35\textwidth}
        \centering
        \caption{DSO-managed model}
        \includegraphics[width=\columnwidth]{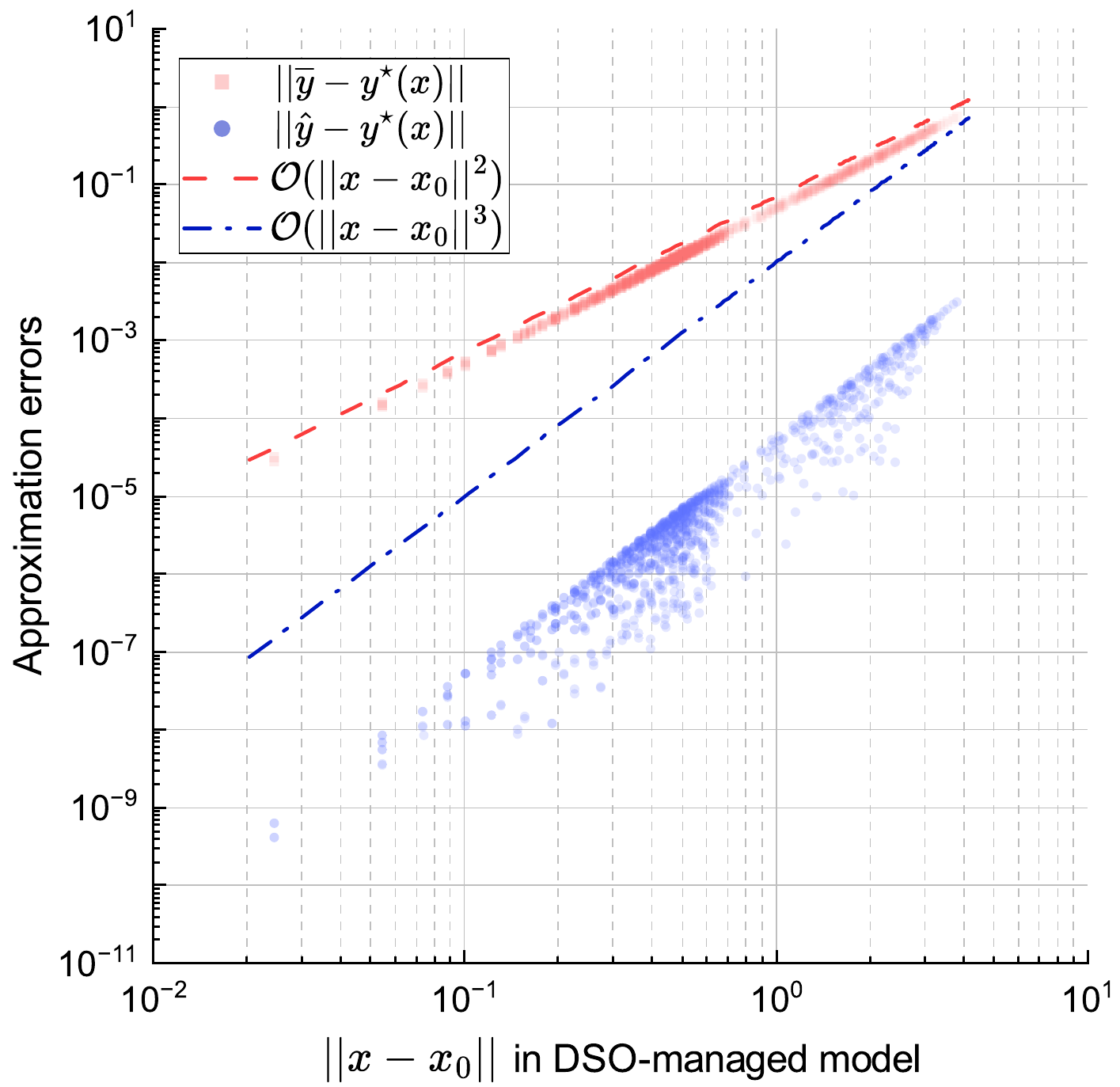}
        \label{fig::error::DSO}
    \end{subfigure}
    \hspace{40pt}
    \begin{subfigure}
        {0.35\textwidth}
        \centering
        \caption{TSO-managed model}
        \includegraphics[width=\columnwidth]{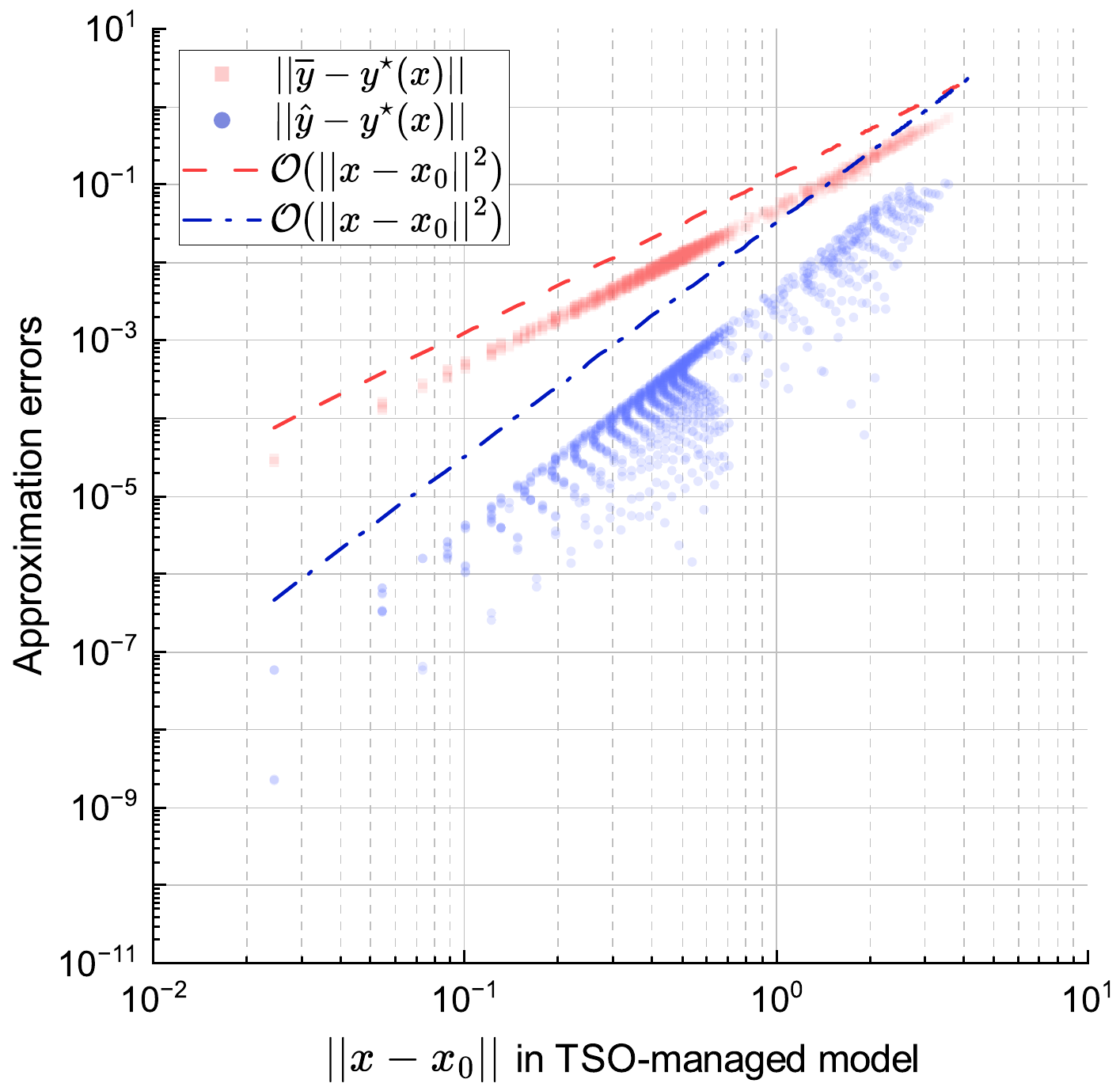}
        \label{fig::error::TSO}
    \end{subfigure}
    \caption{Error behavior of tangential predictor (red) and predictor–corrector (blue) approaches, with corresponding error upper bound (dotted lines).}
    \label{fig::errorbounds}
\end{figure}

In contrast, as shown in Fig.~\ref{fig::comparison::case33}, the proposed predictor–corrector method substantially improves the approximation. Fig.~\ref{fig::DSO::corrector} and Fig.~\ref{fig::TSO::corrector} depict the feasible regions obtained from the proposed set $\hat{\mathcal{X}}^\text{imp}$, showing close alignment with the exact AC feasible region. The improvement becomes more evident in Fig.~\ref{fig::DSO::error} and Fig.~\ref{fig::TSO::error}, which display the error of the voltage magnitudes for the proposed approximation. 
In the flexibility sets, the voltage error is mostly below $4\times10^{-6}$ for the DSO-managed model in Fig.~\ref{fig::DSO::error}, and below $1\times10^{-3}$ almost everywhere for the TSO-managed model in Fig.~\ref{fig::TSO::error}. 
Moreover, Fig.~\ref{fig::errorbounds} quantifies the approximation errors of both the tangential predictor and the predictor–corrector method, together with the theoretical error bounds derived in~\eqref{eq::errorbound::predictor}–\eqref{eq::errorbound::corrector}, with Lipschitz constants 
computed by~\cite[Proposition C.29]{lee2012smooth}.
%computed as Appendix~\ref{app::lipschitz}.

Since the AC power flow equations are quadratic, the predictor–corrector aggregation method effectively provides a second-order approximation of the implicit feasible mapping (Proposition~\ref{prop::quadratic}). This enables the method to capture boundary curvature and preserve feasibility with greater accuracy than purely linear surrogates, making it a powerful tool for representing distribution-level flexibility in hierarchical ITD coordination.

\subsection{Benchmarking Approximation Accuracy}

While the IEEE 33-bus system provides an illustrative comparison between methods, we construct a benchmark suite comprising 24 representative radial distribution system test cases with 7 meshed variants to systematically evaluate the accuracy and robustness of the proposed predictor-corrector method. These include all distribution system models from \matpower and a real-world distribution network, the KIT Campus Nord grid (\texttt{KIT Campus North})~\cite{gonzalez2021probabilistic}, along with their meshed variants. yl{All optimization problems are implemented in MATLAB and solved with \ipopt~\cite{wachter2006implementation} on a MacBook Pro with an M3 Pro chip.}

\begin{figure}[htbp!]
    \centering
    \includegraphics[width=0.4\linewidth]{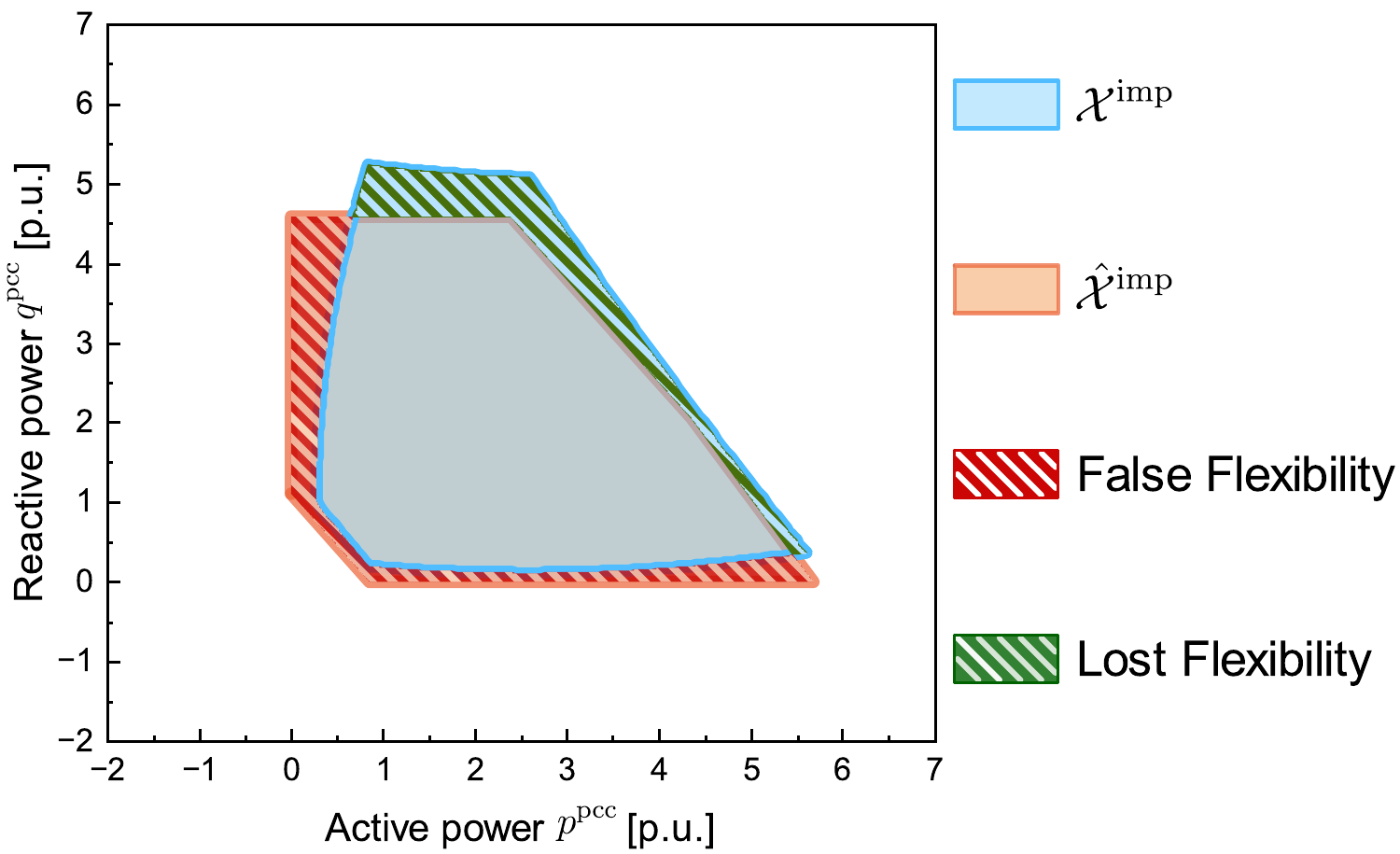}
    \caption{Illustration of the benchmark metrics.}   \label{fig::fractions}\vspace{-5pt}
\end{figure}

%Consider the nonlinear system associated with a single distribution grid~\eqref{eq::nonlinear::model},
%\begin{equation*}
%    g(x,y)=0,\qquad x\in\mathcal X\subseteq\mathbb R^m,\quad y\in\mathcal Y\subseteq\mathbb R^n.
%\end{equation*}
Consider the nonlinear system associated with a single distribution grid~\eqref{eq::nonlinear::model}.
Accuracy is quantified by comparing the approximate feasible set~\eqref{eq::approximate::implicitSet},
\[
\hat{\mathcal X}^{\mathrm{imp}}
=
\left\{
x\in\mathcal X
\,\middle|\,
\hat y(x)\in\mathcal Y
\right\},
\]
with the exact implicit feasible set~\eqref{eq::implicit::set},
\[
\mathcal X^{\mathrm{imp}}
=
\left\{
x\in\mathcal X
\,\middle|\,
\exists\,y\in\mathcal Y
\text{ such that }g(x,y)=0
\right\}.
\]
As shown in Fig.~\ref{fig::fractions}, two complementary performance metrics are considered:
\begin{itemize}
    \item \textbf{False Flexibility [\%]}: the number of uniformly sampled points in
    \(\hat{\mathcal X}^{\mathrm{imp}}\setminus\mathcal X^{\mathrm{imp}}\),
    normalized by the number of sampled points in
    \(\mathcal X^{\mathrm{imp}}\);
    \item \textbf{Lost Flexibility [\%]}: the number of uniformly sampled points in
    \(\mathcal X^{\mathrm{imp}}\setminus\hat{\mathcal X}^{\mathrm{imp}}\),
    normalized by the number of sampled points in
    \(\mathcal X^{\mathrm{imp}}\).
\end{itemize}
The same set of coupling points, sampled uniformly from a common bounding box, is used for all methods.
For each sampled coupling point \(x\), membership in
\(\mathcal X^{\mathrm{imp}}\) is evaluated by fixing \(x\), solving the full nonlinear AC feasibility problem for \(y\), and checking whether \(y\in\mathcal Y\). A sampled point is classified as feasible if the equality constraints~\eqref{eq::pf::distribution::uslack}--\eqref{eq::pf::distribution::vlimit} are satisfied up to the prescribed numerical tolerance and all operational constraints, including~\eqref{eq::pf::distribution::slimit}, are satisfied.
Note that the false flexibility may exceed \(100\%\) when the false-feasible region is larger than the exact feasible region.

By Corollary~\ref{coro::set_error}, both metrics count sampled points in the symmetric difference of $\hat{\mathcal X}^{\mathrm{imp}}$ and $\mathcal X^{\mathrm{imp}}$, which the theory predicts to concentrate in $C\norm{x-x_0}^3$-thin band around the exact boundary of $\mathcal Y$. The benchmarks below quantify this misclassification empirically.

Table~\ref{tab::aggregation-results} compares five methods: the LinDistFlow approximation (LDS), the enhanced DC approximation (EDC), the convex relaxation of the DistFlow model (CR), the tangential predictor (TP), and the proposed predictor-corrector method (PC). The feasible region obtained using convex relaxation is an outer approximation of the exact implicit feasible set and therefore overestimates the available flexibility. This behavior is evident in Table~\ref{tab::aggregation-results}, where CR has zero lost flexibility in all applicable radial cases but large false flexibility due to its false-feasible outer region. By contrast, the proposed method consistently yields the lowest false flexibility and lost flexibility values across the benchmark networks, indicating the closest agreement with the nonlinear AC feasible region.

%These static benchmarks establish the advantages of the predictor–corrector method in approximating nonlinear feasible regions. 
Having demonstrated the accuracy of the predictor-corrector method in the static benchmarks, the next section extends the analysis to a multiperiod setting, where temporal coupling and time-varying demand further complicate flexibility aggregation and coordination.

\begin{table*}[t]
\centering
\setlength{\tabcolsep}{4pt}
\setlength{\extrarowheight}{-1pt}
\caption{Approximation errors: infeasibility \& flexibility loss.}
\label{tab::aggregation-results}
\begin{threeparttable}
\begin{tabular*}{\textwidth}{@{\extracolsep{\fill}}lc rrrrr rrrrr@{}}
\toprule
\multirow{2}{*}{\textbf{Cases}} & \multirow{2}{*}{\textbf{Topology}} & \multicolumn{5}{c}{\textbf{False Flexibility [\%]}} & \multicolumn{5}{c}{\textbf{Lost Flexibility [\%]}} \\
\cmidrule(lr){3-7} \cmidrule(lr){8-12}
& & LDS & EDC & CR & TP & PC & LDS & EDC & CR & TP & PC \\
\midrule
\texttt{case10ba} & radial & 5.14 & 4.58 & 161.02 & 4.22 & 0.04 & 5.63 & 2.87 & 0.00 & 2.12 & 0.04 \\
\texttt{case12da} & radial & 2.18 & 2.90 & 67.44 & 3.25 & 0.00 & 6.45 & 3.81 & 0.00 & 2.12 & 0.03 \\
\texttt{case15da} & radial & 3.78 & 4.13 & 65.22 & 4.29 & 0.01 & 10.77 & 6.48 & 0.00 & 3.46 & 0.01 \\
\texttt{case15nbr} & radial & 2.41 & 2.81 & 70.31 & 2.99 & 0.00 & 9.12 & 6.22 & 0.00 & 2.99 & 0.00 \\
\texttt{case17me} & radial & 3.81 & 4.23 & 164.99 & 3.48 & 0.00 & 6.10 & 2.22 & 0.00 & 1.90 & 0.01 \\
\texttt{case18nbr} & radial & 2.39 & 2.96 & 68.89 & 3.36 & 0.00 & 10.02 & 6.35 & 0.00 & 3.06 & 0.01 \\
\texttt{case22} & radial & 0.99 & 1.31 & 55.30 & 1.41 & 0.01 & 5.25 & 3.39 & 0.00 & 1.32 & 0.01 \\
\texttt{case28da} & radial & 7.74 & 7.96 & 146.87 & 7.81 & 0.01 & 11.30 & 6.57 & 0.00 & 4.08 & 0.04 \\
\texttt{case33bw} & radial & 5.22 & 5.23 & 171.12 & 4.77 & 0.03 & 6.99 & 4.14 & 0.00 & 2.42 & 0.00 \\
\texttt{case33bw} & mesh & -- & 3.56 & -- & 3.96 & 0.00 & -- & 6.63 & -- & 3.43 & 0.01 \\
\texttt{case33mg} & radial & 5.51 & 5.41 & 204.76 & 4.84 & 0.00 & 6.56 & 3.91 & 0.00 & 2.35 & 0.01 \\
\texttt{case33mg} & mesh & -- & 3.87 & -- & 4.23 & 0.01 & -- & 6.88 & -- & 3.68 & 0.01 \\
\texttt{case34sa} & radial & 4.02 & 3.93 & 28.69 & 3.91 & 0.00 & 11.56 & 7.82 & 0.00 & 3.94 & 0.01 \\
\texttt{case38si} & radial & 8.18 & 8.01 & 198.12 & 7.07 & 0.04 & 10.24 & 7.29 & 0.00 & 4.06 & 0.11 \\
\texttt{case51ga} & radial & 6.83 & 6.25 & 171.91 & 4.82 & 0.04 & 5.47 & 2.70 & 0.00 & 1.69 & 0.03 \\
\texttt{case51he} & radial & 2.62 & 3.08 & 65.58 & 3.09 & 0.00 & 8.88 & 6.19 & 0.00 & 3.18 & 0.01 \\
\texttt{case69} & radial & 10.49 & 10.74 & 124.58 & 4.35 & 0.07 & 2.54 & 1.49 & 0.00 & 0.73 & 0.03 \\
\texttt{case74ds} & radial & 2.98 & 3.51 & 79.07 & 3.73 & 0.01 & 6.09 & 4.46 & 0.00 & 2.83 & 0.01 \\
\texttt{case85} & radial & 14.28 & 13.41 & 277.33 & 11.69 & 0.00 & 12.01 & 7.07 & 0.00 & 4.41 & 0.00 \\
\texttt{case94pi} & radial & 14.35 & 13.22 & 353.82 & 11.39 & 0.06 & 11.19 & 7.07 & 0.00 & 4.03 & 0.08 \\
\texttt{case118zh} & radial & 17.22 & 15.99 & 228.94 & 15.46 & 0.82 & 14.51 & 10.25 & 0.00 & 6.55 & 0.05 \\
\texttt{case118zh} & mesh & -- & 6.99 & -- & 7.07 & 0.05 & -- & 6.26 & -- & 3.90 & 0.03 \\
\texttt{case136ma} & radial & 10.27 & 8.77 & 83.39 & 7.08 & 0.22 & 11.10 & 7.84 & 0.00 & 4.96 & 0.27 \\
\texttt{case136ma} & mesh & -- & 4.67 & -- & 3.63 & 0.00 & -- & 6.88 & -- & 3.58 & 0.18 \\
\texttt{case141} & radial & 4.98 & 5.00 & 111.33 & 4.74 & 0.02 & 8.06 & 4.39 & 0.00 & 2.75 & 0.02 \\
\texttt{case533mt\_hi} & radial & 4.47 & 2.75 & 419.33 & 2.11 & 0.18 & 5.69 & 4.35 & 0.00 & 1.64 & 0.02 \\
\texttt{case533mt\_hi} & mesh & -- & 2.30 & -- & 2.31 & 0.20 & -- & 3.67 & -- & 1.78 & 0.00 \\
\texttt{case533mt\_lo} & radial & 5.26 & 3.90 & 465.99 & 2.08 & 0.37 & 2.81 & 2.40 & 0.00 & 1.64 & 0.00 \\
\texttt{case533mt\_lo} & mesh & -- & 3.46 & -- & 2.15 & 0.25 & -- & 2.36 & -- & 1.94 & 0.01\\
\texttt{KIT-CN} & radial & 3.61 & 2.35 & 48.34 & 0.93 & 0.04 & 3.86 & 2.66 & 0.00 & 1.01 & 0.01 \\
\texttt{KIT-CN} & mesh & -- & 1.75 & -- & 1.11 & 0.48 & -- & 1.82 & -- & 0.83 & 0.00 \\
\bottomrule
\end{tabular*}
\begin{tablenotes}[flushleft]
\item[*] \scriptsize LDS, EDC, CR, TP, and PC denote, respectively, LinDistFlow, enhanced DC, convex DistFlow relaxation, tangential predictor, and the proposed predictor-corrector method.
\end{tablenotes}
\end{threeparttable}
\end{table*}

\section{Multiperiod Coordination Problems}
\label{sec::multiperiod}

Building on the spatial decomposition achieved by the predictor–corrector aggregation method, we now extend the framework to multiperiod settings. In this case, the storage \acrshort{der}s introduce intertemporal coupling, linking decisions across successive time steps. We demonstrate how implicit feasible sets can be constructed in parallel for each period while temporal consistency is enforced through coupling constraints, enabling efficient spatio-temporal decomposition.

\subsection{Hierarchical Formulation}
Consider a time horizon $\mathcal{K} = \{1,\dots, N-1\}$. At each time step $k\in\mathcal{K}$, we collect the coupling variables and local variables into 
$$x_{k}  =  \{x_{\ell\hspace{-0.5pt} \mid \hspace{-0.5pt}k}\}_{\ell\in\mathcal{S^\text{D}}} \text{ and }  y_{k}  =  \{y_{\ell\hspace{-0.5pt} \mid \hspace{-0.5pt}k}\}_{\ell\in\mathcal{S}}.$$
Stacking them across the horizon yields
$$x      =  \{x_{k}\}_{k\in\mathcal{K}} \text{ and }                       y     =  \{y_{k}\}_{k\in\mathcal{K}}.$$

In the DSO-managed model, the \acrshort{tso} has no access to individual DER decisions within distribution systems. Temporal coupling therefore involves only local variables, arising for instance from \acrfull{soc} dynamics or ramping-rate limits of DERs over the horizon. As a result, the hierarchical problem~\eqref{eq::opt::original} can be reformulated as
\begin{subequations}
    \label{eq::mopt::original}
    \begin{align}
        \min_{x,y}\,\sum_{k\in\mathcal{K}}\Big\{ f_{0\hspace{-0.5pt} \mid \hspace{-0.5pt}k}(x_{k},y_{0\hspace{-0.5pt} \mid \hspace{-0.5pt}k})+\sum_{\ell\in\mathcal{S}^\text{D}}f_{\ell\hspace{-0.5pt} \mid \hspace{-1pt}k}(x_{\ell\hspace{-0.5pt} \mid \hspace{-0.5pt}k}, y_{\ell\hspace{-0.5pt} \mid \hspace{-0.5pt}k})\Big\}\label{eq::mopt::original::cost}
    \end{align}
    subject to
    \begin{flalign}
        &&g_{0\hspace{-0.5pt} \mid \hspace{-0.5pt}k}(x_{k},y_{0\hspace{-0.5pt} \mid \hspace{-0.5pt}k})=& \,0,\, y_{0\hspace{-0.5pt} \mid \hspace{-0.5pt}k}\in\mathcal{Y}_{0\hspace{-0.5pt} \mid \hspace{-0.5pt}k},&\forall k\hspace{-2.5pt}\in\hspace{-2.5pt}\mathcal{K} \label{eq::mopt::original::master}\\&&
        g_{\ell\hspace{-0.5pt} \mid \hspace{-0.5pt}k}(x_{\ell\hspace{-0.5pt} \mid \hspace{-0.5pt}k},y_{\ell\hspace{-0.5pt} \mid \hspace{-0.5pt}k})=& \,0, &\forall\ell\hspace{-2.5pt}\in\hspace{-2.5pt}\mathcal{S}^{\text{D}},\forall k\hspace{-2.5pt}\in\hspace{-2.5pt}\mathcal{K}  \label{eq::mopt::original::worker}\\&&
        (x_{\ell\hspace{-0.5pt} \mid \hspace{-0.5pt}k},y_{\ell\hspace{-0.5pt} \mid \hspace{-0.5pt}k})\in&\,\mathcal{X}_{\ell\hspace{-0.5pt} \mid \hspace{-0.5pt}k}\times \mathcal{Y}_{\ell\hspace{-0.5pt} \mid \hspace{-0.5pt}k}, \hspace{-7pt} & \forall \ell \hspace{-2.5pt}\in\hspace{-2.5pt}\mathcal{S}^{\text{D}},\forall k\hspace{-2.5pt}\in\hspace{-2.5pt}\mathcal{K} \label{eq::mopt::original::set}\\&&
        M_{\ell\hspace{-0.5pt} \mid \hspace{-0.5pt}k+1}\,y_{\ell\hspace{-0.5pt} \mid \hspace{-0.5pt}k+1}=& \, A_{\ell\hspace{-0.5pt} \mid \hspace{-0.5pt}k}\, y_{\ell\hspace{-0.5pt} \mid \hspace{-0.5pt}k},&\forall\ell\hspace{-2.5pt}\in\hspace{-2.5pt}\mathcal{S}^{\text{D}} ,\forall k\hspace{-2.5pt}\in\hspace{-2.5pt}\mathcal{K}  \label{eq::mopt::original::time} 
    \end{flalign}
\end{subequations}
where the objective \eqref{eq::mopt::original::cost} sums all transmission- and distribution-level costs over the entire time horizon $\mathcal{K}$. Constraints \eqref{eq::mopt::original::master} enforce the transmission power flow feasibility at each period, while \eqref{eq::mopt::original::worker}~\eqref{eq::mopt::original::set} impose the local distribution constraints and variable bounds. These first three groups of constraints are fully decoupled both in space (across subsystems) and in time, while~\eqref{eq::mopt::original::time} introduces a local temporal coupling over horizon $\mathcal{K}$.

\begin{myRemark}
    The temporal coupling~\eqref{eq::mopt::original::time} is a standard state-space relation. By partitioning $y_{\ell\hspace{-0.5pt} \mid \hspace{-0.5pt}k}$ into state variables $y_{\ell\hspace{-0.5pt} \mid \hspace{-0.5pt}k}^\textrm{\upshape state}$ and input variables $y_{\ell\hspace{-0.5pt} \mid \hspace{-0.5pt}k}^\textrm{\upshape input}$, we can rewrite~\eqref{eq::mopt::original::time} into a standard state-space form:
    \begin{align*}
        y_{\ell\hspace{-0.5pt} \mid \hspace{-0.5pt}k+1}^\textrm{s\upshape tate} = A^\textrm{\upshape state}_{\ell\hspace{-0.5pt} \mid \hspace{-0.5pt}k} y_{\ell\hspace{-0.5pt} \mid \hspace{-0.5pt}k}^\textrm{\upshape state} + B^\textrm{\upshape input}_{\ell\hspace{-0.5pt} \mid \hspace{-0.5pt}k} y_{\ell\hspace{-0.5pt} \mid \hspace{-0.5pt}k}^\textrm{\upshape input},
    \end{align*}
    where $A^\textrm{\upshape state}_{\ell\hspace{-0.5pt} \mid \hspace{-0.5pt}k}$ and $B^\textrm{\upshape input}_{\ell\hspace{-0.5pt} \mid \hspace{-0.5pt}k}$ are the corresponding transition and input matrices. See~\cite{faulwasser2020toward} for a discussion of multiperiod AC \acrshort{opf} in this framework.
\end{myRemark}

In contrast, under the TSO-managed model, the \acrshort{tso} has full access to individual DER dispatch and storage states (Section~\ref{sec::formulation::managmentModels}). The temporal link then shifts to the coupling variables themselves, replacing~\eqref{eq::mopt::original::time} with,
\begin{flalign}\label{eq::mopt::original::time::global}
    M_{\ell\hspace{-0.5pt} \mid \hspace{-0.5pt}k+1}\,x_{\ell\hspace{-0.5pt} \mid \hspace{-0.5pt}k+1}= \, A_{\ell\hspace{-0.5pt} \mid \hspace{-0.5pt}k}\, x_{\ell\hspace{-0.5pt} \mid \hspace{-0.5pt}k},\qquad\forall\ell\in\mathcal{S}^{\text{D}} ,\;\forall k\in\mathcal{K},
\end{flalign}
so that state transitions are directly embedded in the coupling variables.
\begin{myRemark}
In both cases, $M_{\ell\hspace{-0.5pt} \mid \hspace{-0.5pt}k}$ may vary over time to capture changing system dynamics or plug‐and‐play DER behavior. 
\end{myRemark}
%In the remainder of this section we introduce a temporal decomposition strategy that keeps the local state transition intact while still exploiting period-wise parallelism, thereby enabling tractable multiperiod flexibility aggregation.

%the flexibility aggregation procedure from the previous sections extends directly: the implicit feasible set $\mathcal{X}^\text{imp}{\ell\hspace{-0.5pt} \mid \hspace{-0.5pt}k}$ of each period can be approximated independently, and the temporal coupling~\eqref{eq::mopt::original::time::global} links these sets across time.

%When temporal consistency is imposed on the coupling variables $x_{\ell\hspace{-0.5pt} \mid \hspace{-0.5pt}k}$, the flexibility aggregation procedure from the previous sections can be directly applied: implicit feasible set of each time period $\mathcal{X}^\text{imp}_{\ell\hspace{-0.5pt} \mid \hspace{-0.5pt}k}$ is approximated independently, and the temperal coupling~\eqref{eq::mopt::original::time::global} then ties the sets together. Under the DSO-managed model, however, the temperal coupling applied to the local variables $y_{\ell\hspace{-0.5pt} \mid \hspace{-0.5pt}k}$. The implicit feasible set $\mathcal{X}^\text{imp}_{\ell\hspace{-0.5pt} \mid \hspace{-0.5pt}k}$ at period $k$ therefore depends on the realised local variables $y_{\ell\hspace{-0.5pt} \mid \hspace{-0.5pt}k-1}$, preventing a fully parallel approximation. 

\subsection{Aggregation under Local Temporal Coupling}\label{sec::temperal::coupling}

When temporal consistency is imposed solely on the coupling variables $x_{\ell\hspace{-0.5pt} \mid \hspace{-0.5pt}k}$, i.e., TSO-managed model, the flexibility aggregation procedure introduced in Section~\ref{sec::predictor-corrector} extends naturally to the multiperiod setting. For each distribution system $\ell\in\mathcal{S}^\text{D}$, the implicit feasible set at time step $k\in\mathcal{K}$ is given by
$$\mathcal{X}^\text{imp}_{\ell\hspace{-0.5pt} \mid \hspace{-0.5pt}k} = \left\{x_{\ell\hspace{-0.5pt} \mid \hspace{-0.5pt}k}\in\mathcal{X}_{\ell\hspace{-0.5pt} \mid \hspace{-0.5pt}k}\,\mid
    \,g_{\ell\hspace{-0.5pt} \mid \hspace{-0.5pt}k}(x_{\ell\hspace{-0.5pt} \mid \hspace{-0.5pt}k},y_{\ell\hspace{-0.5pt} \mid \hspace{-0.5pt}k})=0,\,y_{\ell\hspace{-0.5pt} \mid \hspace{-0.5pt}k}\in\mathcal{Y}_{\ell\hspace{-0.5pt} \mid \hspace{-0.5pt}k}\right\},
$$
These sets can be approximated independently at each time step, while the temporal coupling~\eqref{eq::mopt::original::time::global} enforces consistency across the horizon.

Under a DSO-managed model, however, the temporal coupling applies instead to the local variables $y_{\ell\hspace{-0.5pt} \mid \hspace{-0.5pt}k}$. In this case, the implicit feasible set $\mathcal{X}^\text{imp}_{\ell\hspace{-0.5pt} \mid \hspace{-0.5pt}k}$ at time $k$ depends explicitly on the local state $y_{\ell\hspace{-0.5pt} \mid \hspace{-0.5pt}k-1}$, preventing a fully parallel aggregation.

To address this, we adopt a temporal decomposition that preserves local state transitions while retaining period-wise parallelism. The key idea is to interpret each distribution grid as a \textit{virtual storage} device. At each time step $k\in\mathcal{K}$, the coupling variables for distribution system $\ell\in\mathcal{S}^\text{D}$ are defined as
$$
x_{\ell\hspace{-0.5pt} \mid \hspace{-0.5pt}k}= (p^{\text{pcc}}_{\ell\hspace{-0.5pt} \mid \hspace{-0.5pt}k},\; e_{\ell\hspace{-0.5pt} \mid \hspace{-0.5pt}k}^\text{pcc})\in\mathbb{R}^2,
$$
where $e_{\ell\hspace{-0.5pt} \mid \hspace{-0.5pt}k}^\text{pcc}$ denotes the virtual energy level of distribution system $\ell$ at time period $k$.

By assigning time-varying participation factors, analogous to the single-period case in Section~\ref{sec::DSO::management}, we establish a single-valued allocation between this coupling pair $(p^{\text{pcc}}_{\ell\hspace{-0.5pt} \mid \hspace{-0.5pt}k}, e_{\ell\hspace{-0.5pt} \mid \hspace{-0.5pt}k}^\text{pcc})$ and the internal DER power outputs and \acrshort{soc} values. This abstraction yields a virtual storage balance:
%By choosing time-varying participation vectors, we establish a bijective mapping between this coupling pair $(p^{\text{pcc}}_{\ell\hspace{-0.5pt} \mid \hspace{-0.5pt}k},\; e_{\ell\hspace{-0.5pt} \mid \hspace{-0.5pt}k}^\text{pcc})$ and the individual DER power output and \acrshort{soc}. This abstraction allows virtual storage balance equations
\begin{equation}\label{eq::mopt::virtual::storage}
    e^\text{pcc}_{\ell\hspace{-0.5pt} \mid \hspace{-0.5pt}k+1} = e^\text{pcc}_{\ell\hspace{-0.5pt} \mid \hspace{-0.5pt}k} - \Delta t\, p^g_{\ell\hspace{-0.5pt} \mid \hspace{-0.5pt}k}\qquad\forall \ell\in\mathcal{S}^\text{D},\; k\in\mathcal{K},
\end{equation}
\begin{myRemark}[Coupling for visualization]\label{rmk::visualization}
    For clarity, we restrict the coupling vector to $(p^{\text{pcc}}_{\ell\hspace{-0.5pt} \mid \hspace{-0.5pt}k}, e_{\ell\hspace{-0.5pt} \mid \hspace{-0.5pt}k}^\text{pcc})$. Reactive exchanges $q^{\text{pcc}}_{\ell\hspace{-0.5pt} \mid \hspace{-0.5pt}k}$ can be included without altering the structure of the method, but are omitted here to simplify visualization.
\end{myRemark}
{
\begin{myRemark}[Virtual storage and heterogeneity] \label{rem::virtualBattery}
The virtual-storage representation does not assume freely redistributable energy among heterogeneous units. The participation factors are imposed during aggregation, and a PCC-level point is retained only if the reconstructed unit-level powers and \acrshort{soc} values satisfy the device and network constraints in \(\mathcal Y_{\ell|k}\). Hence, the exact feasible set under the prescribed participation policy is conservative relative to unrestricted internal redispatch, while the predictor-corrector approximation may introduce small boundary-classification errors.
\end{myRemark}
}

With this abstraction, multiperiod problems with local temporal coupling can be reformulated as
\begin{subequations}
    \label{eq::mopt::equivalent}
    \begin{align}
        \min_{x,y}\;\sum_{k\in\mathcal{K}}\Big\{ f_{0\hspace{-0.5pt} \mid \hspace{-0.5pt}k}(x_{k},y_{0\hspace{-0.5pt} \mid \hspace{-0.5pt}k})+\sum_{\ell\in\mathcal{S}^\text{D}}f_{\ell\hspace{-0.5pt} \mid \hspace{-0.5pt}k}(x_{\ell\hspace{-0.5pt} \mid \hspace{-0.5pt}k}, y_{\ell\hspace{-0.5pt} \mid \hspace{-0.5pt}k})\Big\}\label{eq::mopt::equivalent::cost}
    \end{align}
    subject to
   \begin{flalign}
        &&g_{0\hspace{-0.5pt} \mid \hspace{-0.5pt}k}(x_{k},y_{0\hspace{-0.5pt} \mid \hspace{-0.5pt}k})\hspace{-3pt}=&\, 0,\; y_{0\hspace{-0.5pt} \mid \hspace{-0.5pt}k}\in\mathcal{Y}_{0\hspace{-0.5pt} \mid \hspace{-0.5pt}k},&\forall k\hspace{-2.5pt}\in\hspace{-2.5pt}\mathcal{K}    \label{eq::mopt::equivalent::master}\\&&
        \hspace{-3pt}g_{\ell\hspace{-0.5pt} \mid \hspace{-0.5pt}k}(x_{\ell\hspace{-0.5pt} \mid \hspace{-0.5pt}k},y_{\ell\hspace{-0.5pt} \mid \hspace{-0.5pt}k})\hspace{-3pt}=& \, 0,&\forall\ell\hspace{-2.5pt}\in\hspace{-2.5pt}\mathcal{S}^{\text{D}} ,\forall k\hspace{-2.5pt}\in\hspace{-2.5pt}\mathcal{K} \label{eq::mopt::equivalent::worker}\\&&
        (x_{\ell\hspace{-0.5pt} \mid \hspace{-0.5pt}k},y_{\ell\hspace{-0.5pt} \mid \hspace{-0.5pt}k})\hspace{-3pt}\in&\, \mathcal{X}_{\ell\hspace{-0.5pt} \mid \hspace{-0.5pt}k}\times \mathcal{Y}_{\ell\hspace{-0.5pt} \mid \hspace{-0.5pt}k}, \hspace{-7pt} & \forall\ell\hspace{-2.5pt}\in\hspace{-2.5pt}\mathcal{S}^{\text{D}} ,\forall k\hspace{-2.5pt}\in\hspace{-2.5pt}\mathcal{K}  \label{eq::mopt::equivalent::set}\\&&
        M_{\ell\hspace{-0.5pt} \mid \hspace{-0.5pt}k+1}x_{\ell\hspace{-0.5pt} \mid \hspace{-0.5pt}k+1}\hspace{-3pt}=&\,   A_{\ell\hspace{-0.5pt} \mid \hspace{-0.5pt}k} x_{\ell\hspace{-0.5pt} \mid \hspace{-0.5pt}k}\hspace{-3pt}+\hspace{-3pt}B_{\ell\hspace{-0.5pt} \mid \hspace{-0.5pt}k} y_{\ell\hspace{-0.5pt} \mid \hspace{-0.5pt}k},&\forall\ell\hspace{-2.5pt}\in\hspace{-2.5pt}\mathcal{S}^{\text{D}},\forall k\hspace{-2.5pt}\in\hspace{-2.5pt}\mathcal{K}  \label{eq::mopt::equivalent::time} 
    \end{flalign}
\end{subequations}
Applying the predictor-corrector aggregation method (Section~\ref{sec::predictor-corrector}) to approximate the spatio-temporal decomposed constraints~\eqref{eq::mopt::equivalent::worker}–\eqref{eq::mopt::equivalent::set} yields the following upper-level problem:
\begin{subequations}
    \label{eq::mopt::aggregated}
    \begin{align}
        \min_{x,y}\;\sum_{k\in\mathcal{K}}\Big\{ f_{0\hspace{-0.5pt} \mid \hspace{-0.5pt}k}(x_{k},y_{0\hspace{-0.5pt} \mid \hspace{-0.5pt}k})+\sum_{\ell\in\mathcal{S}^\text{D}} \hat{f}_{\ell\hspace{-0.5pt} \mid \hspace{-0.5pt}k}(x_{\ell\hspace{-0.5pt} \mid \hspace{-0.5pt}k})\Big\}\label{eq::mopt::aggregated::cost}
    \end{align}
    subject to
   \begin{flalign}
        &&g_{0\hspace{-0.5pt} \mid \hspace{-0.5pt}k}(x_{k},y_{0\hspace{-0.5pt} \mid \hspace{-0.5pt}k})\hspace{-2pt}=&\, 0,\; y_{0\hspace{-0.5pt} \mid \hspace{-0.5pt}k}\in\mathcal{Y}_{0\hspace{-0.5pt} \mid \hspace{-0.5pt}k},&\forall k\hspace{-2.5pt}\in\hspace{-2.5pt}\mathcal{K}    \label{eq::mopt::aggregated::master}\\&&
        x_{\ell\mid k} \hspace{-2pt}\in&\,\hat{\mathcal{X}}^\text{imp}_{\ell\hspace{-0.5pt} \mid \hspace{-0.5pt}k},&\forall\ell\hspace{-2.5pt}\in\hspace{-2.5pt}\mathcal{S}^{\text{D}} ,\forall k\hspace{-2.5pt}\in\hspace{-2.5pt}\mathcal{K} \label{eq::mopt::aggregated::implicitß}\\&&
        M_{\ell\hspace{-0.5pt} \mid \hspace{-0.5pt}k+1}x_{\ell\hspace{-0.5pt} \mid \hspace{-0.5pt}k+1}\hspace{-1pt}\hspace{-2pt}=&\,   A_{\ell\hspace{-0.5pt} \mid \hspace{-0.5pt}k} x_{\ell\hspace{-0.5pt} \mid \hspace{-0.5pt}k}\hspace{-3pt}+\hspace{-3pt}B_{\ell\hspace{-0.5pt} \mid \hspace{-0.5pt}k} \hat{y}_{\ell\hspace{-0.5pt} \mid \hspace{-0.5pt}k},&\forall\ell\hspace{-2.5pt}\in\hspace{-2.5pt}\mathcal{S}^{\text{D}} ,\forall k\hspace{-2.5pt}\in\hspace{-2.5pt}\mathcal{K}  \label{eq::mopt::aggregated::time}
    \end{flalign}
with the approximated implicit feasible sets
   \begin{align}
    \hat{\mathcal{X}}_{\ell\hspace{-0.5pt} \mid \hspace{-0.5pt}k}^{\text{imp}}= \left\{x_{\ell\hspace{-0.5pt} \mid \hspace{-0.5pt}k}\in\mathcal{X}
    _{\ell\hspace{-0.5pt} \mid \hspace{-0.5pt}k}\;\Big| \;\hat{y}_{\ell\hspace{-0.5pt} \mid \hspace{-0.5pt}k}(x_{\ell\hspace{-0.5pt} \mid \hspace{-0.5pt}k})\in\mathcal{Y}
    _{\ell\hspace{-0.5pt} \mid \hspace{-0.5pt}k}\right\},
    \end{align}
\end{subequations}
defined for every distribution system $\ell\in\mathcal{S}^\text{D}$ and each time step $k\in\mathcal{K}$.

\subsection{Results from Real-World Scenarios}\label{sec::RWScenarios}
While Table~\ref{tab::aggregation-results} benchmarks diverse networks in single-period settings, practical transmission--distribution coordination also involves intertemporal coupling from storage \acrshort{der}s and time-varying demand. We therefore evaluate the proposed method in a dynamic \acrshort{itd} system comprising the IEEE 118-bus transmission grid and 28 heterogeneous distribution networks, for a total of 3,564 buses.
The distribution systems span different sizes and characteristics, including \texttt{case10ba}, \texttt{case533mt}, \texttt{case118zh}, and the real-world KIT Campus North feeder (\texttt{KIT-CN}), which provides measured load and solar-generation profiles~\cite{gonzalez2021probabilistic}.

% yl{For the multiperiod coordination, the distribution networks are connected to transmission buses according to the original active-load levels of case118, ensuring a representative matching between transmission-bus demand and distribution-network size. To evaluate both radial and meshed operating conditions, cases with initially inactive branches are considered half in their original topology and the other half in a fully connected topology obtained by activating all branches. The 24-hour demand trajectories are generated by parameterising the loads of all test cases with normalised KIT Campus North load profiles. multiperiod operation is modelled by augmenting each DER with an energy storage model, introducing inter-temporal energy balance and state-of-charge constraints. The detailed parameter settings for the construction are provided on the accompanying \textbf{website}.}

For evaluation, we first aggregate the flexibility of each distribution system independently at each hourly period. These flexibility sets are then incorporated into a transmission-level problem with temporal coupling constraints~\eqref{eq::mopt::aggregated}, enabling a direct comparison with the original centralized multiperiod formulation~\eqref{eq::mopt::original}.
We assessed the aggregation approaches using two performance metrics:
\begin{itemize}
    \item \textbf{optimality gap}, defined as the deviation in objective value between the aggregated formulation~\eqref{eq::mopt::aggregated} and the centralized optimization problem~\eqref{eq::mopt::original}; and
    \item \textbf{constraints violations}, obtained by testing whether the aggregated dispatch decisions satisfy the original nonlinear distribution constraints~\eqref{eq::mopt::original::worker}-\eqref{eq::mopt::original::set}.
\end{itemize}

\begin{figure*}[htbp!]
    \centering
    % the first row
    \begin{subfigure}{0.32\textwidth}
        \centering
        \caption{KIT North Campus}
        \includegraphics[height=4cm]{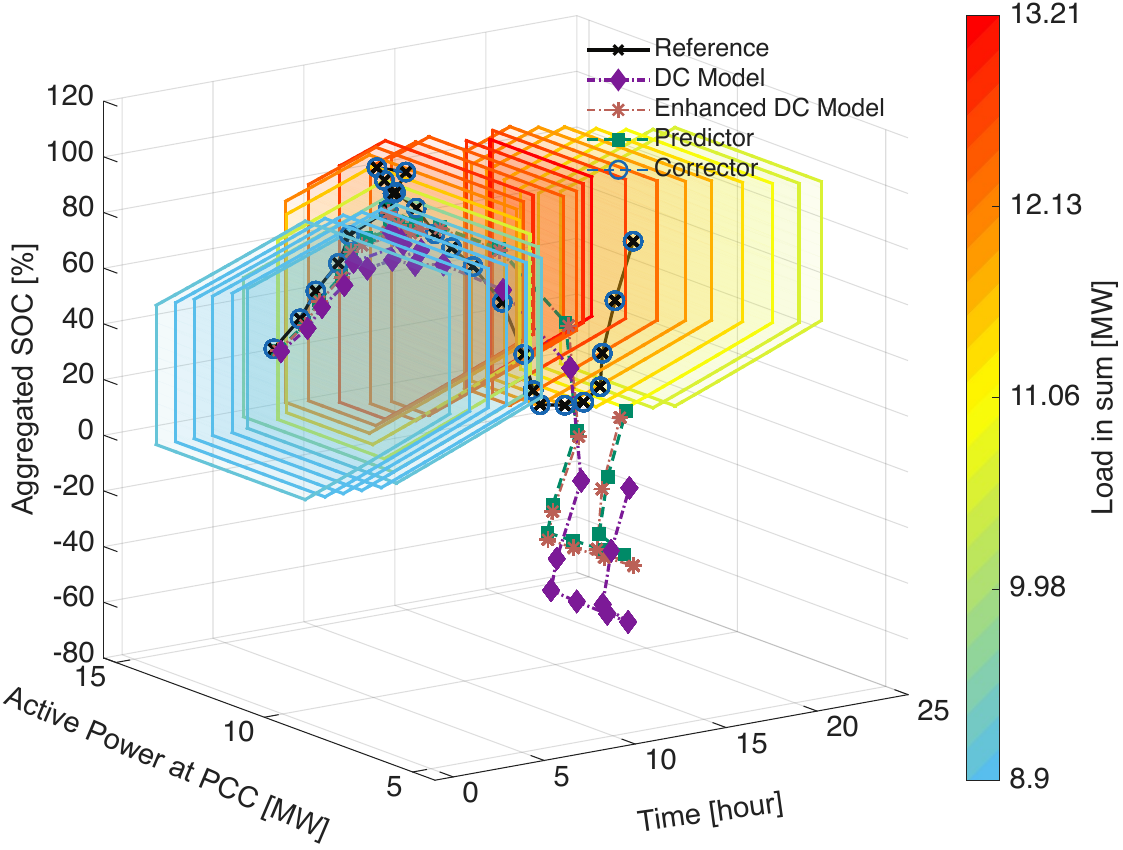}
        \label{fig:sub1}
    \end{subfigure}
    \hfill
    \begin{subfigure}{0.32\textwidth}
        \centering
        \caption{case118zh}
        \includegraphics[height=4cm]{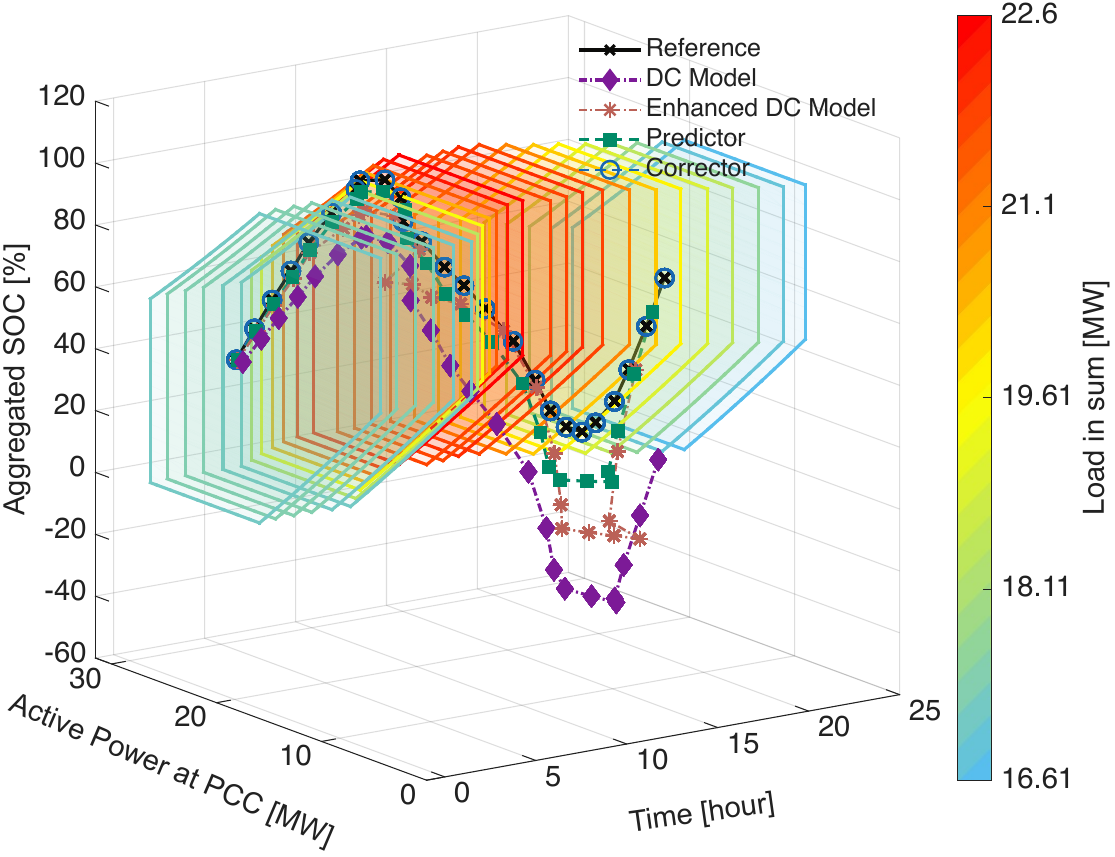}
        \label{fig:sub2}
    \end{subfigure}
    \hfill
    \begin{subfigure}{0.32\textwidth}
        \centering
        \caption{case533mt}
        \includegraphics[height=4cm]{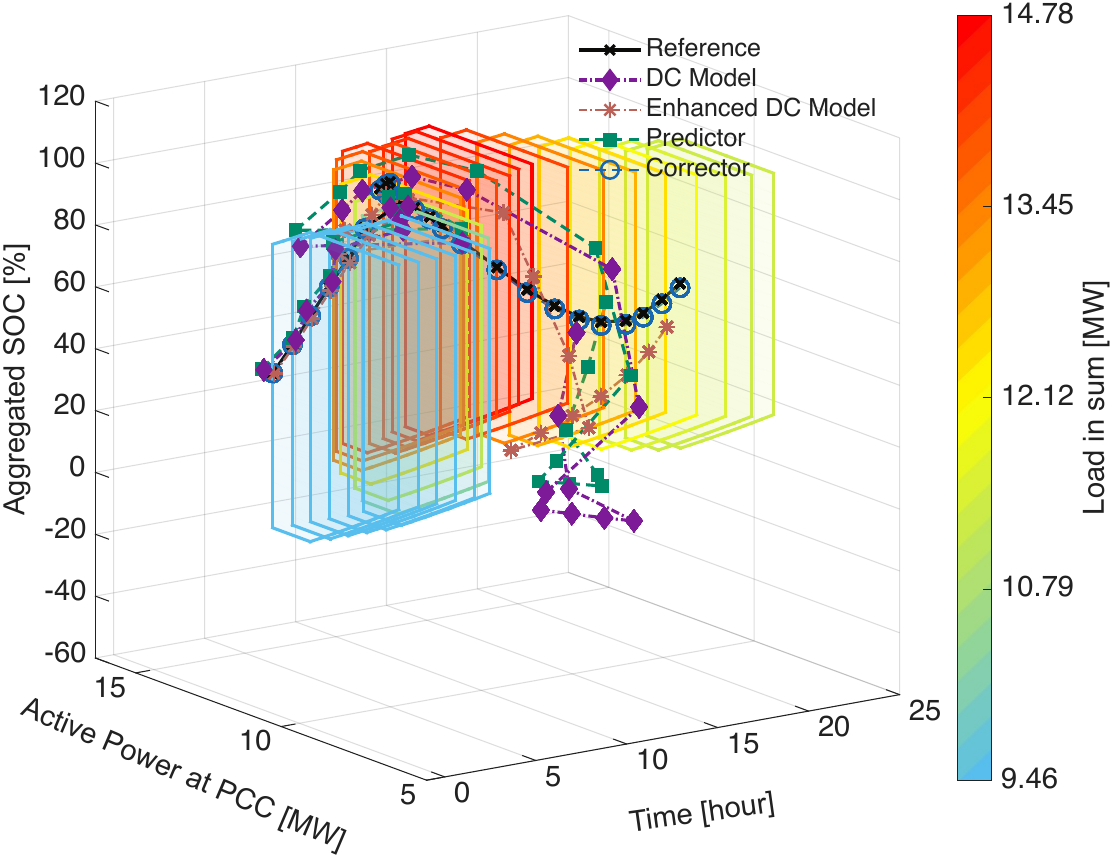}
        \label{fig:sub3}
    \end{subfigure}

    \vspace{0.5cm} % add the vertical space

    \caption{Comparison of optimal dispatch trajectories obtained from different aggregation methods: DC model (diamonds), enhanced DC model (stars), tangential predictor (squares), and the proposed predictor–corrector method (circles). The centralized solution with the exact AC power flow model is shown for reference (crosses). The background regions depict the time-decoupled implicit feasible sets generated by the predictor–corrector method, where red areas correspond to high-demand periods around noon and blue areas correspond to low-demand periods around midnight.}
    \label{fig::comp::trajectory}\vspace{-10pt}
\end{figure*}

Fig.~\ref{fig::comp::trajectory} illustrates the optimal dispatch trajectories for three representative distribution networks: \texttt{KIT-CN}, \texttt{case118zh}, and \texttt{case533mt}. The trajectories correspond to the proposed predictor–corrector aggregation (circles) and three linear surrogates, while the centralized AC solution to~\eqref{eq::mopt::original} serves as the reference (crosses). The background regions depict the time-decoupled implicit feasible sets generated by the predictor–corrector method, with color intensity indicating the level of total system demand. Across all three cases, the predictor–corrector trajectories align almost exactly with the reference solutions.

\begin{figure}[htbp!]
    \centering

    \begin{subfigure}[b]{0.25\textwidth}
        \caption{Power Domain}
        \includegraphics[width=\textwidth]{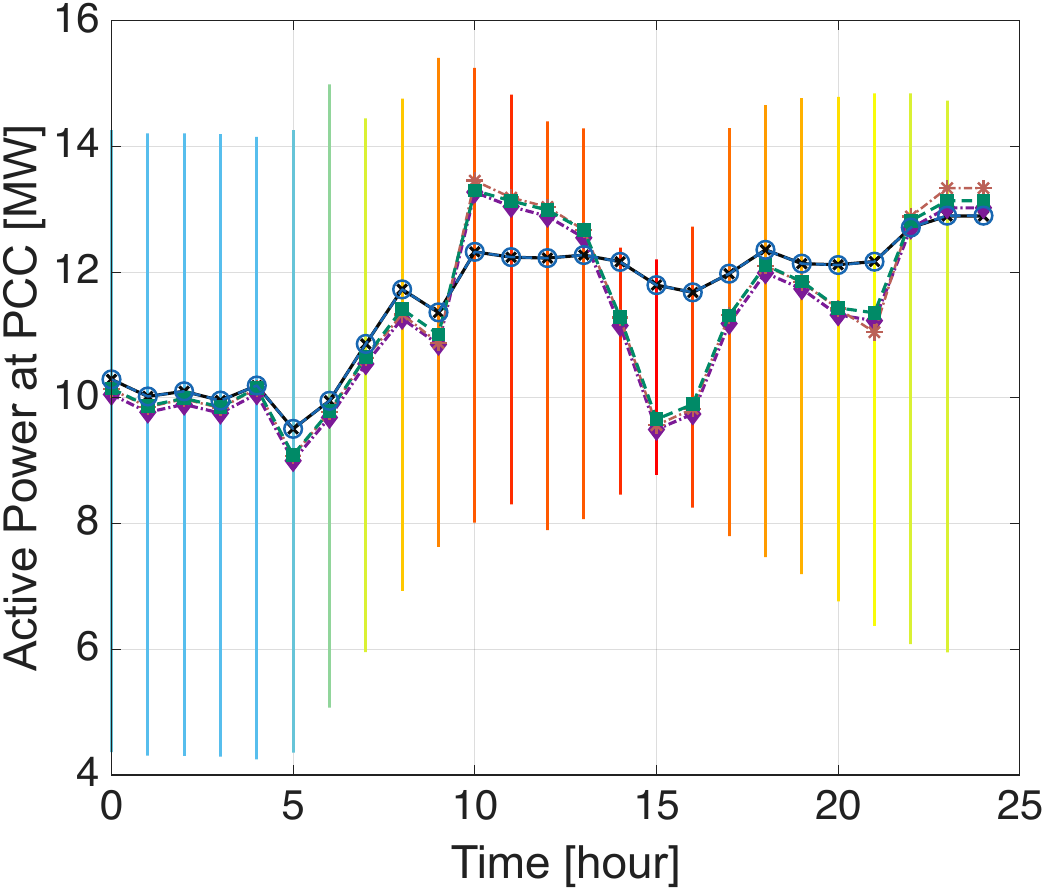}
        \label{fig::PT::projection}
    \end{subfigure}
    \hspace{50pt}
    \begin{subfigure}[b]{0.256\textwidth}
        \caption{Energy Domain}
        \includegraphics[width=\textwidth]{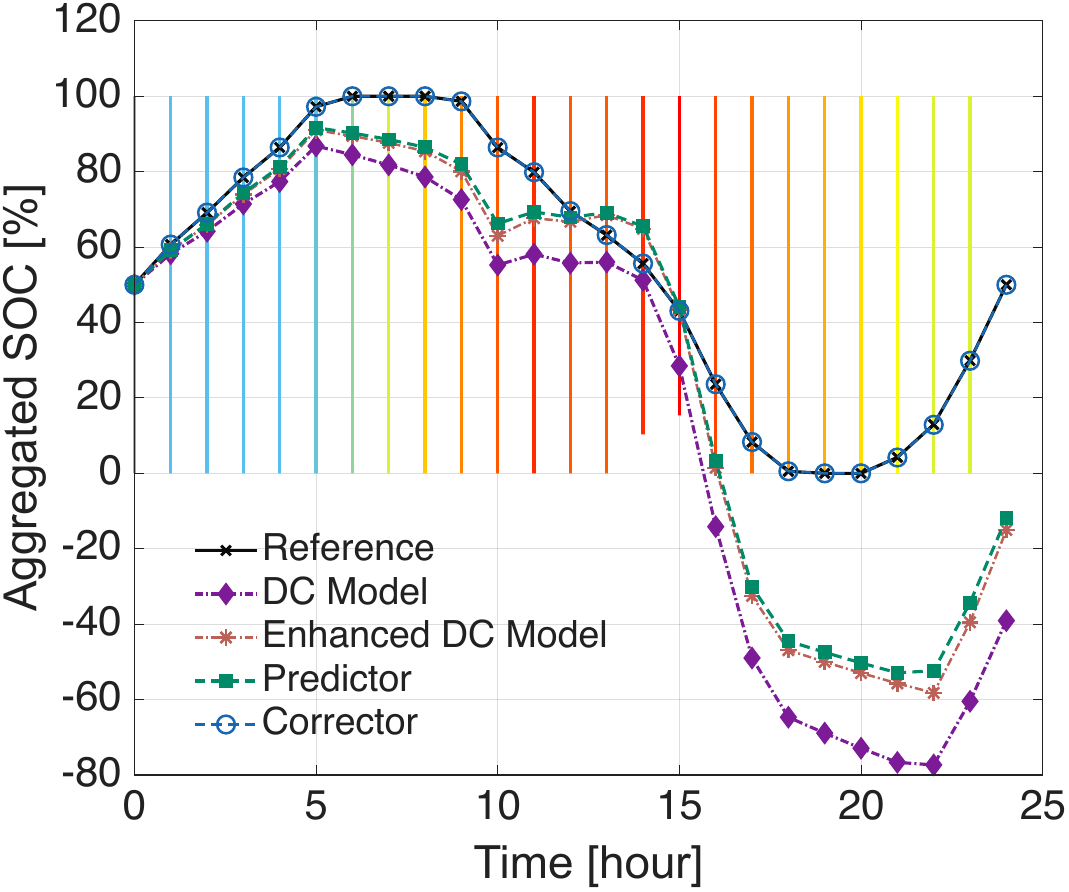}
        \label{fig::ET::projection}
    \end{subfigure}
    \caption{Comparison of optimal dispatch trajectories in the power–energy domain.}
    \label{fig::power::energy}
\end{figure}

\begin{figure}[htbp!]
    \centering
    \begin{subfigure}[b]{0.65\textwidth}
        \centering
        
        \caption{Optimality gap}
        \includegraphics[width=\textwidth]{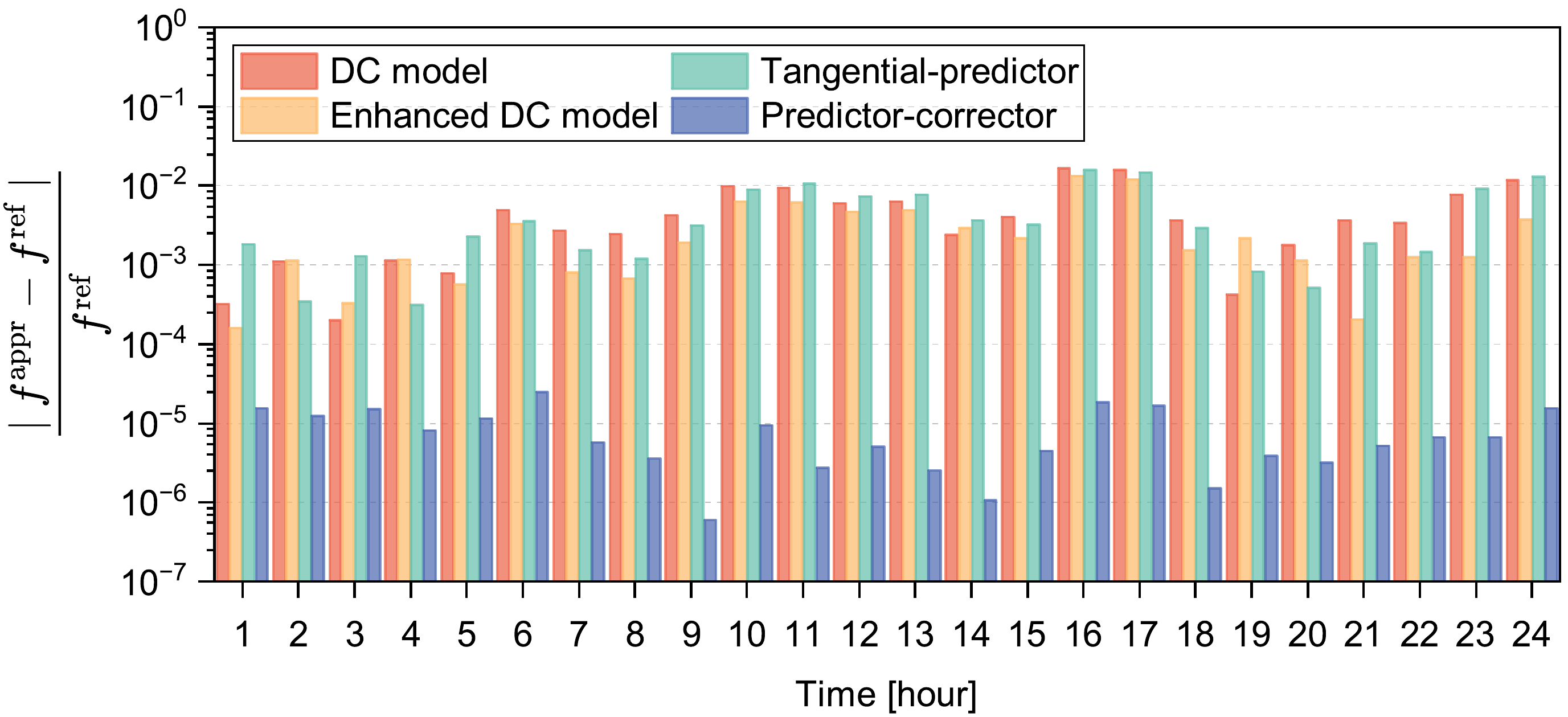}
        \label{fig::scheduling::cost}
    \end{subfigure}\\
    %\begin{subfigure}[b]{0.48\textwidth}
    %    \captionsetup{justification=raggedright, singlelinecheck=false,  margin={3em,0em}}
    %    \centering
    %    \vspace{-2ex}
    %    \caption{Redispatch between different aggregation methods}
    %    \includegraphics[width=\textwidth]{images/redispatch.pdf}
    %    \label{fig::redispatch}
    %\end{subfigure}
    \begin{subfigure}[b]{0.65\textwidth}
        \centering  
        \caption{Constraints violation}
        \includegraphics[width=\textwidth]{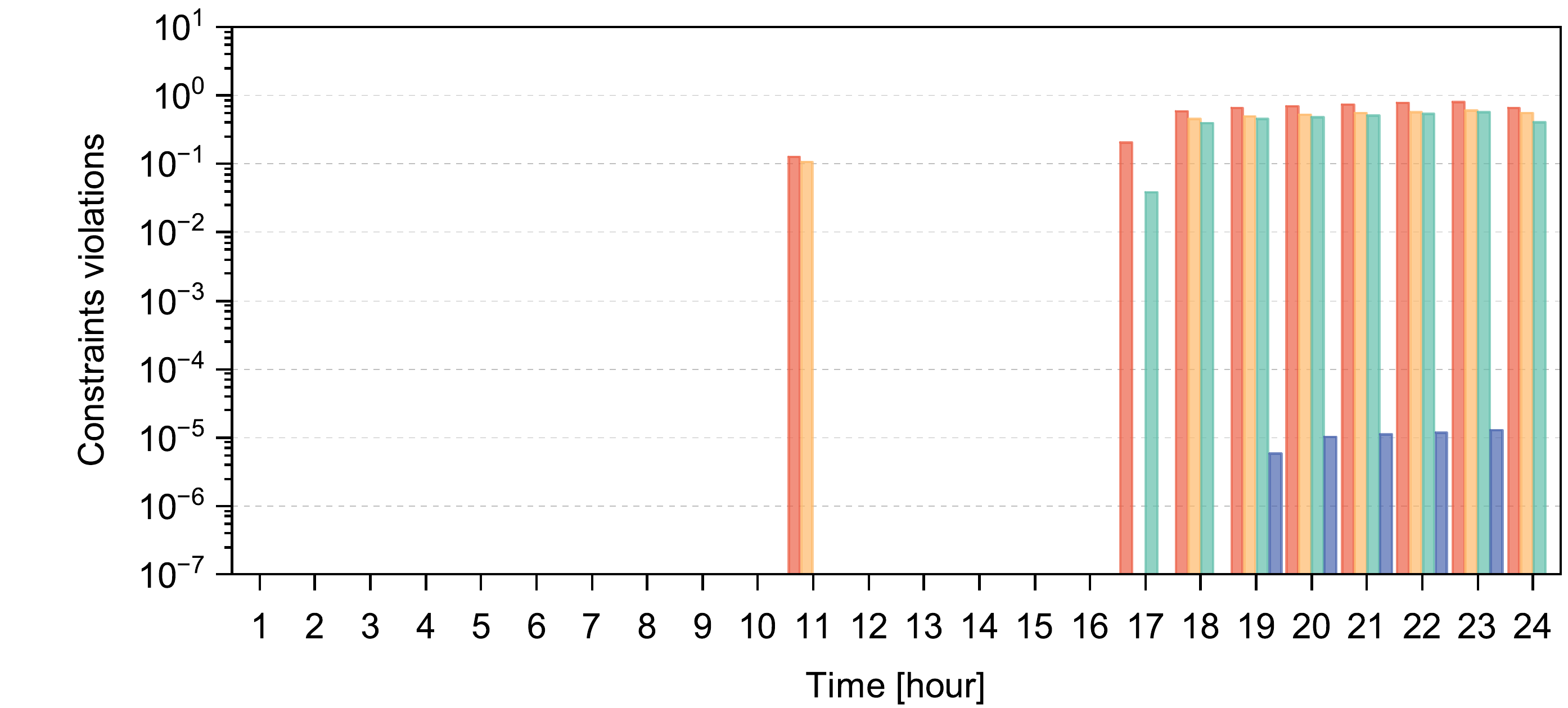}
        \label{fig::scheduling::violations}    
    \end{subfigure}
    \caption{Performance matrices}
    \label{fig::scheduling}
\end{figure}

A complementary comparison in the power–energy domain is also shown in Fig.~\ref{fig::power::energy}. While the linear surrogates can reasonably approximate the optimal trajectories in the power domain, their trajectories in the energy domain exhibit increasing deviations over time. This behavior, consistent with observations in~\cite{jiang2025enhanced,jiang2025error}, reflects the accumulation of approximation errors inherent to linear models.

Detailed hourly evaluations of optimality gaps and constraint violations are reported in Fig.~\ref{fig::scheduling}. The predictor–corrector aggregation method consistently tracks the reference solution with negligible cost deviation. Over the entire day, only some minor voltage violations occur in the \texttt{KIT-CN} system, each less than $10^{-5}$ p.u.

Overall, these results confirm that the predictor–corrector aggregation provides an accurate and reliable representation of nonlinear system behavior, enabling robust flexibility characterization and secure TSO–DSO coordination.

%\begin{figure}[htbp]
%    \centering
%    \begin{subfigure}[b]{0.48\textwidth}
%        \centering
%        \caption{Cost comparison between predictor and corrector}
%        \hspace{-2mm}  
%        \includegraphics[width=\textwidth]{images/cost_radialmesh_mixed.pdf}
%        \label{fig::scheduling::lds}
%    \end{subfigure}
%    \begin{subfigure}[b]{0.48\textwidth}
%        \centering
%        \vspace{-2ex}
%        \caption{Violation comparison between predictor and corrector}
%        \includegraphics[width=\textwidth]{images/violation_radialmesh_mixed.pdf}
%        \label{fig::scheduling::pred}
%    \end{subfigure}
%    \vspace{-2ex}
%    \caption{Multiperiod coordination comparison: In this simulation, multiple distribution networks (case10ba, case of KIT north campus, case118zh, case533mt) are connected to one transmission network. For each of the last three distribution networks, half have been switched to run in a meshed pattern, for which the LinDistFlow model is not applicable. Therefore, the coordination with the predictor and corrector is compared.}
%    \label{fig::scheduling}
%\end{figure}

\subsection{Computation Time: Centralized vs. Aggregated}

In this section, we identify and quantify the sources of the computational speedup achieved by the proposed aggregated approach. In addition to the IEEE 118-bus transmission system considered in Section~\ref{sec::RWScenarios}, we consider a smaller ITD instance based on the IEEE 14-bus transmission system. As the centralized benchmark, the original problem~\eqref{eq::mopt::original} is solved directly using \ipopt~\cite{wachter2006implementation}. In the aggregated approach, the predictor-corrector method first constructs approximations of the spatio-temporally decomposed implicit feasible sets associated with~\eqref{eq::mopt::equivalent::worker}-\eqref{eq::mopt::equivalent::set}, after which the resulting upper-level problem~\eqref{eq::mopt::aggregated} is solved using the same solver. Thus, the comparisons within each test case reflect the computational effect of the aggregation reformulation rather than differences in solver choice.
\begin{table}[htbp!]
    \centering
    \caption{Performance of centralized and aggregated approaches}\label{tab::computing_time::aggregation}
    \setlength{\extrarowheight}{-1pt}
    {
  \begin{tabular}{@{}lcrrrr@{}}
        \toprule
        %Transmission} & Approach & Aggregation &  Initialization & Solving & Total\\
        Case & Approach & AAggregation [s] & Initialization [s] & Solve [s] & Total [s]\\
        \midrule
        \multirow{3}{*}{\texttt{case14}}  &Centralized & -- & 2.12 & 4.94 &7.06\\
        & Aggregated  & 0.15 & 0.27 & 0.73 & 1.15\\
        & Speedup     & --   & 7.85$\times$ & 6.77$\times$ & 6.14$\times$ \\
        \midrule
        \multirow{3}{*}{\texttt{case118}}  &Centralized & -- & 37.33 & 20.01 &57.34\\
        & Aggregated & 0.15 & 4.11 & 4.16 &8.42\\
        & Speedup     & --   & 9.08$\times$ & 4.81$\times$ & 6.81$\times$ \\
        \bottomrule
    \end{tabular}\vspace{-4pt}}
\end{table}

Table~\ref{tab::computing_time::aggregation} separates the wall-clock time into predictor-corrector aggregation, \casadi~\cite{andersson2019casadi} initialization, including model construction and automatic differentiation, and the subsequent \ipopt solve. Including all three stages, the aggregated formulation reduces the total time from \(7.06\) to \(1.15\) seconds for \texttt{case14} and from \(57.34\) to \(8.42\) seconds for \texttt{case118}, corresponding to end-to-end speedups of \(6.14\) and \(6.81\), respectively. The initialization speedups of \(7.85\) and \(9.08\) demonstrate that dimensionality reduction lowers not only the cost of the iterative \ipopt solve but also the cost of model construction and automatic differentiation in \casadi. The sequential aggregation overhead is only \(0.15\) seconds in both cases and can be reduced further by executing the independent subsystem-period aggregation tasks in parallel.
\begin{table}[htbp!]
    \centering
    \caption{Solving time of centralized and aggregated problems}\label{tab::computing_time::comparison}
    \setlength{\extrarowheight}{-1pt}
  \begin{tabular}{@{}llccr@{}}
        \toprule
        Test case\hspace{10pt} & Metric & Centralized & Aggregated & Speedup \\
        \midrule
        \multirow{4}{*}{\texttt{case14}}  &Iterations & 31 & 19 & 1.63$\times$\\
                &Linear Solver [s] & 3.65 & 0.56 & 6.52$\times$\\
        %PDSystemSolverTotal / iterations & 0.482732 & 0.142895 \\  
        &Function Eval. [s]\hspace{15pt} & 0.75 & 0.02 & 37.5$\times$\\
        &Total  [s] & 4.94 & 0.73 & 6.77$\times$\\
        \midrule
        \multirow{4}{*}{\texttt{case118}}  &Iterations & 31 & 22 & 1.41$\times$\\
        %Overalgorithm / iterations & 0.645458 & 0.189164 \\
        & Linear Solver  [s]& 14.96 & 3.14 & 4.76$\times$\\
        %PDSystemSolverTotal / iterations & 0.482732 & 0.142895 \\
        &Function Eval.  [s] & 3.07 & 0.16 & 19.2$\times$\\
        & Total   [s]& 20.01 & 4.16 & 4.81$\times$ \\ 
        %Function Evaluations / iterations & 0.099116 & 0.007477 \\
        \bottomrule
    \end{tabular}\vspace{-10pt}
\end{table}
\begin{figure}[htbp]
    \centering
    \begin{subfigure}
        {0.35\textwidth}
        \centering
        \caption{case14}
        \includegraphics[width=\columnwidth]{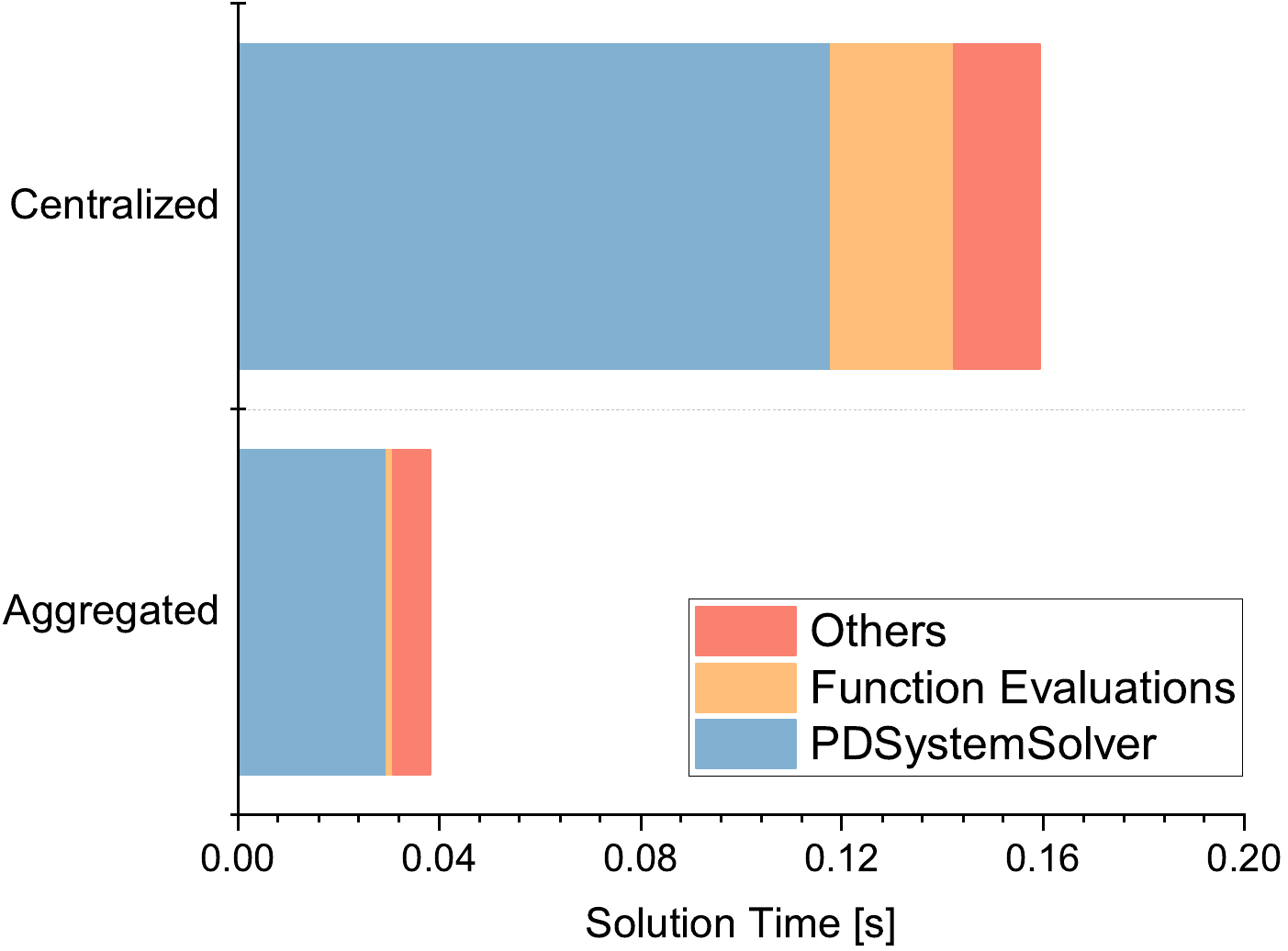}
        \label{fig::timing::case14}
    \end{subfigure}
    \hspace{40pt}
    \begin{subfigure}
        {0.35\textwidth}
        \centering
        \caption{case118}
        \includegraphics[width=\columnwidth]{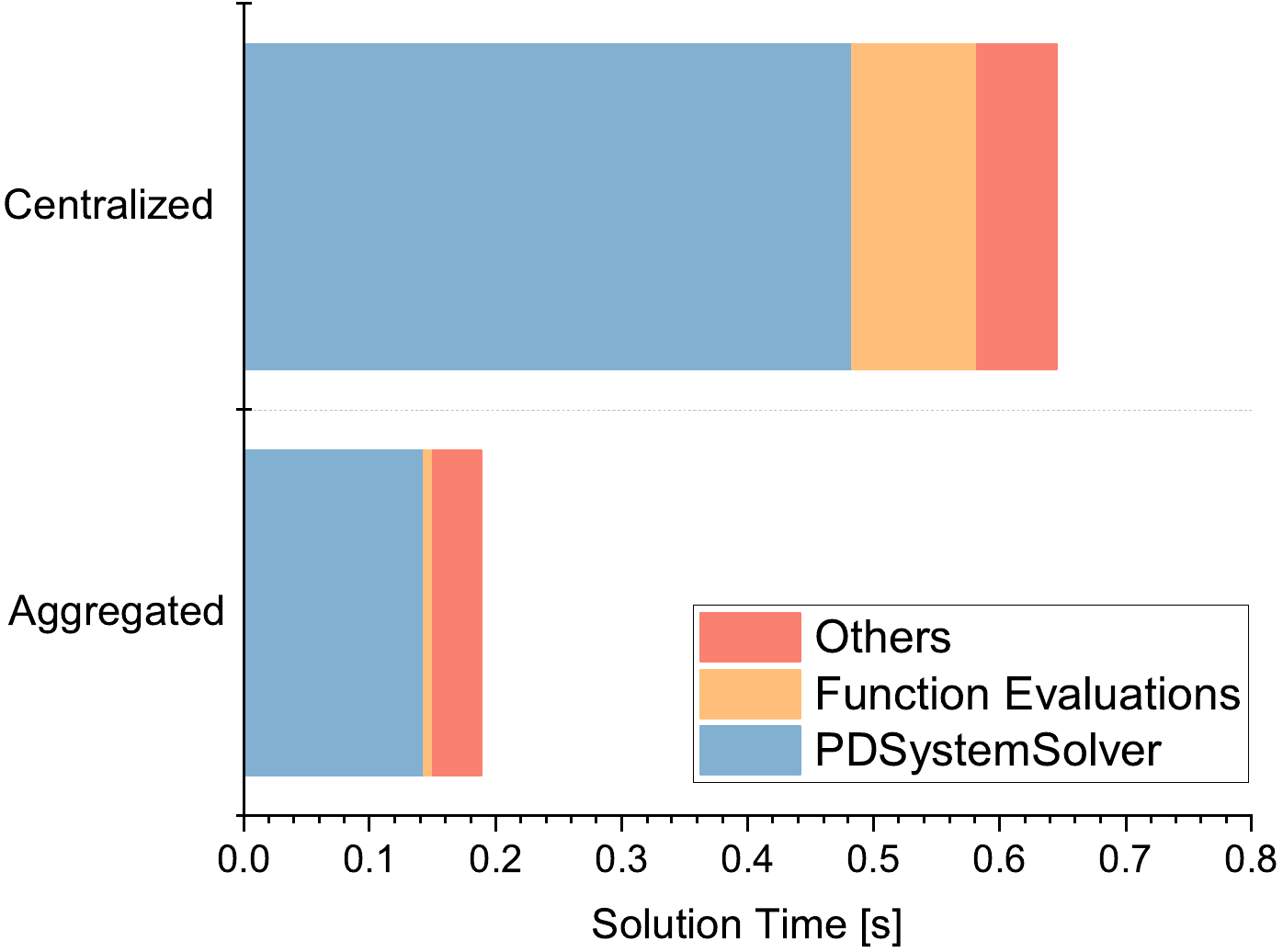}
        \label{fig::timing::case118}
    \end{subfigure}

    \caption{Average \ipopt time per iteration}
    \label{fig::timing::coordination}
\end{figure}

To isolate the upper-level nonlinear-solver stage, Table~\ref{tab::computing_time::comparison} reports the \ipopt timing breakdown for the centralized and aggregated formulations. The aggregated formulation achieves \ipopt-stage speedups of \(4.81\)-\(6.77\), while the relative objective-value difference with respect to the centralized solution and the maximum constraint violation remain on the order of \(10^{-5}\), cf.~Fig.~\ref{fig::scheduling}. It also requires fewer \ipopt iterations in both test cases. Thus, the solver-stage improvement results from both a lower computational cost per iteration and a smaller number of iterations.

\begin{table}[htbp]
    \centering
    \caption{Problem dimensions and KKT structure}
    \label{tab::computing_time::dimension}
    \setlength{\extrarowheight}{-1pt}
   % \begin{threeparttable}
    %\resizebox{\columnwidth}{!}{%
        {
    \begin{tabular}{@{}llrrr@{}}
        \toprule
        Test case & Structural metric
        & Centralized & Aggregated & Reduction \\
        \midrule

        \multirow{5}{*}{\texttt{case14}}
        & No. of variables 
        %& \(16\,656\) & \(1\,680\) & \(89.91\%\) \\
        & \(1.67\mathrm{E}4\) & \(1.68\mathrm{E}3\) & \(89.91\%\) \\

        %& No. of equality  
        %& \(16\,472\) & \(1\,304\) & \multirow{2}{*}{\(92.08\%\)} \\

        %& No. of inequality   & \(0\) & \(7\,776\) &  \\
        & No. of constraints
        %& \(16\,472\) & \(9\,080\) & \(44.88 \%\) \\
        & \(1.65\mathrm{E}4\) & \(9.08\mathrm{E}3\) & \(44.88\%\) \\

        & KKT order 
        %& \(33\,128\) & \(18\,536\) & \(44.05\%\) \\
        & \(3.31\mathrm{E}4\) & \(1.85\mathrm{E}4\) & \(44.05\%\) \\
        & NNZ of KKT
        %& \(282\,816\) & \(54\,816\) & \(80.62\%\) \\
        & \(2.83\mathrm{E}5\) & \(5.48\mathrm{E}4\) & \(80.62\%\) \\
        & KKT density 
        %& \(2.58\mathrm{E}{-4}\) & \(1.60\mathrm{E}{-4}\) & \(38.09\%\) \\
        & \(2.58\mathrm{E}{-4}\) & \(1.60\mathrm{E}{-4}\) & \(38.09\%\) \\
        \midrule

        \multirow{5}{*}{\texttt{case118}}
        & No. of variables 
  %      & \(177\,024\) & \(10\,944\) & \(93.82\%\) \\
        & \(1.77\mathrm{E}5\) & \(1.09\mathrm{E}4\) & \(93.82\%\) \\
        %& No. of equality  
        %& \(174\,508\) & \(7\,756\) & \(95.56\%\) \\

        %& No. of inequality  
        %& \(0\) & \(84\,048\) & -- \\

        & No. of constraints
        %& \(174\,508\) & \(91\,804\) & \(47.39\%\) \\
        & \(1.75\mathrm{E}5\) & \(9.18\mathrm{E}4\) & \(47.39\%\) \\
        & KKT order
        %& \(351\,532\) & \(186\,796\) & \(46.86\%\) \\
        & \(3.52\mathrm{E}5\) & \(1.87\mathrm{E}5\) & \(46.86\%\) \\
        
        & NNZ of KKT
        %& \(3\,199\,104\) & \(555\,888\) & \(82.62\%\) \\
        & \(3.20\mathrm{E}6\) & \(5.56\mathrm{E}5\) & \(82.62\%\) \\
        & KKT density 
        & \(2.59\mathrm{E}{-5}\) & \(1.59\mathrm{E}{-5}\) & \(38.46\%\) \\

        \bottomrule
    \end{tabular}}
    %}
    %\begin{tablenotes}[flushleft]
        %\item[*] \scriptsize
        %Each inequality is reformulated using one bounded slack variable.
        %The KKT nonzero counts assume that all Hessian and slack-variable
        %diagonal entries are structurally nonzero. Density is computed for
        %the full symmetric KKT matrix before numerical-factorization fill-in.
    %\end{tablenotes}
    %\end{threeparttable}
\end{table}

To explain the lower per-iteration cost, Fig.~\ref{fig::timing::coordination} presents the average \ipopt timing breakdown and identifies the primal--dual KKT solve, reported under \texttt{PDSystemSolver}, as the dominant computational component. As discussed in~\cite[Sec.~5]{wachter2009shorttutorial}, \ipopt reformulates general inequality constraints using bounded slack variables. After this reformulation, Table~\ref{tab::computing_time::dimension} shows that aggregation reduces the KKT-system order by \(44.05\%\)--\(46.86\%\) and its structural number of nonzeros by \(80.62\%\)--\(82.62\%\). The KKT density also decreases by approximately \(38\%\). These structural reductions result in \(4.76\)--\(6.52\) reductions in linear-system solution time.

Overall, dimensionality reduction is the primary source of the computational improvement, as it reduces both \casadi initialization costs and the cost of each \ipopt iteration. The spatio-temporal decomposition provides a complementary benefit by enabling parallel preprocessing, although the measured sequential aggregation overhead is already small.

\section{Conclusion}\label{sec::conclusion}
This paper proposes a hierarchical optimization framework for \acrfull{itd} operation based on a predictor–corrector aggregation method. By reinterpreting path-following techniques, the method provides tractable surrogate models for distribution subsystems. The resulting aggregation directly accounts for nonlinear AC power flow constraints, is applicable to both radial and meshed distribution networks,  %under the stated regularity assumptions
 and enables scalable pre-computation through spatio-temporal decomposition. Theoretical error bounds are established, and extensive benchmark studies demonstrate that the approach achieves speedups of $5$–$7$ times compared to centralized nonlinear solvers. Analysis of different network configurations further highlights how structural factors shape aggregated flexibility, offering actionable insights for system operation. An open-source toolbox is released to support practical adoption and further study. 
{Future work will extend the proposed framework by incorporating PCC voltage magnitude as an additional coupling variable, modeling transformer voltage control through discrete tap-position decisions, and evaluating its performance on larger-scale ITD networks. 
Another interesting direction is to investigate global optimality certification for special cases, for example through convex relaxations with verified zero-gap exactness.}
\section*{Author Statements}
X.D. and Y.J. contributed equally to this work. X.D. developed the methodology and algorithms. Y.J. implemented the computational framework and performed the numerical experiments. F.Z., Y.G., and V.H. contributed to the conceptual development, interpretation of the results, and revision of the manuscript. Corresponding: Y.G. 

\bibliography{root}

@article{naik2025variable,
  title={Variable aggregation for nonlinear optimization problems},
  author={Naik, Sakshi and Biegler, Lorenz and Bent, Russell and Parker, Robert},
  journal={arXiv preprint arXiv:2502.13869},
  year={2025}
}

@article{dutta2020topology,
  title={Topology tracking for active distribution networks},
  author={Dutta, Rajarshi and Chakrabarti, Saikat and Sharma, Ankush},
  journal={IEEE Transactions on Power Systems},
  volume={36},
  number={4},
  pages={2855--2865},
  year={2020},
  publisher={IEEE}
}

@article{pacaud2025sensitivity,
  title={Sensitivity analysis for parametric nonlinear programming: A tutorial},
  author={Pacaud, Fran{\c{c}}ois},
  journal={arXiv preprint arXiv:2504.15851},
  year={2025}
}

@article{karthikeyan2019activeDistr,
  title={Predictive control of flexible resources for demand response in active distribution networks},
  author={Karthikeyan, Nainar and Pillai, Jayakrishnan Radhakrishna and Bak-Jensen, Birgitte and Simpson-Porco, John W},
  journal={IEEE Transactions on Power Systems},
  volume={34},
  number={4},
  pages={2957--2969},
  year={2019},
  publisher={IEEE}
}

@article{wang2025non,
  title={Non-Iterative Coordination of Interconnected Power Grids via Dimension-Decomposition-Based Flexibility Aggregation},
  author={Wang, Siyuan and Feng, Cheng and You, Fengqi},
  journal={IEEE Transactions on Power Systems},
  year={2025},
  publisher={IEEE}
}

@article{patari2021distributed,
  title={Distributed optimization in distribution systems: Use cases, limitations, and research needs},
  author={Patari, Niloy and Venkataramanan, Venkatesh and Srivastava, Anurag and Molzahn, Daniel K and Li, Na and Annaswamy, Anuradha},
  journal={IEEE Transactions on Power Systems},
  volume={37},
  number={5},
  pages={3469--3481},
  year={2021},
  publisher={IEEE}
}

@article{dai2025largescale,
  title={Distributed {AC} Optimal Power Flow: A Scalable Solution for Large-Scale Problems},
  author={Dai, Xinliang and Jiang, Yuning and Guo, Yi and Jones, Colin N and Diehl, Moritz and Hagenmeyer, Veit},
  journal={arXiv preprint arXiv:2503.24086},
  year={2025}
}

@article{andersson2019casadi,
  title={CasADi: a software framework for nonlinear optimization and optimal control},
  author={Andersson, Joel AE and Gillis, Joris and Horn, Greg and Rawlings, James B and Diehl, Moritz},
  journal={Mathematical Programming Computation},
  volume={11},
  number={1},
  pages={1--36},
  year={2019},
  publisher={Springer}
}

@article{wachter2006implementation,
  title={On the implementation of an interior-point filter line-search algorithm for large-scale nonlinear programming},
  author={W{\"a}chter, Andreas and Biegler, Lorenz T},
  journal={Mathematical Programming},
  volume={106},
  number={1},
  pages={25--57},
  year={2006},
  publisher={Springer}
}

@article{farivar2013branch,
  title={Branch flow model: Relaxations and convexification ({P}arts {I}, {II})},
  author={Farivar, Masoud and Low, Steven H},
  journal={IEEE Transactions on Power Systems},
  volume={28},
  number={3},
  pages={2554--2564},
  year={2013},
  publisher={IEEE}
}

@article{frank2016introduction,
  title={An introduction to optimal power flow: Theory, formulation, and examples},
  author={Frank, Stephen and Rebennack, Steffen},
  journal={IIE Trans.},
  volume={48},
  number={12},
  pages={1172--1197},
  year={2016},
  publisher={Taylor \& Francis}
}

@article{baran1989optimal1,
  title={Optimal capacitor placement on radial distribution systems},
  author={Baran, Mesut E and Wu, Felix F},
  journal={IEEE Transactions on Power Delivery},
  volume={4},
  number={1},
  pages={725--734},
  year={1989},
  publisher={IEEE}
}

@article{bugosen2023process,
  title={Process flowsheet optimization with surrogate and implicit formulations of a Gibbs reactor},
  author={Bugosen, Sergio I and Laird, Carl D and Parker, Robert B},
  journal={arXiv preprint arXiv:2310.09307},
  year={2023}
}

@article{parker2022implicit,
  title={An implicit function formulation for optimization of discretized index-1 differential algebraic systems},
  author={Parker, Robert and Nicholson, Bethany and Siirola, John and Laird, Carl and Biegler, Lorenz},
  journal={Computers \& Chemical Engineering},
  volume={168},
  pages={108042},
  year={2022},
  publisher={Elsevier}
}

@article{pacaud2024accelerating,
  title={Accelerating condensed interior-point methods on {SIMD/GPU} architectures},
  author={Pacaud, Fran{\c{c}}ois and Shin, Sungho and Schanen, Michel and Maldonado, Daniel Adrian and Anitescu, Mihai},
  journal={Journal of optimization theory and applications},
  volume={202},
  number={1},
  pages={184--203},
  year={2024},
  publisher={Springer}
}

@article{pacaud2022feasible,
  title={A feasible reduced space method for real-time optimal power flow},
  author={Pacaud, Fran{\c{c}}ois and Maldonado, Daniel Adrian and Shin, Sungho and Schanen, Michel and Anitescu, Mihai},
  journal={Electric Power Systems Research},
  volume={212},
  pages={108268},
  year={2022},
  publisher={Elsevier}
}

@article{baran2002network,
  title={Network reconfiguration in distribution systems for loss reduction and load balancing},
  author={Baran, Mesut E and Wu, Felix F},
  journal={IEEE Transactions on Power Delivery},
  volume={4},
  number={2},
  pages={1401--1407},
  year={2002},
  publisher={IEEE}
}

@article{engelmann2025approximate,
  title={Approximate dynamic programming with feasibility guarantees},
  author={Engelmann, Alexander and Bandeira, Ma{\'\i}sa Beraldo and Faulwasser, Timm},
  journal={IEEE Transactions on Control of Network Systems},
  year={2025},
  publisher={IEEE}
}

@article{bandeira2024adp,
  title={An {ADP} framework for flexibility and cost aggregation: Guarantees and open problems},
  author={Bandeira, Maisa Beraldo and Faulwasser, Timm and Engelmann, Alexander},
  journal={Electric Power Systems Research},
  volume={234},
  pages={110818},
  year={2024},
  publisher={Elsevier}
}

@book{johnson1990matrix,
  title={Matrix theory and applications},
  author={Johnson, Charles R},
  volume={40},
  year={1990},
  publisher={American Mathematical Soc.}
}

@article{molzahn2017survey,
  title={A survey of distributed optimization and control algorithms for electric power systems},
  author={Molzahn, Daniel K and D{\"o}rfler, Florian and Sandberg, Henrik and Low, Steven H and Chakrabarti, Sambuddha and Baldick, Ross and Lavaei, Javad},
  journal={IEEE Transactions on Smart Grid},
  volume={8},
  number={6},
  pages={2941--2962},
  year={2017},
  publisher={IEEE}
}

@article{wen2022tdder,
  title={Aggregate Temporally Coupled Power Flexibility of {DERs} Considering Distribution System Security Constraints},
  author={Wen, Yilin and Hu, Zechun and Liu, Likai},
  journal={IEEE Transactions on Power Systems},
  year={2023},
  volume={38},
  number={4},
  pages={3884-3896},
  doi={10.1109/TPWRS.2022.3196708}}

@article{gonzalez2021probabilistic,
  title={Probabilistic forecasts of the distribution grid state using data-driven forecasts and probabilistic power flow},
  author={Gonz{\'a}lez-Ordiano, Jorge {\'A}ngel and M{\"u}hlpfordt, Tillmann and Braun, Eric and Liu, Jianlei and {\c{C}}akmak, H{\"u}seyin and K{\"u}hnapfel, Uwe and D{\"u}pmeier, Clemens and Waczowicz, Simon and Faulwasser, Timm and Mikut, Ralf and others},
  journal={Applied Energy},
  volume={302},
  pages={117498},
  year={2021},
  publisher={Elsevier}
}

@article{faulwasser2020toward,
  title={Toward economic {NMPC} for multistage {AC} optimal power flow},
  author={Faulwasser, Timm and Engelmann, Alexander},
  journal= {Optimal Control Applications and Methods},
  volume={41},
  number={1},
  pages={107--127},
  year={2020},
  publisher={Wiley Online Library}
}

@book{lee2012smooth,
  title={Introduction to Smooth Manifolds},
  author={Lee, John M},
  year={2012},
  publisher={Springer}
}

@article{itd2020review,
  title={A review on {TSO}-{DSO} coordination models and solution techniques},
  author={Givisiez, Arthur Gon{\c{c}}alves and Petrou, Kyriacos and Ochoa, Luis F},
  journal={Electric Power Systems Research },
  volume={189},
  pages={106659},
  year={2020},
  publisher={Elsevier}
}

@article{wen2023improvedDER,
  title={Improved inner approximation for aggregating power flexibility in active distribution networks and its applications},
  author={Wen, Yilin and Hu, Zechun and He, Jinhua and Guo, Yi},
  journal={IEEE Transactions on Smart Grid},
  year={2023},
  volume={early access},
  publisher={IEEE}
}

@article{diehl2002real,
  title={Real-time optimization and nonlinear model predictive control of processes governed by differential-algebraic equations},
  author={Diehl, Moritz and Bock, H Georg and Schl{\"o}der, Johannes P and Findeisen, Rolf and Nagy, Zoltan and Allg{\"o}wer, Frank},
  journal={Journal of Process Control},
  volume={12},
  number={4},
  pages={577--585},
  year={2002},
  publisher={Elsevier}
}

@article{kerscher2022key,
  title={The key role of aggregators in the energy transition under the latest European regulatory framework},
  author={Kerscher, Selina and Arboleya, Pablo},
  journal={International Journal of Electrical Power \& Energy Systems},
  volume={134},
  pages={107361},
  year={2022},
  publisher={Elsevier}
}

@article{lee2021robust,
  title={Robust {AC} optimal power flow with robust convex restriction},
  author={Lee, Dongchan and Turitsyn, Konstantin and Molzahn, Daniel K and Roald, Line A},
  journal={IEEE Transactions on Power Systems},
  volume={36},
  number={6},
  pages={4953--4966},
  year={2021},
  publisher={IEEE}
}

@article{dai2024realtime,
  title = {Real-Time Coordination of Integrated Transmission and Distribution Systems: {{Flexibility}} Modeling and Distributed {{NMPC}} Scheduling},
  author = {Dai, Xinliang and Guo, Yi and Jiang, Yuning and Jones, Colin N. and Hug, Gabriela and Hagenmeyer, Veit},
  year = {2024},
  journal = {Electric Power Systems Research},
  volume = {234},
  pages = {110627},
  doi = {10.1016/j.epsr.2024.110627}
}

@article{zhang2023coordination,
  title = {On the {{Coordination}} of {{Transmission-Distribution Grids}}: {{A Dynamic Feasible Region Method}}},
  author = {Zhang, Tiance and Wang, Jianxiao and Wang, Hao and Ruiyang, Jin and Li, Gengyin and Zhou, Ming},
  year = {2023},
  journal = {IEEE Transactions on Power Systems},
  volume = {38},
  number = {2},
  pages = {1857--1868},
  doi = {10.1109/TPWRS.2022.3197556}
}

@article{kalantar-neyestanaki2020characterizing,
  title = {Characterizing the {{Reserve Provision Capability Area}} of {{Active Distribution Networks}}: {{A Linear Robust Optimization Method}}},
  author = {{Kalantar-Neyestanaki}, Mohsen and Sossan, Fabrizio and Bozorg, Mokhtar and Cherkaoui, Rachid},
  year = {2020},
  journal = {IEEE Transactions on Smart Grid},
  volume = {11},
  number = {3},
  pages = {2464--2475},
  doi = {10.1109/TSG.2019.2956152}
}

@article{yang2017linearized,
  title={A linearized {OPF} model with reactive power and voltage magnitude: A pathway to improve the MW-only {DC} OPF},
  author={Yang, Zhifang and Zhong, Haiwang and Bose, Anjan and Zheng, Tongxin and Xia, Qing and Kang, Chongqing},
  journal={IEEE Transactions on Power Systems},
  volume={33},
  number={2},
  pages={1734--1745},
  year={2017},
  publisher={IEEE}
}

@article{jiang2023feasible,
  title = {Feasible Operation Region of an Electricity Distribution Network},
  author = {Jiang, Xun and Zhou, Yue and Ming, Wenlong and Wu, Jianzhong},
  year = {2023},
  journal = {Applied Energy},
  volume = {331},
  pages = {120419},
  doi = {10.1016/j.apenergy.2022.120419}
}

@article{fruh2023coordinated,
  title = {Coordinated {{Vertical Provision}} of {{Flexibility From Distribution Systems}}},
  author = {Fr{\"u}h, Heiner and M{\"u}ller, Sharon and Contreras, Daniel and Rudion, Krzysztof and {von Haken}, Alix and Surmann, Bartholom{\"a}us},
  year = {2023},
  journal = {IEEE Transactions on Power Systems},
  volume = {38},
  number = {2},
  pages = {1834--1844},
  doi = {10.1109/TPWRS.2022.3162041}
}

@article{contreras2019timebased,
  title = {Time-{{Based Aggregation}} of {{Flexibility}} at the {{TSO-DSO Interconnection Point}}},
  booktitle = {2019 {{IEEE Power}} \& {{Energy Society General Meeting}} ({{PESGM}})},
  author = {Contreras, Daniel A. and Rudion, Krzysztof},
  year = {2019},
  pages = {1--5},
  publisher = {IEEE},
  address = {Atlanta, GA, USA},
  doi = {10.1109/PESGM40551.2019.8973421},
  isbn = {978-1-7281-1981-6}
}

@article{lee2019convex,
  title = {Convex {{Restriction}} of {{Power Flow Feasibility Sets}}},
  author = {Lee, Dongchan and Nguyen, Hung D. and Dvijotham, Krishnamurthy and Turitsyn, Konstantin},
  year = {2019},
  month = sep,
  journal = {IEEE Transactions on Control of Network Systems},
  volume = {6},
  number = {3},
  pages = {1235--1245},
  issn = {2325-5870, 2372-2533},
  doi = {10.1109/TCNS.2019.2930896},
  urldate = {2024-11-04},
  copyright = {https://ieeexplore.ieee.org/Xplorehelp/downloads/license-information/IEEE.html}
}

@article{chen2020aggregate,
  title = {Aggregate {{Power Flexibility}} in {{Unbalanced Distribution Systems}}},
  author = {Chen, Xin and Dall'Anese, Emiliano and Zhao, Changhong and Li, Na},
  year = {2020},
  month = jan,
  journal = {IEEE Transactions on Smart Grid},
  volume = {11},
  number = {1},
  pages = {258--269},
  issn = {1949-3061},
  doi = {10.1109/TSG.2019.2920991},
  urldate = {2024-07-12}
}

@article{chen2021leveraginga,
  title = {Leveraging {{Two-Stage Adaptive Robust Optimization}} for {{Power Flexibility Aggregation}}},
  author = {Chen, Xin and Li, Na},
  year = {2021},
  month = sep,
  journal = {IEEE Transactions on Smart Grid},
  volume = {12},
  number = {5},
  pages = {3954--3965},
  issn = {1949-3061},
  doi = {10.1109/TSG.2021.3068341},
  urldate = {2024-07-12}
}

@article{tan2019enforcinga,
  title = {Enforcing {{Intra-Regional Constraints}} in {{Tie-Line Scheduling}}: {{A Projection-Based Framework}}},
  shorttitle = {Enforcing {{Intra-Regional Constraints}} in {{Tie-Line Scheduling}}},
  author = {Tan, Zhenfei and Zhong, Haiwang and Wang, Jianxiao and Xia, Qing and Kang, Chongqing},
  year = {2019},
  month = nov,
  journal = {IEEE Transactions on Power Systems},
  volume = {34},
  number = {6},
  pages = {4751--4761},
  issn = {1558-0679},
  doi = {10.1109/TPWRS.2019.2913876},
  urldate = {2025-01-12}
}

@article{tan2020estimatinga,
  title = {Estimating the {{Robust P-Q Capability}} of a {{Technical Virtual Power Plant Under Uncertainties}}},
  author = {Tan, Zhenfei and Zhong, Haiwang and Xia, Qing and Kang, Chongqing and Wang, Xuanyuan Sharon and Tang, Honghai},
  year = {2020},
  month = nov,
  journal = {IEEE Transactions on Power Systems},
  volume = {35},
  number = {6},
  pages = {4285--4296},
  issn = {1558-0679},
  doi = {10.1109/TPWRS.2020.2988069},
  urldate = {2025-04-12}
}

@article{wang2021aggregate,
  title = {Aggregate {{Flexibility}} of {{Virtual Power Plants With Temporal Coupling Constraints}}},
  author = {Wang, Siyuan and Wu, Wenchuan},
  year = {2021},
  month = nov,
  journal = {IEEE Transactions on Smart Grid},
  volume = {12},
  number = {6},
  pages = {5043--5051},
  issn = {1949-3061},
  doi = {10.1109/TSG.2021.3106646},
  urldate = {2025-01-22}
}

@article{jiang2025enhanced,
  title={Enhanced flexibility aggregation using LinDistFlow model with loss compensation},
  author={Jiang, Yanlin and Dai, Xinliang and Zahn, Frederik and Hagenmeyer, Veit},
  booktitle={2025 IEEE Kiel PowerTech},
  pages={1--6},
  year={2025},
  organization={IEEE}
}

@article{jiang2025error,
  title={Error accumulation using linearized models for aggregating flexibility in distribution systems},
  author={Jiang, Yanlin and Dai, Xinliang and Zahn, Frederik and Guo, Yi and Hagenmeyer, Veit},
  booktitle={2026 IEEE PES International Meeting (PES IM)},
  pages={1--5},
  year={2026},
  organization={IEEE}
}

@article{wei2015realtime,
  title = {Real-{{Time Dispatchability}} of {{Bulk Power Systems With Volatile Renewable Generations}}},
  author = {Wei, Wei and Liu, Feng and Mei, Shengwei},
  year = {2015},
  month = jul,
  journal = {IEEE Transactions on Sustainable Energy},
  volume = {6},
  number = {3},
  pages = {738--747},
  issn = {1949-3029, 1949-3037},
  doi = {10.1109/TSTE.2015.2413903},
  urldate = {2023-12-14}
}

@article{baran1989lindisflow,
  title={Optimal sizing of capacitors placed on a radial distribution system},
  author={Baran, Mesut E and Wu, Felix F},
  journal={IEEE Transactions on Power Delivery},
  volume={4},
  number={1},
  pages={735--743},
  year={1989},
  publisher={IEEE}
}

@article{lopez2021quickflex,
  title = {{{QuickFlex}}: A {{Fast Algorithm}} for {{Flexible Region Construction}} for the {{TSO-DSO Coordination}}},
  shorttitle = {{{QuickFlex}}},
  booktitle = {2021 {{International Conference}} on {{Smart Energy Systems}} and {{Technologies}} ({{SEST}})},
  author = {Lopez, Luis and {Gonzalez-Castellanos}, Alvaro and Pozo, David and Roozbehani, Mardavij and Dahleh, Munther},
  year = {2021},
  month = sep,
  pages = {1--6},
  publisher = {IEEE},
  address = {Vaasa, Finland},
  doi = {10.1109/SEST50973.2021.9543349},
  urldate = {2023-12-02},
  isbn = {978-1-7281-7660-4}
}

@article{contreras2021congestion,
  title = {Congestion {{Management Using Aggregated Flexibility}} at the {{TSO-DSO Interface}}},
  booktitle = {2021 {{IEEE Madrid PowerTech}}},
  author = {Contreras, Daniel A. and M{\"u}ller, Sharon and Rudion, Krzysztof},
  year = {2021},
  month = jun,
  pages = {1--6},
  doi = {10.1109/PowerTech46648.2021.9494793},
  urldate = {2025-04-13}
}

@article{churkin2023tracing,
  title = {Tracing, {{Ranking}} and {{Valuation}} of {{Aggregated DER Flexibility}} in {{Active Distribution Networks}}},
  author = {Churkin, Andrey and Kong, Wangwei and Gutierrez, Jose N. Melchor and Cese{\~n}a, Eduardo A. Mart{\'i}nez and Mancarella, Pierluigi},
  year = {2023},
  journal = {IEEE Transactions on Smart Grid},
  pages = {1--1},
  issn = {1949-3053, 1949-3061},
  doi = {10.1109/TSG.2023.3296981},
  urldate = {2023-12-03}
}

@article{zavala2010real,
  title={Real-time nonlinear optimization as a generalized equation},
  author={Zavala, Victor M and Anitescu, Mihai},
  journal={SIAM Journal on Control and Optimization},
  volume={48},
  number={8},
  pages={5444--5467},
  year={2010},
  publisher={SIAM}
}

@article{tan2024noniterativea,
  title = {Non-{{Iterative Solution}} for {{Coordinated Optimal Dispatch}} via {{Equivalent Projection}}---{{Part II}}: {{Method}} and {{Applications}}},
  author = {Tan, Zhenfei and Yan, Zheng and Zhong, Haiwang and Xia, Qing},
  year = {2024},
  journal = {IEEE Transactions on Power Systems},
  volume = {39},
  number = {1},
  pages = {899--908},
  doi = {10.1109/TPWRS.2023.3257033}
}

@article{wang2016necessary,
  title={A necessary condition for power flow insolvability in power distribution systems with distributed generators},
  author={Wang, Zhaoyu and Cui, Bai and Wang, Jianhui},
  journal={IEEE Transactions on Power Systems},
  volume={32},
  number={2},
  pages={1440--1450},
  year={2016},
  publisher={IEEE}
}

@article{aolaritei2018hierarchical,
  title={Hierarchical and distributed monitoring of voltage stability in distribution networks},
  author={Aolaritei, Liviu and Bolognani, Saverio and D{\"o}rfler, Florian},
  journal={IEEE Transactions on Power Systems},
  volume={33},
  number={6},
  pages={6705--6714},
  year={2018},
  publisher={IEEE}
}

@article{rau2003issues,
  title={Issues in the path toward an RTO and standard markets},
  author={Rau, Narayan S},
  journal={IEEE Transactions on Power Systems},
  volume={18},
  number={2},
  pages={435--443},
  year={2003},
  publisher={IEEE}
}

@book{zhu2009optimization,
  title={Optimization of power system operation},
  author={Zhu, Jizhong},
  year={2009},
  publisher={Wiley Online Library}
}

@misc{brandle2026flexibility,
  title = {On the {{Flexibility Potential}} of a {{Swiss Distribution Grid}}: {{Opportunities}} and {{Limitations}}},
  author = {Br{\"a}ndle, Jan and Rousseau, Julie and Nahata, Pulkit and Hug, Gabriela},
  year = 2026,
  number = {arXiv:2510.13449},
  eprint = {2510.13449},
  primaryclass = {eess.SY},
  publisher = {arXiv},
  doi = {10.48550/arXiv.2510.13449},
  archiveprefix = {arXiv}
}

@article{gan2015exact,
  title = {Exact {{Convex Relaxation}} of {{Optimal Power Flow}} in {{Radial Networks}}},
  author = {Gan, Lingwen and Li, Na and Topcu, Ufuk and Low, Steven H.},
  year = 2015,
  journal = {IEEE Transactions on Automatic Control},
  volume = {60},
  number = {1},
  pages = {72--87},
  doi = {10.1109/TAC.2014.2332712}
}

@article{wachter2009shorttutorial,
  title={Short tutorial: Getting started with ipopt in 90 minutes},
  author={W{\"a}chter, Andreas},
  year={2009},
  organization={Schloss Dagstuhl--Leibniz-Zentrum f{\"u}r Informatik}
}
\end{document}